\documentclass[11pt,a4paper]{article}

\usepackage{mathptmx} 
\usepackage{graphicx}

\usepackage[english]{babel}
\usepackage[latin1]{inputenc}
\usepackage{amsfonts,theorem}
\usepackage{amsmath}
\usepackage{epsfig,url}
\usepackage{bbm}

\newcommand{\myfigure}[2]{ \includegraphics*[#1]{#2} }

\newcommand{\qed}{\hfill $\Box$}

\newenvironment{proofof}[1]{%
\begin{trivlist}\item[]{\em Proof of #1}\ }{\end{trivlist}}

\newcommand{\iteg}{{\em e.g.\/}}
\newcommand{\itie}{{\em i.e.\/}}
\newcommand{\itcf}{{\em cf.\/}}
\newcommand{\tdfloor}[1]{[ #1 ]}
\newcommand{\Qlfin}{Q^{\lambda}_\Lambda}
\newcommand{\tmicro}{\tilde{t}}
\newcommand{\omla}{\omega^\lambda}
\newcommand{\Fone}{F_1^\lambda}
\newcommand{\Fzero}{F_0^\lambda}
\newcommand{\PFzero}{\Phi_0^\lambda}
\newcommand{\PFone}{\Phi_1^\lambda}
\newcommand{\Ptrunc}{\widehat{\mathcal{P}}}

\newcommand{\Z}{{\mathbb Z}}
\newcommand{\R}{{\mathbb R}}
\newcommand{\C}{{\mathbb C\hspace{0.05 ex}}}
\newcommand{\ci}{{\rm i}}
\newcommand{\re}{{\rm Re\,}}
\newcommand{\im}{{\rm Im\,}}

\newcommand{\norm}[1]{\Vert #1\Vert}
\newcommand{\defset}[2]{ \left\{ #1\left|\,
 #2\makebox[0cm]{$\displaystyle\phantom{#1}$}\right.\!\right\} }
\newcommand{\set}[1]{\{#1\}}
\newcommand{\order}[1]{{\mathcal O}(#1)}
\newcommand{\mean}[1]{\langle #1\rangle}
\newcommand{\vep}{\varepsilon}
\newcommand{\defem}[1]{{\em #1\/}}
\newcommand{\qand}{\quad\text{and}\quad}

\newcommand{\E}{{\mathbb E}}
\newcommand{\Elfin}{{\mathbb E}^{\lambda}_\Lambda}
\newcommand{\Evepfin}{{\mathbb E}^{\lambda}_\Lambda}

\newcommand{\EG}{{\mathbb E}^{0}}

\newcommand{\sabs}[1]{\langle #1\rangle}

\newcommand{\Msing}{M^{\rm sing}}
\newcommand{\Ns}{N_{\rm s}}
\newcommand{\Fbcr}{F^{\rm cr}}
\newcommand{\Qmain}{Q^{\rm main}}
\newcommand{\Qpairs}{Q^{\rm pairs}}
\newcommand{\Nnp}{N_{\rm np}}

\newcommand{\oTtree}{\mathcal{T}_{\rm o}}
\newcommand{\oTedges}{\mathcal{E}_{\mathcal{T}_{\rm o}}}
\newcommand{\Ttree}{\mathcal{T}}
\newcommand{\Tverts}{\mathcal{V}_{\mathcal{T}}}
\newcommand{\Tedges}{\mathcal{E}_{\mathcal{T}}}
\newcommand{\verts}{\mathcal{V}}
\newcommand{\Rverts}{\mathcal{V}_{\rm R}}
\newcommand{\Fverts}{\mathcal{V}_{\rm F}}
\newcommand{\Iverts}{\mathcal{V}_{\rm I}}
\newcommand{\Cverts}{\mathcal{V}_{\rm C}}
\newcommand{\Dverts}{\mathcal{V}_{\rm 0}}
\newcommand{\fedges}{\mathcal{F}}
\newcommand{\edges}{\mathcal{E}}
\newcommand{\graph}{\mathcal{G}}

\newcommand{\rootv}{v_{\rm R}}

\newcommand{\upn}[1]{^{(#1)}}
\newcommand{\taup}{\hat{\tau}}

\newcommand{\N}{{\mathbb N}}
\newcommand{\T}{{\mathbb T}}
\newcommand{\1}{{\mathbbm 1}}
\newcommand{\rme}{{\rm e}}
\newcommand{\rmd}{{\rm d}}

\newcommand{\FT}[1]{\hat{#1}}
\newcommand{\IFT}[1]{\widetilde{#1}}

\newcounter{jlisti}
\newenvironment{jlist}[1][(\thejlisti)]{\begin{list}{{\rm #1}\ \ }{ %
      \usecounter{jlisti} %
    \setlength{\itemsep}{0pt}
    \setlength{\parsep}{0pt}  %
    \setlength{\leftmargin}{0pt} %
    \setlength{\labelwidth}{0pt} %
    \setlength{\labelsep}{0pt} %
}}{\end{list}}

\newtheorem{theorem}{Theorem}[section]
\newtheorem{definition}[theorem]{Definition}
\newtheorem{corollary}[theorem]{Corollary}
\newtheorem{lemma}[theorem]{Lemma}
\newtheorem{proposition}[theorem]{Proposition}
\newtheorem{assumption}[theorem]{Assumption}

{\theorembodyfont{\upshape}
\newtheorem{remark}[theorem]{Remark}
}
\newenvironment{proof}{\begin{trivlist}\item[]{\em Proof:}\/}{\end{trivlist}}

\numberwithin{equation}{section}
\numberwithin{theorem}{section}

\begin{document}
\selectlanguage{english}

\newcommand{\email}[1]{E-mail: \tt #1}
\newcommand{\emailjani}{\email{jani.lukkarinen@helsinki.fi}}
\newcommand{\addressjani}{\em University of Helsinki,
Department of Mathematics and Statistics\\
\em P.O. Box 68,
FI-00014 Helsingin yliopisto, Finland}
\newcommand{\addressherbert}{\em Zentrum Mathematik,
Technische Universit\"at M\"unchen, \\
\em Boltzmannstr. 3, D-85747 Garching, Germany}
\newcommand{\emailherbert}{\email{spohn@ma.tum.de}}

\title{Weakly nonlinear Schr\"{o}dinger equation with\\ random initial data}
\author{Jani Lukkarinen\thanks{\emailjani},
  Herbert Spohn\thanks{\emailherbert}\\[1em]
$^*$\addressjani \\[1em] $^\dag$\addressherbert }

\maketitle

\begin{abstract}
 It is common practice to approximate a weakly nonlinear wave equation through
 a kinetic transport equation, thus raising the issue of controlling
 the validity of the kinetic limit for a suitable choice of the random initial 
data. While for the 
 general case a proof of the kinetic limit  remains open, we report on 
 first progress. As wave equation we consider the
 nonlinear Schr\"{o}dinger equation discretized on a hypercubic
 lattice. Since this is a Hamiltonian system, a natural choice of
 random initial data is distributing them according to the corresponding Gibbs
measure with a
chemical potential chosen so that the Gibbs field has exponential mixing.
 The solution $\psi_t(x)$ of the nonlinear Schr\"{o}dinger equation 
 yields then a stochastic process stationary in $x\in\mathbb{Z}^d$ and
 $t\in \mathbb{R}$. If $\lambda$ denotes the strength of the
 nonlinearity, we prove that the space-time covariance of
 $\psi_t(x)$ has a limit as $\lambda\to 0$ for $t=\lambda^{-2}\tau$,
 with $\tau$ fixed and $|\tau|$ sufficiently small.
 The limit agrees with the prediction from kinetic theory.
\end{abstract}

\tableofcontents

\section{Introduction}\label{sec:intro}

The nonlinear Schr\"{o}dinger equation (NLS) governs the evolution
of a complex valued wave field
$\psi:\mathbb{R}\times\mathbb{Z}^d\to\mathbb{C}$ and reads
\begin{equation}\label{eq:1.1}
\textrm{i}\frac{\textrm{d}}{\textrm{d}t} \psi_t(x)=\sum_{y\in
\mathbb{Z}^d}\alpha(x-y)\psi_t(y)+\lambda|\psi_t(x)|^2 \psi_t(x)\,.
\end{equation}
Here $\alpha(x)$ are the ``hopping amplitudes'' and we assume
that they satisfy
\begin{jlist}
\item $\alpha:\mathbb{Z}^d\to \mathbb{R}$, $\alpha(x)=\alpha(-x)$.
\item $|\alpha|$ has an exponentially decreasing upper bound.
\end{jlist}
We consider only the dispersive case $\lambda\geq 0$. Usually the
NLS is studied in the continuum setting, where $\mathbb{Z}^d$ is
replaced by $\mathbb{R}^d$ and the linear term is $\Delta\psi_t(x)$.
It will become evident later on why for our purposes the spatial
discretization is a necessity.

The NLS is a Hamiltonian system. To see this, we define the
canonical degrees of freedom $q_x,p_x\in \R$, $x\in\mathbb{Z}^d$,
via $\psi(x)=(q_x+\ci p_x)/\sqrt{2}$.  Their Hamiltonian function 
is obtained by substitution in
\begin{equation}\label{eq:1.2}
H(\psi)= \sum_{x,y\in\mathbb{Z}^d}\alpha(x-y)\psi(x)^\ast\psi(y)
+\tfrac{1}{2}\lambda \sum_{x\in\mathbb{Z}^d} |\psi(x)|^4\,.
\end{equation}
It is easy to check that the corresponding equations of motion,
\begin{equation}\label{eq:1.3}
\frac{\textrm{d}}{\textrm{d}t}q_x=\frac{\partial}{\partial p_x}
 H\,,\quad
\frac{\textrm{d}}{\textrm{d}t}p_x=-\frac{\partial}{\partial q_x} H\, ,
\end{equation}
are identical to the NLS. In particular, we conclude that the energy
is conserved, $H(\psi_t)=H(\psi_0)$ for all $t\in\mathbb{R}$.
Also the $\ell_2$-norm is conserved, in this context also referred
to as particle number $N$,
\begin{equation}\label{eq:1.3a}
N(\psi)= \sum_{x\in\mathbb{Z}^d}|\psi(x)|^2\,,\quad
N(\psi_t)=N(\psi_0)\textrm{ for all } t\in\mathbb{R}\,.
\end{equation}
Because of energy conservation law, if $H(\psi_0)<\infty$, then the Cauchy
problem for (\ref{eq:1.1}) has a
unique global solution. We refer to \cite{SuSu99} for a more detailed
information on the NLS.

In this work we are interested in random initial data. From a
statistical physics point of view a very natural choice is to take
the initial $\psi$-field to be distributed according to a Gibbs
measure for $H$ and $N$, which physically means that the wave field
is in thermal equilibrium. Somewhat formally the Gibbs measure is
defined through
\begin{equation}\label{eq:1.4}
\frac{1}{Z}\exp \big[-\beta\big(H(\psi)-\mu N(\psi)\big)\big]\prod_x
\left[\rmd\bigl(\re \psi(x)\bigr) \rmd\bigl(\im \psi(x)\bigr)\right] .
\end{equation}
Here $\beta>0$ is the inverse temperature and $\mu\in\mathbb{R}$ the
chemical potential.
The partition function $Z$ is a constant chosen so that
(\ref{eq:1.4}) is a probability measure. To properly define
the Gibbs measure one has to restrict (\ref{eq:1.4}) to some finite box
$\Lambda\subset\mathbb{Z}^d$, which yields a well-defined probability
measure $\mathbb{P}^\lambda_{\beta,\mu,\Lambda}$ on
$\mathbb{C}^{|\Lambda|}$. The Gibbs probability measure
$\mathbb{P}^\lambda_{\beta,\mu}$ on $\mathbb{Z}^d$ is then obtained
in the limit $\Lambda\nearrow\mathbb{Z}^d$. The existence of this
limit is a well-studied problem \cite{LePr76}. If $\lambda$ is
sufficiently small and $\mu$ sufficiently negative, then the Gibbs
measure $\mathbb{P}^\lambda_{\beta,\mu}$ exists.
The random field
$\psi(x)$, $x\in\mathbb{Z}^d$, distributed according to
$\mathbb{P}^\lambda_{\beta,\mu}$, is stationary with a rapid decay
of correlations. It is also gauge invariant in the sense that
$\psi(x)=\mathrm{e}^{i\theta}\psi(x)$ in distribution for 
any $\theta\in [0,2\pi]$.

Of course, $\mathbb{P}^\lambda_{\beta,\mu}$-almost surely it holds
$H(\psi)=\infty$ and $N(\psi)=\infty$. Thus one has to define
solutions for the NLS with initial data of infinite energy. This has
been accomplished for standard anharmonic Hamiltonian systems by
Lanford, Lebowitz, and Lieb \cite{lll77}, who prove existence and
uniqueness  under a suitable growth condition at infinity for the
initial data.  These arguments extend to the Hamiltonian system (\ref{eq:1.3}).
It remains to prove that the so-defined infinite volume dynamics is well
approximated by the  
finite volume dynamics with periodic boundary conditions. Very likely 
such a result can be achieved using the methods developed in
\cite{butta07}. For our purposes 
it is more convenient to circumvent the issue by proving estimates which are
uniform in the volume.

Let us briefly comment why the more conventional continuum NLS, 
\begin{equation}
\ci \frac{\partial }{\partial t}u(x,t) = -\Delta_x u(x,t) + \lambda 
\int_{\R^d} \!\rmd y\,
|u(y,t)|^2 V(y-x)u(x,t),\quad x\in \R^d,\ t\in \R\, ,
\end{equation}
poses additional difficulties. The Gibbs measure at finite volume is a
perturbed Gaussian measure 
which is singular at short distances. Thus the construction of the dynamics
requires an effort. 
Furthermore,  the limit $V(x) \to \delta(x)$ is a fundamental problem 
of constructive quantum field theory and 
is known to be difficult \cite{glimmjaffe}. To establish the existence
of the dynamics 
for such singular initial data has not even been attempted.

In the present context the most basic quantity is the stationary
covariance
\begin{equation}\label{eq:1.5}
\mathbb{E}^\lambda_{\beta,\mu}\big(\psi_0(x_0)^\ast\psi_t(x)\big)=
F^\lambda_2 (x-x_0,t)\, ,
\end{equation}
where $\mathbb{E}^\lambda_{\beta,\mu}$ denotes expectation with respect to 
$\mathbb{P}^\lambda_{\beta,\mu}$. 
The existence of such a function $F^\lambda_2$ follows from the translation
invariance of the measure, and one would like to know its
qualitative dependence on $x,t$. For deterministic infinitely extended
Hamiltonian systems, such as the NLS, establishing the qualitative 
behavior of equilibrium time correlations is known to be an extremely
difficult problem with very few results available, despite intense efforts.
For linear systems one has an explicit solution in Fourier space, see below. 
But already for
completely integrable systems, like the Toda chain, not much is
known about time correlations in thermal equilibrium.

It is instructive first to discuss the linear case, $\lambda=0$,
for which purpose we introduce Fourier transforms. For
$f:\mathbb{Z}^d\to\mathbb{C}$ let us denote its Fourier transform by
\begin{equation}\label{eq:1.6}
\hat{f}(k)=\sum_{x\in\mathbb{Z}^d} f(x) \textrm{e}^{-\textrm{i} 2\pi
k\cdot x}\,,
\end{equation}
$k\in\mathbb{R}$, and the inverse Fourier transform by
\begin{equation}\label{eq:1.7}
\tilde{g}(x)=\int_{\mathbb{T}^d} \textrm{d}k\, g(k)
\textrm{e}^{\textrm{i} 2\pi k\cdot x}
\end{equation}
with $\mathbb{T}^d=[0,1]^d$, a parametrization of the $d$-dimensional torus. 
(We will use arithmetic relations on $\T^d$.  These are defined using 
the arithmetic induced on the torus via its definition as equivalence classes
$\R^d/\Z^d$,
\itie , by using ``periodic boundary conditions''.)
In particular, we set
\begin{equation}\label{eq:1.8}
\omega(k)=\hat{\alpha}(k),\quad k\in \T^d\, .
\end{equation}
The function $\omega$ is the dispersion relation 
of our discretized linear Schr\"{o}\-din\-ger equation. 
It follows from the assumptions on $\alpha$ that
\begin{jlist}
\item $\omega:\mathbb{T}^d\to\mathbb{R}$ and its periodic extension is
a real analytic function.
\item $\omega(k)=\omega(-k)$.
\end{jlist}
In Fourier space the energy is given by
\begin{align}\label{eq:1.9}
& H(\psi)=  \int_{\mathbb{T}^d} \textrm{d}k\,
\omega(k)|\hat{\psi}(k)|^2
\nonumber\\ & \quad
+ \tfrac{1}{2}\lambda
 \int_{(\mathbb{T}^{d})^4} \textrm{d}k_1 \textrm{d}k_2 \textrm{d}k_3
\textrm{d}k_4
 \delta(k_1+k_2-k_3-k_4) \hat{\psi}(k_1)^\ast \hat{\psi}(k_2)^\ast
\hat{\psi}(k_3)\hat{\psi}(k_4)\, ,
\end{align}
where $\delta$ is a formal Dirac $\delta$-function, 
used here to simplify 
the notation for the convolution integral.
Clearly, $H(\psi)\geq (\inf_k \omega(k)) N(\psi)$. The NLS after Fourier
transform reads
\begin{align}\label{eq:1.10}
&
\frac{\textrm{d}}{\textrm{d}t}\hat{\psi}_t(k_1)=
-\textrm{i} \omega(k_1) \hat{\psi}_t(k_1) -\textrm{i} \lambda \int
\textrm{d}k_2 \textrm{d}k_3 \textrm{d}k_4
 \delta(k_1+k_2-k_3-k_4)
\nonumber\\ & \qquad\times
\hat{\psi}_t(k_2)^\ast\hat{\psi}_t(k_3)\hat{\psi}_t(k_4)\,.
\end{align}

For $\lambda=0$, $\mathbb{P}^0_{\beta,\mu}$ is a Gaussian measure
with mean zero and covariance
\begin{equation}\label{eq:1.11}
\mathbb{E}^0_{\beta,\mu} \big( \psi(0)^\ast \psi(x)\big)= F^0_2
(x,0)=\int_{\mathbb{T}^d} \textrm{d}k
\big(\beta(\omega(k)-\mu)\big)^{-1} \textrm{e}^{\textrm{i} 2\pi
k\cdot x}\,,
\end{equation}
provided $\mu<\inf_k \omega(k)$. Under our assumptions on
$\omega$ the Gaussian field has exponential mixing. For the
time-dependent equilibrium covariance one obtains
\begin{equation}\label{eq:1.12}
F^0_2(x,t)=\int_{\mathbb{T}^d} \textrm{d}k
\big(\beta(\omega(k)-\mu)\big)^{-1} \textrm{e}^{\textrm{i} 2\pi
k\cdot x} \textrm{e}^{-\textrm{i} \omega(k)t}\,.
\end{equation}
Clearly, $F^0_2(x,t)$ is a solution of the linear wave
equation for exponentially localized initial data and thus spreads
dispersively.

If $\lambda> 0$, as general heuristics the nonlinearity 
should induce an exponential damping of $F^\lambda_2$. The physical
picture is based on excitations of wave modes which interact weakly
and are damped through collisions. Approximate theories have been
developed in the context of phonon physics and wave turbulence, see
\iteg ~\cite{Gu86,ZLF92}. To mathematically establish such a time-decay is
completely out of reach, at present, whatever the choice of the
nonlinear wave equation.

To make some progress we will investigate here the regime of small
nonlinearity, $\lambda\ll 1$. The idea is not to aim for results
which are valid globally in time, but rather to consider the first
time scale on which the effect of the nonlinearity becomes visible.
For small $\lambda$ the rate of collision for two resonant waves is
of order $\lambda^2$. Therefore, the nonlinearity is expected to show
up on a time scale $\lambda^{-2}$. This suggests to study the limit
\begin{equation}\label{eq:1.13}
F^\lambda_2 (x,\lambda^{-2} t)\, ,\quad \textrm{as }\lambda\to 0\,.
\end{equation}
Note that the location $x$ is not scaled. For this limit to exist,
one has to remove the oscillating phase
resulting from (\ref{eq:1.12}), which on the speeded-up time scale is
rapidly oscillating, of order $\lambda^{-2}$.   In fact, a second rapidly
oscillating phase of order $\lambda^{-1}$ will show up, which also has
to be removed. Under suitable conditions on $\omega$, we will prove
that $F^\lambda_2 (x,\lambda^{-2} t)$, with the removals just
mentioned, has a limit for $\lambda\to 0$, at least for $|t|\leq
t_0$ with some suitable $t_0>0$. The limit function indeed exhibits
exponential damping.

A similar result has been obtained a long time ago for a system of
hard spheres in equilibrium and at low, but fixed, density
\cite{BLLS80}. There the small parameter is the density rather than the
strength $\lambda$ of the nonlinearity. But the over-all philosophy is the
same. To establish the decay of time-correlations in equilibrium at
a fixed low density is an apparently very hard problem. Therefore,
one looks for the first time scale on which the collisions between
hard spheres have a visible effect. By fiat, hard spheres remain
well localized in space, and on the time scale of interest only a
finite number of collisions per particle are taken into account. In 
contrast, waves tend to delocalize through collisions. This is
the reason why the problem under study has remained open. Our resolution
uses techniques totally different from \cite{BLLS80}.

The limit $\lambda\to 0$, $t=\lambda^{-2}\tau$ with $\tau$ fixed, together
with a possible rescaling of space by a factor $\lambda^{-2}$, is called
\textit{kinetic limit}, because the limit object is governed by a kinetic type
transport equation. Formal derivations are discussed extensively in the
literature, \iteg , see \cite{janssen03,LN04}. 
On the mathematical side, Erd{\H o}s and Yau \cite{erdyau99} 
study in great detail the linear 
Schr\"{o}dinger equation with a random potential, extended to
even longer time scales in \cite{erdyau05b,erdyau05a}.
The discretized wave equation with a random
index of refraction is covered in \cite{ls05}. For nonlinear wave equations
the only related study is by Benedetto \textit{et al.\/} \cite{BCEP08} on the
dynamics of weakly interacting quantum particles. They transform to multipoint
Wigner functions, 
which leads to an expansion somewhat different from the one used here. We
refer to \cite{ls09} 
for a comparison. As in our contribution, Benedetto \textit{et al}. have to
analyze the asymptotics 
of high-dimensional oscillatory integrals. But in contrast, they
have no control on the error term in the expansion.

Before closing the introduction, we owe the reader some explanations
why a seemingly perturbative result requires so many pages for its
proof. From the solution to (\ref{eq:1.1}) one can regard
$\psi_t(x)$ as some functional $\mathcal{F}_{x,t}$ of the initial
field $\psi$,
\begin{equation}\label{eq:1.14}
\psi_t(x)=\mathcal{F}_{x,t}(\psi)\,.
\end{equation}
For given $t$ it depends only very little on those $\psi(y)$'s for
which $|y-x|\gg t$. To make progress it seems necessary to first
average the initial conditions over
$\psi(x_0)^\ast\mathbb{P}^\lambda_{\beta,\mu}$ so that subsequently
one can control the limit $\lambda\to 0$ with $t=\lambda^{-2}\tau$,
$\tau>0$. Such an average can be accomplished by writing
$\mathcal{F}_{x,t}$ as a power series in $\psi$, which is done
through the Duhamel formula. For any $n\geq 1$ we write
\begin{align}\label{eq:1.15}
 & \prod^n_{j=1} \textrm{e}^{\textrm{i} 2\pi \sigma_j
\omega(k_j)t} \hat{\psi}_{t}(k_j,\sigma_j )
 = \prod^n_{j=1}\hat{\psi}_{0}(k_j,\sigma_j )
 +\int^t_0 \textrm{d}s
\frac{\textrm{d}}{\textrm{d}s}\prod^n_{j=1} \textrm{e}^{\textrm{i}
2\pi \sigma_j  \omega(k_j)s}\hat{\psi}_{s}(k_j,\sigma_j)\,.
\end{align}
Here $\sigma_j \in \set{\pm 1}$ and $\hat{\psi}_t(k,1)=\hat{\psi}_t(k)$,
$\hat{\psi}_{t}(k,-1)=\hat{\psi}_t(-k)^\ast$,
$\hat{\psi}_0(k)=\hat{\psi}(k)$. Using the product rule and the
equations of motion (\ref{eq:1.10}) yields a formula relating the
$n$:th moment at time $t$ to the time-integral of a sum over
$(n+2)$:th moments at time $s$. Iterating this equation leads to a (formal)
series representation
\begin{equation}\label{eq:1.16}
\hat{\psi}_t(k)=\sum^\infty_{n=1} \mathcal{P}^n_{k,t}
(\hat{\psi})\,,
\end{equation}
where $\mathcal{P}^n_{k,t}$ is a sum/integral over monomials of
order $n$ in $\hat{\psi}$ and $\hat{\psi}^\ast$. Since each
time-derivative increases the degree of the monomial by two, we
have
\begin{equation}\label{eq:1.17}
\delta(k'-k)\sum_{x\in\mathbb{Z}^d} \textrm{e}^{-\textrm{i}2\pi
k\cdot x} \mathbb{E}^\lambda_{\beta,\mu}\big(\psi(0)^\ast
\psi_t(x)\big)=\sum^\infty_{n=0}
\mathbb{E}^\lambda_{\beta,\mu}\big(\hat{\psi}(k')^\ast
\mathcal{P}^{2n+1}_{k,t}(\hat{\psi})\big)\,.
\end{equation}

The first difficulty arises from the fact that the sum in
(\ref{eq:1.17}) does not converge absolutely for any $t$. Very roughly,
$\mathcal{P}^n_{k,t}$ is a sum of $n!$ terms of equal size. The
iterated time-integration yields a factor $t^n/n!$. However, for
the approximately Gaussian average the $n$:th moment grows also as
$n!$. To be able to proceed one has to stop the series expansion at some large
$N$ which depends on $\lambda$. A similar situation was encountered by
Erd\H{o}s and Yau \cite{erdyau99} in their study of the Schr\"{o}dinger
equation with a weak random potential. We will use the powerful  
Erd\H{o}s-Yau techniques as a guideline for handling the series in
(\ref{eq:1.17}). 

The stopping of the series expansion will leave a remainder term containing 
the full original time-evolution.
Erd\H{o}s and Yau control the error term in essence by unitarity of the
time-evolution. For the NLS mere conservation of $N(\psi)$ will not suffice. 
Instead, we use stationarity of $\psi_t(x)$. In
wave turbulence \cite{ZLF92} one is also interested in non-stationary
initial measures, \iteg , in Gaussian measures with a covariance
different from $(\beta(\omega(k)-\mu))^{-1}$. For such initial data
we have no idea how to control the error term, while other parts of
our proof apply unaltered.

The central difficulty resides in
$\mathbb{E}^\lambda_{\beta,\mu}\big(\hat{\psi}(k')^\ast
\mathcal{P}^{2n+1}_{k,t}(\hat{\psi})\big)$ which is a sum of rather
explicit, but high-dimensional, dimension $n(1+3d)+d$, oscillatory
integrals. On top, because of the $\delta$-function in
(\ref{eq:1.10}), the integrand is restricted to a non-trivial linear
subspace. In the limit $\lambda\to 0$, $t=\lambda^{-2}\tau$,
$\tau>0$, only a few oscillatory integrals have a non-zero limit.
Summing up these leading oscillatory integrals results in the anticipated 
exponential damping. The major task of our paper is to discover
an iterative structure in all remaining oscillatory integrals, 
in a way which allows for an estimate in terms of a few basic ``motives''. 
Each of these subleading integrals is shown to contain at least one motive
whose appearance leads to
an extra fractional power of $\lambda$, thereby ensuring a zero limit.

In Section \ref{sec:finite volume}  we first give the
mathematical definition of the above system in finite volume, and
state in Section \ref{sec:model} the assumptions and main results. 
Their connection to kinetic theory is discussed in Section \ref{sec:link}.
The proof of the main result
is contained in the remaining sections: we derive a suitable
time-dependent perturbation expansion in Section \ref{sec:graphs},
and develop a graphical language to describe the large, but finite,
number of terms in the expansion in Section \ref{sec:diagrams}.  The
analysis of the oscillatory integrals in the expansion is contained
in Sections \ref{sec:momdeltas}--\ref{sec:fullypaired}. More detailed outline
of the technical structure of the proof can be found in Section
\ref{sec:structureofproof}. The estimates are collected together and
the limit of the non-zero terms is computed in Section
\ref{sec:completion} where we complete the proof of the main
theorem. In an Appendix, we show that the standard nearest neighbor
couplings in $d\ge 4$ dimensions lead to dispersion relations
satisfying all assumptions of the main theorem.

\bigskip

\noindent
{\it Acknowledgments.} 
We would like to thank L\'{a}szl\'{o} Erd\H{o}s and
Horng-Tzer Yau for many illuminating discussions on the subject.
The research of J.\ Lukkarinen was supported by the Academy of Finland.

\section{Kinetic limit and main results}\label{sec:main}

\subsection{Finite volume dynamics}\label{sec:finite volume}

To properly define expectations such as (\ref{eq:1.5}), one has to go through a
finite volume construction, which will be specified in this
subsection.

Let
\begin{equation}\label{eq:1.1.1}
  L\ge 2\,,\quad \Lambda = \{0,1,\ldots,L-1\}^d\,,
\end{equation}
the dimension $d$ an arbitrary positive integer. We apply periodic
boundary conditions on $\Lambda$, and let $[x]=x \bmod L\in \Lambda$
for all $x\in \mathbb{Z}^d$.  Fourier transform of
$f:\Lambda\to\mathbb{C}$ is denoted by $\hat{f} : \Lambda^* \to
\mathbb{C}$, with the dual lattice $\Lambda^* =
\{0,\frac{1}{L},\ldots,\frac{L-1}{L}\}^d$ and with
\begin{equation}
   \hat{f}(k) = \sum_{x\in \Lambda} f(x) \textrm{e}^{-\textrm{i} 2\pi k \cdot x}
   \end{equation} for all
   $k\in \Lambda^*$ (or for all $k\in (\mathbb{Z}/L)^d$, which yields the
periodic extension of $\hat{f}$).   The inverse transform is given
by
\begin{equation} 
   \tilde{g}(x) = \frac{1}{|\Lambda|} \sum_{k\in \Lambda^*}  g(k)
     \textrm{e}^{\textrm{i} 2\pi k \cdot x}\,,
\end{equation}
where $|\Lambda|=L^d$. For all $x\in \Lambda$, it holds
$\tilde{{\hat{f}}}(x)=f(x)$. The arithmetic operations on $\Lambda$
are done periodically, identifying it as a parametrization of
$\mathbb{Z}_L^d$, the cyclic group of $L$ elements (for instance, for
$x,y\in \Lambda$, we have then $x+y=[x+y]$ and $-x=[-x]$.)
Similarly, $\Lambda^*$
is identified as a subset of the $d$-torus $\mathbb{T}^d$.

We will use the short-hand notations 
\begin{equation} 
  \int_{\Lambda^*} \textrm{d} k\, \cdots =
   \frac{1}{|\Lambda|} \sum_{k\in \Lambda^*} \cdots \, ,
\end{equation}
and
\begin{equation}\label{eq:2.5}
\langle f,\psi\rangle=\sum_{x\in\Lambda} f(x)^\ast \psi(x)\, ,
\end{equation}
as well as the similar but unrelated notation for ``regularized'' absolute
values
\begin{align}
 \sabs{x} =\sqrt{1+x^2}, \qquad \text{for all } x\in \R\, .
\end{align}
Let us also denote the limit $L\to \infty$ by $\Lambda\to\infty$.
Let $\omega:\mathbb{T}^d\to\mathbb{R}$ be defined as in (\ref{eq:1.8}).
For the finite volume, we introduce the periodized $\alpha_\Lambda$
through
 \begin{equation} 
  \alpha_\Lambda(x) = \int_{\Lambda^*}\! \textrm{d} k \, \textrm{e}^{\textrm{i}
2\pi x \cdot k}
  \omega(k) = \frac{1}{|\Lambda|} \sum_{k\in \Lambda^*}
  \textrm{e}^{\textrm{i} 2\pi x\cdot k} \omega(k).
\end{equation}
Clearly, $\alpha_\Lambda\in \mathbb{R}$ and
$\alpha_\Lambda(-x)=\alpha_\Lambda(x)$ for all $x\in \Lambda$.

After these preparations, we define the finite volume Hamiltonian for
$\psi:\Lambda \to \mathbb{C}$ by
\begin{align}\label{eq:HfinV}
&  H_\Lambda(\psi) = \sum_{x,y\in \Lambda}
\alpha_\Lambda(x-y)
  \psi(x)^* \psi(y)  +  \tfrac{1}{2} \lambda \sum_{x\in\Lambda} |\psi(x)|^4
\nonumber \\ & \quad
   = \int_{\Lambda^*} \textrm{d} k \, \omega(k) |\hat{\psi}(k)|^2
\nonumber \\ & \qquad
   + \tfrac{1 }{2}\lambda
   \int_{(\Lambda^*)^4}\!\!\! \textrm{d} k_1 \textrm{d} k_2 \textrm{d}
k_3\textrm{d} k_4\,
  \delta_{\Lambda}(k_1+k_2-k_3-k_4)
   \hat{\psi}(k_1)^*\hat{\psi}(k_2)^*\hat{\psi}(k_3)
   \hat{\psi}(k_4)\,,
\end{align}
where $\lambda\ge 0$ and $\delta_{\Lambda}: (\mathbb{Z}/L)^d \to
\mathbb{R}$ is the following discrete $\delta$-function
\begin{equation} 
  \delta_{\Lambda}(k) = |\Lambda| \mathbbm{1}(k \bmod 1 = 0) .
\end{equation}
Here $\1$ denotes a generic characteristic function: $\1(P)=1$, if
the condition $P$ is true, and $\1(P)=0$ otherwise.
$H_\Lambda(\psi)\ge c\|\psi\|^2_2$ for all $\psi$, with
$c=\inf_k\omega(k)>-\infty$ and $\|\psi\|_2$ denoting the
$\ell_2(\Lambda)$-norm.

Introducing, as before, the canonical conjugate pair $q_x,p_x\in\R$ through
$\psi(x)=(q_x+\ci p_x)/\sqrt{2}$, and then applying 
the evolution equations associated to
$H_\Lambda$,
we find that $\psi_t(x)$ satisfies the finite volume discrete NLS
\begin{equation}\label{eq:dNLS}  
\textrm{i} \frac{\textrm{d} }{\textrm{d} t} \psi_t(x) =
  \sum_{y\in \Lambda} \alpha_\Lambda(x-y) \psi_t(y)
+ \lambda |\psi_t(x)|^2 \psi_t(x) \, .
\end{equation}
The Fourier-transform $\FT{\psi}_t(k)$ satisfies the evolution equation
\begin{align}\label{eq:FTdNLS2}
& \frac{\rmd }{\rmd t} \FT{\psi}_t(k_1) = -\ci \omega(k_1) \FT{\psi}_t(k_1)
\nonumber \\ & \quad
-\ci \lambda \int_{(\Lambda^*)^3}\!\! \rmd k_2 \rmd k_3\rmd k_4\,
 \delta_{\Lambda}(k_1+k_2-k_3-k_4)
  \FT{\psi}_t(k_2)^* \FT{\psi}_t(k_3) \FT{\psi}_t(k_4) .
\end{align}

The evolution equations have a continuously differentiable
solution for all $t\in\R$ and for any given initial conditions
$\psi_0\in \C^{\Lambda}$, which follows by a standard
fixed point argument and the conservation laws stated below.
The energy $H_\Lambda(\psi)$ is naturally
conserved by the time-evolution.  In addition, for all $x$,
\begin{align}\label{eq:l2deriv}
 \frac{\rmd }{\rmd t} |\psi_t(x)|^2 = -\ci
\sum_{y\in \Lambda} \alpha_\Lambda(x-y) \left(\psi_t(x)^* \psi_t(y)-
 \psi_t(y)^* \psi_t(x)\right) .
\end{align}
The right hand side sums to zero if we sum over all $x\in\Lambda$.
Therefore, for $t\in\R$,
\begin{align}\label{eq:normisconstant}
\norm{\psi_t}_2^2=\sum_{x\in \Lambda} |\psi_t(x)|^2 = \norm{\psi_0}_2^2\, ,
\end{align}
and thus also $\norm{\psi_t}_2$ is a constant of motion.

The initial field is taken to be distributed according to the finite
volume Gibbs measure as explained in the introduction.  We assume
that its parameters are fixed to some values satisfying 
$\beta>0$ and $\mu<\inf_k\omega(k)$,
and we drop the dependence on these parameters from the notation. 
Then the Gibbs measure is
\begin{align}\label{eq:LGibbs}
& \int_{\C^\Lambda} \mathbb{P}_{\Lambda}^\lambda(\rmd \psi) f(\psi)
= \frac{1}{Z_{\Lambda}^\lambda}
\int_{(\R^2)^{\Lambda}} \prod_{x\in \Lambda}
\left[\rmd(\re \psi(x))\, \rmd(\im \psi(x))\right]
  \rme^{-\beta (H_{\Lambda}(\psi)-\mu \norm{\psi}^2)} f(\psi)
 \, .
\end{align}
Expectation values with respect to the
finite volume, perturbed measure $\mathbb{P}_{\Lambda}^\lambda$ are
denoted by $\Elfin$.  Taking the limits $\Lambda\to\infty$ and
$\lambda\to 0$ leads to a Gaussian measure.  It is 
defined via its covariance function which
has a Fourier transform
\begin{align} 
W(k) = \FT{F}_2^0(k,0)= \frac{1}{\beta(\omega(k)-\mu)}\, .
\end{align}
We denote expectations over this Gaussian measure by $\EG$.
Note that by the translation invariance of the finite volume Gibbs
measure, there always exists a function
$W^{\lambda}_\Lambda:\Lambda^*\to \C$ such that for all $k,k'\in
\Lambda^*$,
\begin{align} 
  \Elfin[\FT{\psi}(k)^* \FT{\psi}(k')] = \delta_\Lambda(k-k')
  W^{\lambda}_\Lambda(k)\, .
\end{align}

Since the energy and norm are conserved, the Gibbs measure is time
stationary.  In other words, for all $t$ and any integrable $f$
\begin{align} 
\Elfin[f(\psi_t)] = \Elfin[f(\psi_0)] .
\end{align}
In addition, since the dynamics and the Gibbs measure are invariant under
periodic translations of $\Lambda$, under $\mathbb{P}^\lambda_\Lambda$ the
stochastic process $(x,t)\mapsto \psi_t(x)$ is stationary jointly in
space and time.

\subsection{Main results}\label{sec:model}

We have to impose two types of assumptions. Those in Assumption 
\ref{th:disprelass}
are conditions on the dispersion relation $\omega$. 
Assumption \ref{th:Ainitcond}
is concerned with a specific form of the clustering of the Gibbs
measure. In each case we comment on their current status.

\begin{assumption}[Equilibrium correlations]\label{th:Ainitcond}
Let $\beta>0$ and  $\mu < \inf_k\omega(k)$ be given.
We take the initial conditions $\psi_0$ to be distributed according
to the Gibbs measure $\mathbb{P}_{\Lambda}^\lambda$ which is  assumed to be
{\em $\ell_1$-clustering\/} in the following sense:
We assume that
there exists $\lambda_0 > 0$ and $c_0 > 0$, independent of $n$,
such that for  $0 <\lambda\le \lambda_0$ and all $n\ge  4$ one has the
following bound for the fully truncated correlation functions (\itie , cumulants)
\begin{align}\label{eq:l1clustering}
 \sup_{\Lambda,\sigma\in \set{\pm 1}^n} \sum_{x\in \Lambda^{n}}
\1(x_1=0)
  \Bigl|\Elfin\Bigl[\prod_{i=1}^n \psi(x_i,\sigma_i)\Bigr]^{\rm trunc}\Bigr|
 \le \lambda (c_0)^n n!\, ,
\end{align}
where $\psi(x,1)=\psi(x)$, $\psi(x,-1)=\psi(x)^*$.
We also assume a comparable convergence of
the two-point correlation functions for $0 <\lambda\le \lambda_0$,
\begin{equation}\label{eq:twopointlim}
 \limsup_{\Lambda\to \infty} \sum_{\norm{x}_\infty\le L/2} \left|
  \Elfin[\psi(0)^* \psi(x)] - \E^0[\psi(0)^* \psi(x)]\right| \le \lambda
2 (c_0)^2\,.
\end{equation}
\end{assumption}
In the present proof, valid for $d\ge 4$, we do not use the full strength of
the  bound in (\ref{eq:l1clustering}), namely, we could omit 
the pre\-factor $\lambda$. 
However, the pre\-factor could be needed in any proof which concerns
$d\le 3$.
In contrast, we do make use of the pre\-factor in (\ref{eq:twopointlim}).
The second condition can equivalently be recast in terms of $W$ as
\begin{equation}
 \limsup_{\Lambda\to \infty} \sum_{\norm{x}_\infty\le L/2} \left|
   \IFT{W}^{\lambda}_\Lambda(x)-\IFT{W}(x)\right| \le \lambda
2 (c_0)^2\,.
\end{equation}

Technically, Assumption \ref{th:Ainitcond} refers to the clustering of a weakly
coupled massive two-com\-po\-nent $\lambda\phi^4$-theory. Such problems
have a long tradition in equilibrium statistical mechanics and are
handled through cluster expansions, \iteg , see \cite{MM91,Salm08}.
The difficulty with Assumption \ref{th:Ainitcond} resides in the precise $n$-
and $\lambda$-dependence of the bounds. Motivated by our work, 
the issue was reinvestigated for 
the equilibrium measure (\ref{eq:LGibbs}) 
in the contribution of Abdesselam, Procacci, and
Scoppola  \cite{abdesselam-2009}, in which 
they prove Assumption \ref{th:Ainitcond} 
for hopping amplitudes of finite range and with zero
boundary conditions, \itie , setting $\psi(x) = 0$ for $x \notin\Lambda$.
The authors ensure us that their results remain valid also for periodic
boundary conditions, thereby establishing Assumption  \ref{th:Ainitcond} for a
large class of hopping amplitudes.

For the main theorem we will need properties of the linear dynamics,
$\lambda=0$, which can be thought of as implicit conditions on
$\omega$.
\begin{assumption}[Dispersion relation]\label{th:disprelass}
Suppose $d\ge 4$, and $\omega:\T^d\to\R$ satisfies all of the following:
\begin{jlist}[(DR\thejlisti)]
\item\label{it:DR1} The periodic extension of $\omega$ is real-analytic
and $\omega(-k)=\omega(k)$.
\item\label{it:DRdisp} ($\ell_3$-dispersivity). Let us consider the 
{\em free propagator\/}
  \begin{align}\label{eq:defptx}
    p_t(x) = \int_{\T^d} \!\rmd k\, \rme^{\ci 2\pi x\cdot k}
    \rme^{-\ci t \omega(k)} \, .
  \end{align}
We assume that there are $C,\delta>0$ such that for all $t\in\R$,
\begin{align} 
  \norm{p_t}_3^3 = \sum_{x\in\Z^d} |p_t(x)|^3 \le C \sabs{t}^{-1-\delta} \, .
\end{align}
\item\label{it:DRinterf} (constructive interference).
There exists a set $\Msing \subset \T^d$ consisting
of a union of a finite number of closed, one-dimensional, smooth submanifolds,
and a constant $C$ such that
for all $t\in \R$, $k_0\in \T^d$, and $\sigma\in \set{\pm 1}$,
\begin{align} 
  \Bigl| \int_{\T^d}\!\rmd k\, \rme^{-\ci t (\omega(k)+\sigma
    \omega(k-k_0))}\Bigr| \le \frac{C\sabs{t}^{-1} }{d(k_0,\Msing)}\, ,
\end{align}
where $d(k_0,\Msing)$ is the distance (with respect to the standard metric
on the $d$-torus, $\R^d/\Z^d$) of $k_0$ from $\Msing$.
\item\label{it:DRcrossing} (crossing bounds).
Define for $t_0,t_1,t_2\in \R$, $u_1,u_2\in \T^d$, and $x\in \Z^d$,
  \begin{align}\label{eq:defp2tx}
    K(x;t_0,t_1,t_2,u_1,u_2) =
    \int_{\T^d} \!\rmd k\, \rme^{\ci 2\pi x\cdot k}
    \rme^{-\ci (t_0 \omega(k)+t_1 \omega(k+u_1)+ t_2 \omega(k+u_2))} \, .
  \end{align}
We assume that there is
a measurable function $\Fbcr:\T^d \times \R_+ \to [0,\infty]$
so that constants  $0<\gamma\le 1$, $c_1,c_2$,
for the following bounds can be found.
\begin{enumerate}
\item  For any $u_i\in \T^d$, $\sigma_i\in \set{\pm 1}$, $i=1,2,3$,
and  $0<\beta\le 1$, the following bounds are satisfied:
\begin{align}
&
 \int_{-\infty}^\infty\! \rmd t\, \norm{p_{t}}_3^2
 \int_{-\infty}^\infty\! \rmd s\, \rme^{-\beta |s|}
 \norm{K(t,\sigma_1 s,\sigma_2 s,u_1,u_2)}_3
 \le \beta^{\gamma-1}  \Fbcr(u_2-u_1;\beta) \, ,
\label{eq:crossingest1c}
\\ &
  \int_{-\infty}^\infty\! \rmd t
 \int_{-\infty}^\infty\! \rmd s\, \rme^{-\beta |s|}
 \prod_{i=1}^3 \norm{K(t,\sigma_i s,0,u_i,0)}_3
 \nonumber \\ & \quad
\le \beta^{\gamma-1}  \Fbcr(u_n;\beta),\quad \text{for any }n\in\set{1,2,3}\, .
\label{eq:crossingest1b}
\end{align}
\item  For all   $0<\beta\le 1$ we have
\begin{align}\label{eq:crossingest2a}
&  \int_{\T^d} \rmd k\, \Fbcr(k;\beta) \le c_1 \sabs{\ln \beta}^{c_2}  ,
\end{align}
and if also $u,k_0\in \T^d$,
  $\alpha\in \R$,
  $\sigma\in\set{\pm 1}$, and $n\in \set{1,2,3}$,
and we denote  $k=(k_1,k_2,k_0-k_1-k_2)$, then
\begin{align}\label{eq:crossingest2}
&  \int_{(\T^d)^2} \rmd k_1\rmd k_2\, \Fbcr(k_n+u;\beta)
\frac{1}{|\alpha-\Omega(k,\sigma)+\ci\beta|}
 \le c_1 \sabs{\ln \beta}^{1+c_2}  ,
\end{align}
where $\Omega:(\T^d)^3\times \set{\pm 1}\to \R$ is defined by
\begin{align}\label{eq:defOmega}
& \Omega(k,\sigma)
= \omega(k_3)-\omega(k_1)+\sigma(\omega(k_2)-\omega(k_1+k_2+k_3))\, .
\end{align}
\end{enumerate}
\end{jlist}
\end{assumption}

\begin{remark}
We prove in Appendix \ref{sec:appNN} that the nearest neighbor
interactions satisfy all of the above assumptions for $d\ge 4$, 
if we use $\gamma=\frac{4}{7}$, $c_2=0$, and the
function
\begin{align}\label{eq:defnnFcr}
 \Fbcr(u;\beta) = C
 \prod_{\nu=1}^d \frac{1}{|\sin (2 \pi u^\nu)|^{\frac{1}{7}}}
\end{align}
with a certain constant $C$ depending only on $d$ and $\omega$.
Presumably a larger class of $\omega$'s could be covered, but this
needs a separate investigation.

The estimates in Appendix \ref{sec:appNN} in fact imply that also
for $d=3$ the dispersion relation of the
nearest neighbor interactions satisfies assumptions 
DR\ref{it:DR1}--DR\ref{it:DRinterf}.  However, 
even if also DR\ref{it:DRcrossing} could be checked, this would 
not be sufficient to generalize the result to $d=3$ since $d\ge 4$ is
used to facilitate the analysis of
constructive interference effects in Sec.~\ref{sec:proveerrcut}.  The
present estimates require that the co-dimension of the bad set is at least
three which for $d=3$ would allow only a finite collection of bad points.
As we have no examples of such dispersion relations, we have assumed $d\ge 4$
throughout the proof.  Nevertheless, by more careful analysis of the
constructive interference effects we expect the results to generalize to
interactions in $d=3$.  Again, this remains a topic for further
investigation. \qed
\end{remark}

We wish to inspect the decay of the space-time covariance on the
kinetic time scale $t=\order{\lambda^{-2}}$.  More precisely, given
some test-functions $f,g\in \ell_2(\Z^d)$, with a compact support,
we study the expectation of a quadratic form,
\begin{align}
 \Elfin\bigl[\,\mean{f_\Lambda,\psi_0}^*
   \mean{g_\Lambda,\psi_{t/\vep}}\,\bigr]\, ,
\end{align}
where $\vep=\lambda^2$, $f_\Lambda(x)= \sum_{n\in \Z^d} f(x+L n)$,
and $g_\Lambda$ is obtained from $g$ similarly.  Since we assume the
test-functions to have a compact support, $f_\Lambda$ and
$g_\Lambda$ are, in fact, independent of $\Lambda$ for all large
enough lattice sizes. In addition, $\FT{f}_\Lambda(k)=\FT{f}(k)$,
and $\FT{g}_\Lambda(k)=\FT{g}(k)$ for all $k\in \Lambda^*$. To get a
finite limit, it will be necessary to cancel the rapidly oscillating
factors. To this end, let us define
\begin{align}\label{eq:defomla}
\omla(k) = \omega(k) + \lambda R_0 \, ,
\end{align}
where
\begin{equation}\label{eq:defR0}
 R_0=R_0(\lambda,\Lambda)= 2  \Evepfin[|\psi_0(0)|^2]\,.
\end{equation}
Differentiating the expectation value and applying Assumption
\ref{th:Ainitcond} shows that
\begin{align}\label{eq:R0lim}
\lim_{\Lambda\to\infty}R_0(\lambda,\Lambda) = 2\int_{\T^d}\rmd
 k\,W(k) \Bigl(1-2 \beta \lambda \int_{\T^d}\rmd k'\,W(k')^2\Bigr)
 + \order{\lambda^2}\, .
\end{align}
Then the task is to control the limit of the quadratic form
\begin{align}\label{eq:cov}
 \Qlfin[g,f](t) =
\Elfin\!\left[\mean{\FT{f},\FT{\psi}_0}^*
\mean{\rme^{-\ci \omla  t/\vep}
\FT{g},\FT{\psi}_{t/\vep}}\right] ,  \quad \vep=\lambda^2\, .
\end{align}

\begin{theorem}\label{th:main}
Consider the system described in Section \ref{sec:model} with an initial Gibbs
measure
satisfying Assumption \ref{th:Ainitcond} and
a dispersion relation satisfying Assumption
\ref{th:disprelass}.  Then there is $t_0>0$ such that
for all $|t|< t_0$, and for any $f,g\in \ell_2(\Z^d)$ with finite support,
\begin{align}\label{eq:mainQlim}
\lim_{\lambda\to 0}
\limsup_{\Lambda\to \infty} \left|
\Qlfin[g,f](t) - \int_{\T^d}\rmd k\, \FT{g}(k)^* \FT{f}(k)
W(k) \rme^{-\Gamma_1(k)|t|-\ci t \Gamma_2(k)}\right| = 0\, ,
\end{align}
where $\Gamma_j(k)$ are real, and $\Gamma(k)=\Gamma_1(k)+\ci \Gamma_2(k)$
is given by
\begin{align}\label{eq:defGamma}
& \Gamma(k_1) = -2 \int_0^\infty \!\rmd t
 \int_{(\T^d)^3} \rmd k_2 \rmd k_3 \rmd k_4  \delta(k_1+k_2-k_3-k_4)
\nonumber \\ & \quad  \times
\rme^{\ci t (\omega_1+\omega_2-\omega_3-\omega_4)} \left(
W_3 W_4 - W_2 W_4 - W_2 W_3
\right)
\end{align}
with $\omega_i = \omega(k_i)$, $W_i= W(k_i)$.
\end{theorem}

As discussed in the introduction,
we expect that the infinite volume limit of $\Qlfin[g,f](t)$ exists,
but since
proving this would have been a diversion from our main results, we have stated
the main theorem in a form which does not need this property.  Clearly,
if the limit does exist, then (\ref{eq:mainQlim}) implies
the stronger result
\begin{align}
\lim_{\lambda\to 0}
\lim_{\Lambda\to \infty}
\Qlfin[g,f](t) = \int_{\T^d}\rmd k\, \FT{g}(k)^* \FT{f}(k)
W(k) \rme^{-\Gamma_1(k)|t|-\ci t \Gamma_2(k)}\, .
\end{align}
Independently of the convergence issue, the theorem implies that, 
if $\Lambda$ is sufficiently large, $\lambda$ is small enough, 
and $|\tmicro | \lambda^2$ is not too large, we can approximate
\begin{align}
&\Elfin[\FT{\psi}_0(k')^*\FT{\psi}_{\tmicro }(k)]
\approx 
\delta_\Lambda(k'-k)  W(k)  \rme^{-\ci \omega_{\text{ren}}^\lambda(k) \tmicro }
 \rme^{- \left|\lambda^2 \tmicro \right| \Gamma_1(k)}\, ,
\end{align}
where the ``renormalized dispersion relation'' is given by
\begin{align}
  \omega_{\text{ren}}^\lambda(k) = 
\omega(k) + \lambda R_0 + \lambda^2 \Gamma_2(k)\, .
\end{align}

We point out that $\Gamma_1(k)\ge 0$, as by explicit computation
\begin{align}\label{eq:Gamma2}
& \Gamma_1(k_1) = 2\pi W(k_1)^{-2}
 \int_{(\T^d)^3} \rmd k_2 \rmd k_3 \rmd k_4  \delta(k_1+k_2-k_3-k_4)
\nonumber \\ & \qquad  \times
\delta(\omega_1+\omega_2-\omega_3-\omega_4)  \prod_{i=1}^4 W(k_i) \, .
\end{align}
(We prove in Section \ref{sec:freeint} that the integral in 
(\ref{eq:defGamma}) and the positive measure 
in (\ref{eq:Gamma2}) are well-defined for any $\omega$ satisfying items
DR\ref{it:DR1} and DR\ref{it:DRdisp} of Assumption \ref{th:disprelass}.)
If $\Gamma_{1}(k) > 0$, then the term $\exp[-\Gamma_{1}(k)|t|]$ yields the
exponential damping in $|t|$, both forward and backwards in time, and if
$\Gamma_{1}(k) \geq \gamma >0$ for all  $k\in \T^d$, then  
on the kinetic scale the covariance has an exponential bound
$\mathrm{e}^{-\gamma|t|}$. 

\begin{remark}
The restriction to finite
times with $|t|<t_0<\infty$ appears artificial since
the limit equation is  obviously well-defined for all $t\in\R$.  In fact, as
can be inferred from the proof given in Sec.~\ref{sec:completion}, 
if we collect only the terms having a non-zero contribution to the
limit, the expansion is not restricted by such a finite radius of convergence.
However, the 
bounds used to control the remaining terms are not summable beyond certain
radius.  As a comparison, let us observe that even the perturbation
expansions of solutions to nonlinear kinetic equations, such as
(\ref{eq:4}) below, have generically only a finite radius of convergence. \qed
\end{remark}

\subsection{Link to kinetic theory}\label{sec:link}

To briefly explain the connection of our result to the kinetic
theory for weakly nonlinear wave equations, we assume that the
initial data $\psi(x)$, $x\in\mathbb{Z}^d$, are distributed
according to a Gaussian measure, $\mathbb{P}_\mathrm{G}$, with mean
zero and covariance
\begin{equation}\label{1}
\mathbb{E}_\mathrm{G}\big(\psi(y)^\ast\psi(x)\big)=\int_{\mathbb{T}^d}
\rmd k\,  h^0(k) \rme^{\ci 2\pi k\cdot(x-y)}\,,\quad
\mathbb{E}_\mathrm{G}\big(\psi(y)\psi(x)\big)=0\,.
\end{equation}
$\mathbb{P}_\mathrm{G}$ is stationary under the $\lambda=0$
dynamics, but nonstationary for $\lambda>0$. Since translation and
gauge invariance are preserved in time, necessarily
\begin{equation}\label{2}
\mathbb{E}_\mathrm{G}\big(\psi_t(y)^\ast\psi_t(x)\big)=\int_{\mathbb{T}^d}
\rmd k\,  h_\lambda(k,t) \rme^{\ci 2\pi k\cdot(x-y)}\,,\quad
\mathbb{E}_\mathrm{G}\big(\psi_t(y)\psi_t(x)\big)=0\,.
\end{equation}
The central claim of kinetic theory is the existence of the limit
\begin{equation}\label{3}
\lim_{\lambda\to 0} h_\lambda(k,\lambda^{-2}t)=h(k,t)\,,
\end{equation}
where $h(t)$ is the solution of the spatially homogeneous kinetic
equation
\begin{equation}\label{eq:4}
\frac{\partial}{\partial t}
h(k,t)=\mathcal{C}\big(h(\cdot,t)\big)(k)
\end{equation}
with initial conditions $h(k,0)=h^0(k)$. The collision operator,
$\mathcal{C}$, is defined by
\begin{align}\label{5}
& \mathcal{C}(h)(k_1) =  4\pi\int_{(\T^{d})^3} \rmd k_2 \rmd k_3 \rmd k_4
\delta(k_1+k_2-k_3-k_4) \delta
(\omega_1+\omega_2-\omega_3-\omega_4)\nonumber \\ & \qquad\qquad\qquad\quad
\times\left(h_2 h_3 h_4 + h_1 h_3 h_4 - h_1 h_2 h_3 - h_1 h_2 h_4\right)
\end{align}
with $h_j$ shorthand for $h(k_j)$, $j=1,2,3,4$. The proof of the
limit (\ref{3}) remains as mathematical challenge.

Under our conditions on $\beta$ and $\mu$, the covariance function 
$h^\mathrm{eq}(k)= W(k) = (\beta(\omega(k)-\mu))^{-1}$ is a
stationary solution of (\ref{eq:4}). 
The time correlation $\Qlfin[g,f](t)$ can be viewed as a small perturbation of
the equilibrium situation 
and should thus be governed by the linearization of (\ref{eq:4}) at 
$h^\mathrm{eq}$. As
discussed in \cite{spohn05},
the precise form of the linearization depends on the context. Our result
corresponds to the linearization of the loss term of $\mathcal{C}(h)$ relative 
to  $h^\mathrm{eq}$. In addition,
 only ``half'' of the energy conservation shows up: instead
of
\begin{equation}\label{6}
 \int_{-\infty}^\infty \rmd t\, \rme^{\ci t (\omega_1+ \omega_2 - \omega_3 -
\omega_4)} 
 = 2\pi\delta (\omega_1+\omega_2-\omega_3-\omega_4)\, ,
\end{equation}
only
\begin{equation}\label{7}
\int^\infty_0 \rmd t\, \rme^{\ci t (\omega_1+ \omega_2 - \omega_3 - \omega_4)}
\end{equation}
appears in the definition of the decay rate (\ref{eq:defGamma}).

\subsection{Restriction to times $t>0$}
\label{sec:firstpf}

From now on we assume that Assumptions \ref{th:Ainitcond} and
\ref{th:disprelass} are satisfied.
We begin by showing that then
it is sufficient to prove the theorem under the assumption $t>0$.
For simplicity, let us  denote $\E= \Elfin$ and $F_2= F^\lambda_2$, \itie ,
we define
\begin{align} 
 F_2(x,t) = \E[\psi_0(0)^* \psi_t(x)], \quad x\in \Lambda,\ t\in \R\,.
\end{align}

In order to study the infinite volume limit $\Lambda\to \Z^d$, we
define the natural ``cell step function'' 
$\tdfloor{k}:\R^d\to\Lambda^*$ by setting $\tdfloor{k}_i$ equal to 
$\lfloor L (k_i \bmod 1)\rfloor /L$. Since $\tdfloor{k}$ is periodic, we can
also identify
it with a map $\T^d\to \Lambda^*$. Clearly, for any $F:\Lambda^* \to \C$ 
we can then apply the following obvious formula relating the
discrete sum over $\Lambda^*$ and a Lebesgue integral:
\begin{align}\label{eq:LamtoLeb}
  \int_{\Lambda^*}\rmd k\, F(k)
 =  \int_{\T^d}\rmd k\, F(\tdfloor{k}) \, ,
\end{align}
where $F(\tdfloor{k})$ is a piecewise constant ``step function''
on $\T^d$.
Now if $F_\Lambda$ is any sequence of functions $\Lambda^* \to \C$
such that $F_\Lambda([k])$ converges on $\T^d$ to $F$, and for which 
$\sup_{\Lambda}\sup_{k\in\Lambda^*} |F_\Lambda(k)|<\infty$,
then by dominated convergence, we have
\begin{align} 
 \lim_{\Lambda\to \infty} \int_{\Lambda^*}\rmd k\, F_\Lambda(k)
 =  \int_{\T^d}\rmd k\, F(k) \, .
\end{align}

At $t=0$, $\FT{F}_2(k,0)=W^{\lambda}_\Lambda(k)$, for $k\in \Lambda^*$.
In the following Lemma, we show that it remains uniformly bounded in the
infinite volume limit, with a bound that vanishes as $\lambda\to 0$.  
Thus we can employ the definition
(\ref{eq:cov}) for $t=0$, and apply
the smoothness of $\FT{f},\FT{g}$, to conclude that
\begin{align} 
  \lim_{\lambda\to 0}
\limsup_{\Lambda\to \infty}\left|
\Qlfin[g,f](0) - \int_{\T^d}\rmd k\, \FT{g}(k)^* \FT{f}(k)
W(k)\right|=0\, .
\end{align}
This proves that the main theorem holds  at $t=0$.
\begin{lemma}\label{th:unifW2}
For all $0<\lambda\le \lambda_0$,
\begin{align} 
\limsup_{\Lambda\to \infty}\sup_{k\in \T^d}\left|
W^{\lambda}_\Lambda([k])-W(k)\right|\le 2c_0^2 \lambda \, .
\end{align}
\end{lemma}
\begin{proof}
Since our conditions on $\beta$ and $\mu$ imply that
\begin{align} 
  W(k) = \frac{1}{\beta(\omega(k)-\mu)} =
\sum_{x\in \Z^d}  \rme^{-\ci 2\pi x \cdot k}
  \EG\bigl[\psi(0)^*\psi(x)\bigr]
\end{align}
is smooth, its inverse Fourier transform
$\IFT{W}(x)=\EG\bigl[\psi(0)^*\psi(x)\bigr]$ belongs to
$\ell_1(\Z^d)$.  Now for any $k\in \T^d$
\begin{align} 
& |W^{\lambda}_\Lambda([k])-W(k)| \le
|W(k)-W([k])|
\nonumber \\ & \qquad
+ \sum_{\norm{x}_\infty \ge L/2}\!\!\! |\IFT{W}(x)|
+  \sum_{\norm{x}_\infty \le L/2} \Bigl|\Elfin\bigl[\psi(0)^*\psi(x)\bigr] -
  \EG\bigl[\psi(0)^*\psi(x)\bigr] \Bigr| \, ,
\end{align}
and the second part of Assumption \ref{th:Ainitcond} implies that
the Lemma holds.
\qed \end{proof}

The initial state is invariant under periodic translations of the lattice.
Since the time evolution also commutes with these translations,
we have
\begin{align} 
 \E[\psi_0(x_0)^* \psi_t(x)] = F_2(x-x_0,t)\, ,
\end{align}
and thus for $k,k'\in \Lambda^*$,
\begin{align} 
 \E[\FT{\psi}_0(k')^* \FT{\psi}_t(k)] = \delta_{\Lambda}(k'-k)
\FT{F}_2(k,t)\, .
\end{align}
Therefore,
\begin{align}\label{eq:Qlformula}
  \Qlfin[g,f](t) =
  \int_{\Lambda^*}\rmd k\, \FT{g}(k)^* \FT{f}(k)
  \rme^{\ci \omla (k)  t/\vep}
  \FT{F}_2(k,t/\vep)\, .
\end{align}

In addition, since the initial measure is stationary
and the process fully translation invariant, we have
\begin{align}\label{eq:C2inv}
  F_2(-x,-t)^* =
 \E[\psi_0(0) \psi_{-t}(-x)^*] =
 \E[\psi_0(x) \psi_{-t}(0)^*] =
  F_2(x,t) \, ,
\end{align}
and thus
\begin{align}\label{eq:FTC2inv}
  \FT{F}_2(k,-t)^* =   \FT{F}_2(k,t) \, .
\end{align}
Applied to (\ref{eq:Qlformula}) this implies that, in fact,
\begin{align}\label{eq:Qlinvrel}
  \bigl(\Qlfin[g,f](-t)\bigr)^* =
  \Qlfin[f,g](t) \, .
\end{align}

Let us assume that the main theorem has been proven for $t>0$.
Then for any $-t_0<t<0$ we have $0<-t<t_0$ and thus
\begin{align}
\lim_{\lambda\to 0}
\limsup_{\Lambda\to \infty} \left|
\Qlfin[f,g](-t) - \int_{\T^d}\rmd k\, \FT{f}(k)^* \FT{g}(k)
W(k) \rme^{-\Gamma_1(k)(-t)+\ci t \Gamma_2(k)}\right| = 0\, .
\end{align}
By (\ref{eq:Qlinvrel}), this implies that (\ref{eq:mainQlim}) holds then
also for all $-t_0<t<0$.  We have thus shown that it is sufficient to
prove the main theorem under the additional assumption $t>0$.
This will be done in the following sections.

\section{Duhamel expansion}
\label{sec:graphs}

From now on, let $d\ge 4$ and $t>0$ be given and fixed.
We also denote $\E= \Elfin$, as before.
In this section, we describe how the time-correlations are expanded into a sum
over amplitudes---integrals with somewhat complicated structure  
which can be encoded in Feynman graphs.

We begin from the Fourier transformed evolution equations,
(\ref{eq:FTdNLS2}).  Constructive interference turns out to be a problem for
the perturbation expansion, and we have to treat the wave numbers near the
``singular'' manifold $\Msing$ differently from the rest.
To this end, we introduce a cutoff function $\PFzero:(\T^d)^3\to [0,1]$ which
is smooth, depends on $\lambda$,
and is zero apart from a small neighborhood of $\Msing$.  Given such a
function let us denote $\PFone=1-\PFzero$.  We postpone the explicit
construction of the  function $\PFzero$ until Section \ref{sec:cutoff} where
we will also show that there is a constant $\lambda'_0>0$ 
such that the following Proposition holds.
\begin{proposition}\label{th:PFcorr}
Set $b=\frac{3}{4}$. There is a constant $C_1>0$ such that
for any $k\in (\T^d)^3$, $0<\lambda<\lambda'_0$, and for
any pair of indices $i\ne j$,
$i,j\in \set{1,2,3}$, all of the following statements hold:
\begin{enumerate}
\item If $k_i+k_j=0$,  then $\PFone(k)=0$ and $\PFzero(k)=1$.
\item $0\le \PFone (k)\le C_1 \lambda^{-b} d(k_i+k_j,\Msing)$.
\end{enumerate}
In addition,
$\PFone(k_3,k_2,k_1)=\PFone(k_1,k_2,k_3)$,
$\PFzero(k_3,k_2,k_1)=\PFzero(k_1,k_2,k_3)$, and
\begin{align}\label{eq:PFzeroineq}
0\le \PFzero (k)\le \sum_{i,j=1; i<j}^3
\1\!\left(d(k_i+k_j,\Msing)<\lambda^b\right)\, .
\end{align}
\end{proposition}

We can then use the equality $1=\PFzero+\PFone$ to split the integral 
in (\ref{eq:FTdNLS2}) into two parts. More precisely, this way we obtain
\begin{align}\label{eq:NLS2}
& \frac{\rmd }{\rmd t} \FT{\psi}_t(k_1) = -\ci \omega(k_1) \FT{\psi}_t(k_1)
-\ci \lambda \int_{(\Lambda^*)^3}\!\! \rmd k_2 \rmd k_3\rmd k_4\,
 \delta_{\Lambda}(k_1+k_2-k_3-k_4)
\nonumber \\ & \qquad \times
  \FT{\psi}_t(k_2)^* \FT{\psi}_t(k_3) \FT{\psi}_t(k_4)
 \bigl[\PFone (-k_2,k_3,k_4) + \PFzero(-k_2,k_3,k_4)\bigr]  .
\end{align}
To see that there will be anharmonic effects of the order of $t \lambda$,
one only needs to multiply (\ref{eq:NLS2}) by $\FT{\psi}_t(k')^*$ and take
expectation value of the right hand side.  If we, for the moment, assume that
a Gaussian approximation is accurate, this indicates that the leading term
arises entirely from the term proportional to $\PFzero$, and that it can be
canceled by
using the constant $R_0=R_0(\lambda,\Lambda)$ defined in (\ref{eq:defR0}).
The following Lemma provides the exact connection.
\begin{lemma} 
Let $\Ptrunc$ denote the following ``pairing truncation'' operation:
\begin{align}\label{eq:defPtrunc}
  \Ptrunc\left[a_1 a_2 a_3 \right] =
  a_1 a_2 a_3 - \E[a_1 a_2] a_3
  - \E[a_1 a_3] a_2  - \E[a_2 a_3] a_1 \, .
\end{align}
Then considering any solution $\psi_t$,
for all $k_1\in \Lambda^*$, $t\in \R$,
\begin{align}\label{eq:1storder2}
&  \int_{(\Lambda^*)^3}\!\! \rmd k_2 \rmd k_3\rmd k_4\,
 \delta_{\Lambda}(k_1+k_2-k_3-k_4)
  \PFzero(-k_2,k_3,k_4) \FT{\psi}_t(k_2)^* \FT{\psi}_t(k_3) \FT{\psi}_t(k_4)
 \nonumber \\ & \quad
= R_0(\lambda,\Lambda)\FT{\psi}_t(k_1)
 + \int_{(\Lambda^*)^3}\!\! \rmd k_2 \rmd k_3\rmd k_4\,
 \delta_{\Lambda}(k_1+k_2-k_3-k_4)
\nonumber \\ & \qquad \times
 \PFzero (-k_2,k_3,k_4) \Ptrunc
\bigl[\FT{\psi}_t(k_2)^* \FT{\psi}_t(k_3) \FT{\psi}_t(k_4)\bigr] .
\end{align}
\end{lemma}
\begin{proof}
Let us consider the second term on the right hand side of
(\ref{eq:1storder2}), and use the definition of $\Ptrunc$, equation
(\ref{eq:defPtrunc}), to expand it as a sum of four terms.
One of them is equal to the left hand side of (\ref{eq:1storder2}).
To evaluate the other three, let us first note that
\begin{align} 
\E\bigl[\FT{\psi}_t(k_3) \FT{\psi}_t(k_4)\bigr]=
\E\bigl[\FT{\psi}_0(k_3) \FT{\psi}_0(k_4)\bigr]= 0\, ,
\end{align}
by stationarity and invariance of the initial measure under rotations of
the total
phase of $\psi$.  Secondly, using also invariance of the measure under
spatial translations, we find, for both $i=3$ and $i=4$,
\begin{align} 
& \E\bigl[\FT{\psi}_t(k_2)^* \FT{\psi}_t(k_i)\bigr]
= \delta_\Lambda(k_i-k_2)
\sum_{x'\in\Lambda} \rme^{-\ci 2\pi x'\cdot k_i}
\E\bigl[\psi_0(0)^* \psi_0(x')\bigr]\, .
\end{align}
By Proposition \ref{th:PFcorr}, $\PFzero (-k_2,k_3,k_4)=1$, if $k_3=k_2$, or
$k_4=k_2$.
On the other hand, the above result implies that the expectation value is
zero otherwise.  Therefore,
\begin{align} 
 \PFzero (-k_2,k_3,k_4) \E\bigl[\FT{\psi}_t(k_2)^* \FT{\psi}_t(k_i)\bigr]
= \delta_\Lambda(k_i-k_2)
\E\bigl[\psi_0(0)^* \FT{\psi}_0(k_2)\bigr]\, .
\end{align}
Thus we can conclude that the right minus the left hand side of
(\ref{eq:1storder2}) is equal to
\begin{align} 
& R_0(\lambda,\Lambda)\FT{\psi}_t(k_1) -
 \int_{(\Lambda^*)^3}\!\! \rmd k_2 \rmd k_3\rmd k_4\,
 \delta_{\Lambda}(k_1+k_2-k_3-k_4)
\nonumber \\ & \qquad \times
\E\bigl[\psi_0(0)^* \FT{\psi}_0(k_2)\bigr]
\left( \delta_\Lambda(k_3-k_2)
\FT{\psi}_t(k_4) + \delta_\Lambda(k_4-k_2) \FT{\psi}_t(k_3) \right) \, .
\end{align}
Then summation over $k_3$ and $k_4$, and application of the definition of
$R_0$, shows that the term is equal to zero.  This completes the proof of
the Lemma.
\qed \end{proof}

We recall the definition of $\omla$ in (\ref{eq:defomla}), and use it to define
a random field $a_t(x)$ via its Fourier transform,
\begin{align} 
\FT{a}_t(k)= \rme^{\ci t \omla(k)} \FT{\psi}_t(k),
\quad k\in \Lambda^*.
\end{align}
Then $a_0(x)=\psi_0(x)$ and $\FT{a}_t$ satisfies the evolution equation
\begin{align}\label{eq:aNLS}
& \frac{\rmd }{\rmd t} \FT{a}_t(k_1) =
-\ci \lambda \int_{(\Lambda^*)^3}\!\! \rmd k_2 \rmd k_3\rmd k_4\,
 \delta_{\Lambda}(k_1+k_2-k_3-k_4)
\nonumber \\ & \qquad \times
\rme^{\ci t (\omega(k_1)+\omega(k_2)-\omega(k_3)-\omega(k_4))}
\Bigl\{ \PFone (-k_2,k_3,k_4)
 \FT{a}_t(k_2)^* \FT{a}_t(k_3) \FT{a}_t(k_4)
\nonumber \\ & \qquad\quad
 + \PFzero (-k_2,k_3,k_4) \Ptrunc\bigl[
 \FT{a}_t(k_2)^* \FT{a}_t(k_3) \FT{a}_t(k_4) \bigr]
\Bigl\} \, .
\end{align}
Note that the pure phase factor depending on $R_0$ cancels out from the
equation.  Clearly, now
\begin{align}\label{eq:Qwitha}
 \Qlfin[g,f](t) =
\E\!\left[\mean{\FT{f},\FT{a}_0}^* \mean{\FT{g},\FT{a}_{t/\vep}}\right] .
\end{align}

We set next
$a_t(x,1) = a_t(x)$ and $a_t(x,-1) = a_t(x)^*$,
which imply $\FT{a}_t(k,-1) = \FT{a}_t(-k,1)^*$ and
$\FT{a}_t(k,1)^*=\FT{a}_t(-k,-1)$.
By the above discussion, we need to study
the limit of
\begin{align} 
&  \Qlfin[g,f](t) =
\int_{(\Lambda^*)^2}\!\! \rmd k\rmd k'\, \FT{g}(k)^* \FT{f}(-k')
\E[\FT{a}_0(k',-1)\FT{a}_{t/\vep}(k,1) ]
\end{align}
where we have changed variables to $-k'$ in the second integral.
The new fields satisfy
\begin{align}\label{eq:apmNLS}
& \frac{\rmd }{\rmd t} \FT{a}_t(k,\sigma) =
-\ci \lambda \sigma \int_{(\Lambda^*)^3}\!\! \rmd k_1' \rmd k_2'\rmd k_3'\,
 \delta_{\Lambda}(k-k'_1-k'_2-k'_3) \rme^{-\ci t \Omega(k',\sigma)}
\nonumber \\ & \ \times
\Bigl\{ \PFone (k'_1,k'_2,k'_3)
 \FT{a}_t(k'_1,-1) \FT{a}_t(k'_2,\sigma) \FT{a}_t(k'_3,1)
\nonumber \\ & \qquad
+ \PFzero (k'_1,k'_2,k'_3) \Ptrunc\bigl[
 \FT{a}_t(k'_1,-1) \FT{a}_t(k'_2,\sigma) \FT{a}_t(k'_3,1)
 \bigr]
\Bigl\}
\end{align}
where $\Omega$ is defined by (\ref{eq:defOmega}),
\begin{align} 
& \Omega(k,\sigma)
= \omega(k_3)-\omega(k_1)+\sigma(\omega(k_2)-\omega(k_1+k_2+k_3))\, .
\end{align}

One more obstacle needs to be overcome. The simplest estimates
for the additional decay for non-leading terms will produce decay only in a
form of a small additional power of $\lambda$. However, our methods of
estimating the 
error terms produce always an additional factor of $\lambda^{-2}$.
One could try to improve the decay estimates by resorting to much more refined
classification of the decay of each term, similarly to what was needed in the
analysis of the random Schr\"odinger equation beyond kinetic time scales in
\cite{erdyau04,erdyau05b,erdyau05a}.
For our present goal of studying the kinetic time scale, a more convenient
tool is
to use ``partial time-integration'' first introduced in \cite{erdyau99}, and
somewhat improved to ``soft partial time-integration'' in \cite{ls05}.
The idea of the partial time-integration is to repeat the Duhamel expansion in
the error term, but only ``inside'' a certain small time-window.  The
additional decay is then produced by a large number of collisions which are
forced to  happen in the short time available.

To use the soft partial time-integration, we first record the obvious relation
\begin{align}\label{eq:PIrel}
 \frac{\rmd }{\rmd t}\left[ \rme^{\kappa t} \FT{a}_t(k_1,\sigma) \right] =
 \kappa \rme^{\kappa t} \FT{a}_t(k_1,\sigma)
 + \rme^{\kappa t} \frac{\rmd }{\rmd t} \FT{a}_t(k_1,\sigma) ,
\end{align}
valid for all $\kappa\in \C$. Thus for higher moments
\begin{align}\label{eq:ahighermom}
& \frac{\rmd }{\rmd t}\Bigl[  \rme^{\kappa t}
  \prod_{i=1}^n \FT{a}_t(k_i,\sigma_i) \Bigr]
= \kappa   \rme^{\kappa t} \prod_{i=1}^n \FT{a}_t(k_i,\sigma_i)
-\ci \lambda \sum_{j=1}^n \sigma_j
 \prod_{i=1;i\ne j}^n\!\! \FT{a}_t(k_i,\sigma_i)
\nonumber \\ & \ \times
 \int_{(\Lambda^*)^3}\!\! \rmd k'\,
  \delta_{\Lambda}(k_j-k'_1-k'_2-k'_3)
 \rme^{\kappa t-\ci t \Omega(k',\sigma_j)}
\Bigl\{ \PFone (k')
 \FT{a}_t(k'_1,-1) \FT{a}_t(k'_2,\sigma_j) \FT{a}_t(k'_3,1)
\nonumber \\ & \qquad\qquad
 + \PFzero (k') \Ptrunc\bigl[
 \FT{a}_t(k'_1,-1) \FT{a}_t(k'_2,\sigma_j) \FT{a}_t(k'_3,1)
 \bigr]
\Bigl\}
  \, .
\end{align}
Integrating (\ref{eq:ahighermom}) over time, and then multiplying with
$\rme^{-\kappa t}$,
yields the following Duhamel
formula with soft partial time-integration: for any $\kappa\ge 0$,
\begin{align}\label{eq:kappaDuh}
& \prod_{i=1}^n \FT{a}_t(k_i,\sigma_i)
= \rme^{-\kappa t} \prod_{i=1}^n \FT{a}_0(k_i,\sigma_i)
+ \kappa \int_0^t\!\rmd s\, \rme^{-(t-s)\kappa}
 \prod_{i=1}^n \FT{a}_s(k_i,\sigma_i)
\nonumber \\ & \quad
-\ci \lambda  \int_0^t\!\rmd s \sum_{j=1}^n \sigma_j
 \int_{(\Lambda^*)^3}\!\! \rmd k'\,
  \delta_{\Lambda}(k_j-k'_1-k'_2-k'_3)
 \rme^{-\kappa (t-s)-\ci s  \Omega(k',\sigma_j)}
\nonumber \\ & \qquad \times
 \prod_{i=1;i\ne j}^n\!\! \FT{a}_s(k_i,\sigma_i)
\Bigl\{ \PFone (k')
 \FT{a}_s(k'_1,-1) \FT{a}_s(k'_2,\sigma_j) \FT{a}_s(k'_3,1)
\nonumber \\ & \qquad\qquad
 + \PFzero (k') \Ptrunc\bigl[
 \FT{a}_s(k'_1,-1) \FT{a}_s(k'_2,\sigma_j) \FT{a}_s(k'_3,1)  \bigr]
\Bigl\}
 \, .
\end{align}
If $\kappa>0$, the first two terms can be combined to get a formula similar to
that given in Theorem 4.3 in \cite{ls05},
\begin{align}\label{eq:softpiold}
& \rme^{-\kappa t} \prod_{i=1}^n \FT{a}_0(k_i,\sigma_i)
+ \kappa \int_0^t\!\rmd s\, \rme^{-(t-s)\kappa}
 \prod_{i=1}^n \FT{a}_s(k_i,\sigma_i)
= \kappa \int_0^\infty \!\rmd r\, \rme^{-r \kappa}
 \prod_{i=1}^n \FT{a}_{(t-r)_+}(k_i,\sigma_i)
\end{align}
with $(r)_+=r$, if $r\ge 0$, and $(r)_+=0$, if $r< 0$.

We now iterate this formula for  $N_0\ge 1$ times, using it \defem{only} in the
term containing $\PFone$,
the complement of the cutoff function.  Then at each iteration step we
get three new terms,
one depending only on the initial field, $a_0$, one coming from the remainder
of the partial time integration and one containing the cutoff function
$\PFzero$.  Explicitly, this yields for any
$\kappa\in \R_+^{\set{0,1,\ldots,N_0-1}}$ an expansion
\begin{align}\label{eq:mainaiter}
& \FT{a}_t(k,\sigma) = \sum_{n=0}^{N_0-1}
\mathcal{F}_n(t,k,\sigma,\kappa)[\FT{a}_0]
+ \sum_{n=0}^{N_0-1} \kappa_n \int_0^t\! \rmd s\,
\mathcal{G}_{n}(s,t,k,\sigma,\kappa)[\FT{a}_s]
 \nonumber \\ & \quad
+ \sum_{n=1}^{N_0} \int_0^t\! \rmd s\,
\mathcal{Z}_n(s,t,k,\sigma,\kappa)[\FT{a}_s]
+ \int_0^t\! \rmd s\,
\mathcal{A}_{N_0}(s,t,k,\sigma,\kappa)[\FT{a}_s] .
\end{align}
Each of the functionals is a polynomial of $\FT{a}_s$, for some fixed time
$s$, and their structure can be encoded in diagrams
whose construction we describe next.

For given $n, n'$, with $0\le n' \le n$, we first
define the index sets $I_n=\set{1,2,\ldots,n}$ and
$I_{n',n}=\set{n',n'+1,\ldots,n}$.
For further use,
let $m_0\ge 1$ denote the number of fields at the final time $t$ (in the above
case of $\FT{a}_t$
we thus set $m_0=1$).  Also, let $N\ge 0$ denote the total number of
interactions, \itie , iterations
of the Duhamel formula.
A term with $N$ interactions has the total time $t$ divided
into $N+1$ ``time slices'' of length $s_i$, $i=0,1,\ldots,N$, labeled in
their proper time-order (from bottom to top in the diagram).
Associated with a time slice $i$ there are in total
$m_{N-i}$ ``momentum integrals'' over $\Lambda^*$, where $m_n=m_0+2 n$.
We label the momenta by $k_{i,j}$
and associate a line segment in the diagram to each of them.
The interactions are denoted by an interaction vertex.
Each interaction vertex thus contains
a $\delta_\Lambda$-function which enforces the momenta below the vertex
(belonging to an earlier time slice)
to sum up to the momenta above the vertex.  The momenta not involved in an
interaction are continued unchanged from one time slice to the next.  Thus a
natural way of representing the line segments is to connect them into straight
lines passing through several time slices
until they encounter an interaction vertex at which three such lines fuse into
one new momentum line.
For this reason, we will call the interactions \defem{fusions} from now on.

To summarize the notations, the fusion number $1$, denoted by an interaction
vertex $v_1$,
happens after time $s_0$ which is the length of
the time slice number $0$, fusion $2$ happens after time $s_0+s_1$,
where $s_1$ is the length of the time slice $1$, etc.
In general, fusion $i$ happens in the beginning of the slice $i$.
For each time slice $i\in I_{0,N}$ we label the momenta
by $k_{i,j}$, $j=1,\ldots,m_{N-i}$.
Similar labeling is used for the ``parity'' $\sigma_{i,j}\in \set{\pm 1}$.
The structure of interactions is such that the parity of each line is uniquely
determined by the parities of the final lines.  In our diagrams, we use the
order implicit in (\ref{eq:kappaDuh}):
the parities of the fusing line-segments are required to appear in the order
$(-1,\sigma,+1)$, and then the parity of the
\defem{middle} line  will be carried on
to determine the parity resulting in a fusion.  Fig.~ \ref{fig:indexgr}
illustrates these definitions.

\begin{figure}
  \centering
  \myfigure{width=0.9\textwidth}{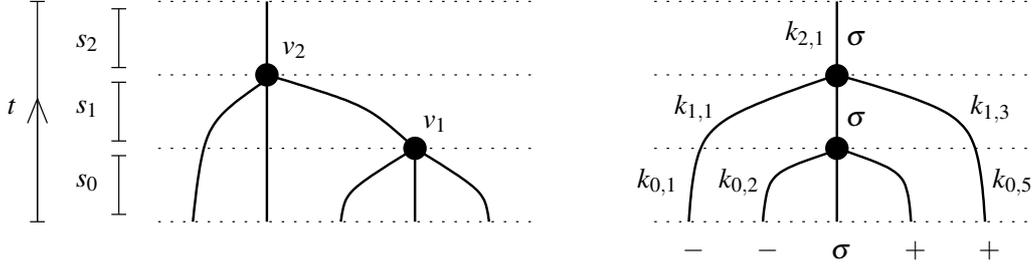}
\caption{Two examples of interaction diagrams with $N=2$ interactions and
$m_0=1$ final fields.
In the left diagram we have indicated the notations used for time slices and
interaction vertices. In the right, the parity of each line is shown, assuming
that the final line has parity $\sigma$, as well as some of the notations used
for momenta associated with line-segments on each time slice.
The ``interaction history'' of the diagram is $\ell=(3,1)$ on the left and
$\ell=(2,1)$ on the right.}\label{fig:indexgr}
\end{figure}

Let
$\mathcal{I}_{n;m_0} = \defset{(i,j)}{0\le i\le n-1, 1\le j \le m_0+2(n-i)}$ 
which is a subset of 
$I_{0,n-1}\times I_{m_0+2n}$. Then  the set $\mathcal{I}_{N;m_0}$
collects all index pairs associated with momentum line-segments,
excluding the final time slice with $i=N$.  We also employ the shorthand
notation $\mathcal{I}_{n}$ for $\mathcal{I}_{n;1}$.
We use a vector $\ell$ to define
the interaction history by collecting, for every time slice with $i\ge 1$,
the index of the new line formed in the fusion at the beginning of the
slice.  Then $\ell\in G_N$, with
$G_N= I_{m_{N-1}}\times I_{m_{N-2}}\times \cdots \times I_{m_{0}}$, and let also
$G_0=\emptyset$.  By the earlier explained procedure, the indices in each time
slice are matched so that
the indices made vacant by the fusion are filled by shifting the indices
following the fusion line down by two.
This corresponds to labeling the momenta in each time slice by
counting them from left to right in the natural graphical representation of the
interaction history, where the lines intersect only at interaction
vertices.  (See Fig.~ \ref{fig:indexgr} for an illustration.)

Explicitly, 
$\mathcal{F}_0(t,k,\sigma,\kappa)[\FT{a}]=\rme^{-\kappa_0 t} \FT{a}(k,\sigma)$
and, for $n> 0$,
\begin{align}\label{eq:defFn}
& \mathcal{F}_n(t,k_{n1},\sigma_{n1},\kappa)[\FT{a}]
=
(-\ci \lambda)^n \sum_{\ell \in G_n}
\sum_{\sigma\in \set{\pm 1}^{\mathcal{I}_{n}}}
\int_{(\Lambda^*)^{\mathcal{I}_{n}}} \!\rmd k\,
\Delta_{n,\ell}(k,\sigma;\Lambda)
\prod_{j=1}^{m_0+2 n} \FT{a}(k_{0,j},\sigma_{0,j})
 \nonumber \\ &  \quad \times
\prod_{i=1}^{n} \Bigl[ \sigma_{i,\ell_{i}} \PFone(k_{i-1;\ell_i}) \Bigr]
\int_{(\R_+)^{I_{0,n}}}\!\rmd s \,  \delta\Bigl(t-\sum_{i=0}^{n} s_i\Bigr)
\prod_{i=0}^{n} \rme^{-s_i \kappa_{n-i}}
\prod_{i=1}^{n} \rme^{-\ci t_i(s) \Omega_{i-1;\ell_i}(k,\sigma)}
 \, ,
\end{align}
where $t_i(s) = \sum_{j=0}^{i-1} s_j$
is the time needed to reach the beginning of the slice $i$, and
\begin{align} 
&
k_{i;j} = (k_{i,j},k_{i,j+1},k_{i,j+2})\in (\T^d)^3\, ,
 \\ &
\Omega_{i;j}(k,\sigma) = \Omega(k_{i;j},\sigma_{i+1,j}),
\end{align}
and $\Delta_{n,\ell}$ contains $\delta$-functions restricting the
integrals over
$k$ and $\sigma$ to coincide with the interaction history defined by $\ell$,
as described above.
Explicitly,
\begin{align} 
&\Delta_{n,\ell}(k,\sigma;\Lambda) =
\prod_{i=1}^{n} \Bigl\{
\prod_{j=1}^{\ell_i-1} \Bigl[
 \delta_\Lambda(k_{i,j}-k_{i-1,j}) \1(\sigma_{i,j}=\sigma_{i-1,j})
\Bigr]
 \nonumber \\ & \qquad \times
\delta_\Lambda\Bigl(k_{i,\ell_{i}}-\sum_{j=0}^{2} k_{i-1,\ell_i+j} \Bigr)
\prod_{j=\ell_i+1}^{m_{n-i}} \Bigl[
 \delta_\Lambda(k_{i,j}-k_{i-1,j+2}) \1(\sigma_{i,j}=\sigma_{i-1,j+2})
\Bigr]
 \nonumber \\ & \qquad \times
\1(\sigma_{i-1,\ell_{i}}=-1)
\1(\sigma_{i-1,\ell_{i}+1}=\sigma_{i,\ell_i})
\1(\sigma_{i-1,\ell_{i}+2}=+1)
\Bigr\}  \, .
\end{align}

The remaining terms are very similar.  For $n\ge 1$,
\begin{align} 
& \mathcal{A}_n(s_0,t,k_{n1},\sigma_{n1},\kappa)[\FT{a}]
= (-\ci \lambda)^n \sum_{\ell \in G_n}
\sum_{\sigma\in \set{\pm 1}^{\mathcal{I}_{n}}}
\int_{(\Lambda^*)^{\mathcal{I}_{n}}} \!\rmd k\,
\Delta_{n,\ell}(k,\sigma;\Lambda)
 \nonumber \\ & \quad \times
\prod_{j=1}^{m_0+2 n} \FT{a}(k_{0,j},\sigma_{0,j})
\prod_{i=1}^{n} \Bigl[ \sigma_{i,\ell_{i}}
 \PFone(k_{i-1;\ell_i})
 \Bigr]
 \nonumber \\ & \quad \times
\int_{(\R_+)^{I_{n}}}\!\rmd s \,  \delta\Bigl(t-s_0-\sum_{i=1}^{n} s_i\Bigr)
\prod_{i=1}^{n} \rme^{-s_i \kappa_{n-i}}
\prod_{i=1}^{n} \rme^{- \ci t_i(s) \Omega_{i-1;\ell_i}(k,\sigma)}
 \, .
\end{align}
We set $\mathcal{G}_{0}(s,t,k,\sigma,\kappa)[\FT{a}]
=\rme^{-(t-s)\kappa_0}\FT{a}(k,\sigma)=
\mathcal{F}_0(t-s,k,\sigma,\kappa)[\FT{a}]$, and for $n> 0$,
\begin{align}\label{eq:defGnviaA}
\mathcal{G}_{n}(s,t,k,\sigma,\kappa)[\FT{a}]
=\int_0^{t-s}\!\rmd r\, \rme^{-r\kappa_n}
\mathcal{A}_{n}(s+r,t,k,\sigma,\kappa)[\FT{a}],
\end{align}
and, finally,
\begin{align} 
& \mathcal{Z}_n(s_0,t,k_{n1},\sigma_{n1},\kappa)[\FT{a}]
= (-\ci \lambda)^n \sum_{\ell \in G_n}
\sum_{\sigma\in \set{\pm 1}^{\mathcal{I}_{n}}}
\int_{(\Lambda^*)^{\mathcal{I}_{n}}} \!\rmd k\,
\Delta_{n,\ell}(k,\sigma;\Lambda)
\sigma_{1,\ell_1} \PFzero (k_{0;\ell_1}) 
 \nonumber \\ & \ \times
\prod_{i=2}^{n} \Bigl[
\sigma_{i,\ell_{i}}
 \PFone(k_{i-1;\ell_i})
 \Bigr]
 \prod_{j=1}^{\ell_1-1} \FT{a}(k_{0,j},\sigma_{0,j})
\Ptrunc\Bigl[
 \prod_{j=\ell_1}^{\ell_1+2}  \FT{a}(k_{0,j},\sigma_{0,j})
\Bigr]
 \prod_{j=\ell_1+3}^{m_0+2 n} \FT{a}(k_{0,j},\sigma_{0,j})
 \nonumber \\ & \ \times
\int_{(\R_+)^{I_{n}}}\!\rmd s \,  \delta\Bigl(t-s_0-\sum_{i=1}^{n} s_i\Bigr)
\prod_{i=1}^{n} \rme^{-s_i \kappa_{n-i}}
\prod_{i=1}^{n} \rme^{-\ci t_i(s)  \Omega_{i-1;\ell_i}(k,\sigma)}
 \, .
\end{align}

Using these definitions, the validity of (\ref{eq:mainaiter}) can be proven by
induction in $N_0$,
applying (\ref{eq:kappaDuh}) to
$\mathcal{A}_{N_0}$ in (\ref{eq:mainaiter}).
For later use, let us point out that the total oscillating phase factor
in the above formulae can also be written as
\begin{align}\label{eq:phase2}
\prod_{i=1}^{n} \rme^{-\ci t_i(s) \Omega_i}
= \prod_{j=0}^{n-1}
\exp\Bigl[-\ci s_j \sum_{i=j+1}^n \Omega_i\Bigr], \quad \Omega_i =
\Omega_i(\ell,k,\sigma) = \Omega_{i-1;\ell_i}(k,\sigma) \, .
\end{align}
Applying the expansion to (\ref{eq:Qwitha}) proves the following result.
\begin{proposition}\label{th:Qmainerr}
For any $N_0\ge 1$ and for any choice of
$\kappa\in \R_+^{I_{0,N_0-1}}$, we have
\begin{align}\label{eq:pexpmainterm}
&  \Qlfin[g,f](t) = \Qmain + Q^{\rm err}_{\rm pti}+
Q^{\rm err}_{\rm cut} + Q^{\rm err}_{\rm amp} \, ,
\end{align}
where
\begin{align}\label{eq:Qmaindef}
\Qmain =
\int_{(\Lambda^*)^2}\!\! \rmd k\rmd k'\, \FT{g}(k)^* \FT{f}(-k')
\sum_{n=0}^{N_0-1}
\E\!\left[\FT{\psi}_0(k',-1)\mathcal{F}_n(t/\vep,k,1,\kappa)[\FT{\psi}_0]
\right]
\end{align}
and the error terms are given by
\begin{align}
&  Q^{\rm err}_{\rm pti} =
 \sum_{n=0}^{N_0-1} \kappa_n \int_0^{t/\vep}\! \rmd s\,
\E\Bigl[\mean{\FT{f},\FT{a}_0}^*
\int_{\Lambda^*}\!\! \rmd k\,  \FT{g}(k)^*
\mathcal{G}_{n}(s,t/\vep,k,1,\kappa)[\FT{a}_s]\Bigr], \\
&  Q^{\rm err}_{\rm cut} =
  \sum_{n=1}^{N_0} \int_0^{t/\vep}\! \rmd s\,
\E\Bigl[\mean{\FT{f},\FT{a}_0}^*
\int_{\Lambda^*}\!\! \rmd k\,  \FT{g}(k)^*
\mathcal{Z}_{n}(s,t/\vep,k,1,\kappa)[\FT{a}_s]\Bigr], \\
&  Q^{\rm err}_{\rm amp} =
 \int_0^{t/\vep}\! \rmd s\,
\E\Bigl[\mean{\FT{f},\FT{a}_0}^*
\int_{\Lambda^*}\!\! \rmd k\,  \FT{g}(k)^*
\mathcal{A}_{N_0}(s,t/\vep,k,1,\kappa)[\FT{a}_s]\Bigr] \, .
\end{align}
\end{proposition}

\subsection{Structure of the proof}
\label{sec:structureofproof}

We have now derived a
time-evolution equation for arbitrary moments of the field, and constructed a
related
Duhamel expansion of our observable.  Already at this stage we had to
introduce
certain additional structure compared to the standard Duhamel formula.
Certain regions of wavenumbers are treated differently, in
order to control ``bad'' constructive interference effects.
In addition, we have introduced an artificial exponential
decay for partial time integration which will be used to
amplify decay estimates which are too weak to be used in the error estimates.

The terms in this expansion either contain only finite moments of the
initial fields, or after relying on stationarity of the initial measure,
can be bounded by such moments.  We will employ our assumption
about the strong clustering properties of the initial measure to turn the
moments into cumulants whose analysis in the Fourier-space will result only
in additional ``Kirchhoff's rules'' on the initial time slice.  The
expectation values can then be expressed as a sum over graphs encoding the
various possible momentum- and time-dependencies of the integrand.
The construction of the graphs will be explained in Section \ref{sec:diagrams}

We will then derive a certain, essentially unique, way to resolve all the
momentum dependencies dictated by a graph, see Section \ref{sec:momdeltas}.
After this, it will
be a modest step to show that the limit $\Lambda\to \infty$ in essence
corresponds to replacing the discrete sums over $\Lambda^*$ by integrals
over $\T^d$.
The resulting graphs can then be classified, in the spirit of \cite{erdyau99},
by identifying in most of them certain integrals with oscillating factors which
produce additional decay compared to the leading graphs.
Here the idea is first to identify all ``motives'' which make the phase
factors to vanish identically in every second time slice, while the remaining
time slices are forced to have a subkinetic length due to the oscillating
phases.
These correspond to immediate recollisions in the language of the
earlier works, and repetitions of these motives yield the leading term
graphs.  Other graphs will be subleading either because they contain
additional
$k$-integrals, or because the $k$-integrals overlap in such a way that
additional time slices can be proven to have a subkinetic length.
As before, the overlap needs to be controlled in several different
fashions to find the appropriate mechanism for decay. This results in a
classification of these graphs into partially paired, nested, and crossing
graphs.

The control of the three different types of remainder terms can be accomplished
by slight modifications of the estimates used for  the main term.
The limit of the sum of the leading graphs is then shown to coincide with the
expression given in the main theorem.  The precise choice of expansion
parameters, as well as a preliminary classification of the graphs, will be
given in Section \ref{sec:classification}.  After establishing the main
technical lemmata in Section
\ref{sec:lemmas}, we derive the various estimates in two parts.
In Section \ref{sec:higherorder}
we consider higher order effects and the infinite volume limit.  Pairing graphs
can only be treated
after taking $\Lambda\to \infty$, and their analysis is given in Section
\ref{sec:fullypaired}.
Combined, the various estimates yield the result stated in Theorem
\ref{th:main}, as is shown in Section \ref{sec:completion}.


\section{Diagrammatic representation}\label{sec:diagrams}

In this section, we derive diagrammatic representations related to
the terms in Proposition \ref{th:Qmainerr}.  For the main terms summing to
$\Qmain$ the representation describes the value of the term, whereas for the
error terms, the representation is a contribution to an upper bound of the
term. The representations arise since we are able to derive upper bounds which
depend only on moments of the initial fields.
We first recall a standard result which relates moments to truncated
correlation functions of the time zero fields.

\subsection{Initial time clusters from a cumulant expansion}
\label{sec:cumulants}

Since $a_0(x)=\psi_0(x)$, the conditions for initial fields imply
$\EG[\FT{a}_0(k,\sigma)]=0$, and
\begin{align} 
  \EG[\FT{a}_0(k,\sigma) \FT{a}_0(k',\sigma')] =
  \1(\sigma+\sigma'=0) \delta(k+k') W(\sigma k)\, ,
\end{align}
for $\sigma',\sigma \in \set{\pm 1}$.  However, this formula is correct
only after taking the infinite volume and the weak coupling limit.
Before taking these limits there will be corrections to the cumulants.
These corrections can be controlled by relying on the strong clustering
assumption, as will be described next. 

Given $n\in \N$, we define for $k\in (\Lambda^*)^n$,
$\sigma\in \set{\pm 1}^n$, the truncated correlation function (or a cumulant
function) in Fourier-space as
\begin{align} 
 C_n(k,\sigma;\lambda,\Lambda) :=
 \sum_{x\in \Lambda^{n}} \1(x_1=0)
 \rme^{-\ci 2\pi \sum_{i=1}^{n} x_i\cdot k_i}
 \Elfin\Bigl[\prod_{i=1}^n \psi(x_i,\sigma_i)\Bigr]^{\rm trunc} \, .
\end{align}
An immediate consequence of the gauge invariance of the measure 
is that $C_n = 0$ if $\sum_{i=1}^n \sigma_i\ne 0$.  In particular, all odd
truncated correlation functions vanish. 
By Assumption \ref{th:Ainitcond}, apart from $n=2$, the functions for
all other even $n$ have uniform bounds,
\begin{align}\label{eq:Cnbound}
| C_n(k,\sigma;\lambda,\Lambda)|  \le \lambda (c_0)^n n!\, .
\end{align}

For $n=2$, we have
\begin{align} 
 C_2(k,\sigma;\lambda,\Lambda) =
 \sum_{x\in \Lambda}
 \rme^{-\ci 2\pi x \cdot k_2}
 \Elfin\!\left[\psi(0,\sigma_1)\psi(x,\sigma_2)\right] \, .
\end{align}
Thus, $C_2=0$ if $\sigma_1=\sigma_2$, and clearly also
$C_2(k,(1,-1))=C_2(-k,(-1,1))^*$, for all $k\in (\T^d)^2$.
A comparison with the definition
of $W^{\lambda}_\Lambda$ shows that
$C_2((k',k),(-1,1))=W^{\lambda}_\Lambda(k)$, and thus also
$C_2((k',k),(1,-1))=W^{\lambda}_\Lambda(-k)^*$.  However, by a direct
application of translation invariance in the definition we find that
$C_2((k',k),(1,-1))=W^{\lambda}_\Lambda(-k)$.  This implies that 
$W^{\lambda}_\Lambda$ is real valued.
By Lemma \ref{th:unifW2},
there is a constant $c_0'$ such that
$|W^{\lambda}_\Lambda(k)|\le c_0'$.  Thus also for $n=2$ the cumulant
functions are uniformly bounded, but this bound does not contain the factor
$\lambda$, as in the bound (\ref{eq:Cnbound}) for $n>2$. 

The cumulant functions are of interest since they allow expanding moments
in terms of uniformly bounded functions
via the following general result.
\begin{definition}\label{th:defPiI}
For any finite, non-empty set $I$, let $\pi(I)$ denote
the set of its partitions:
$S\in \pi(I)$ if and only if $S \subset \mathcal{P}(I)$ such that
each $A\in S$ is non-empty, $\cup_{A\in S} A = I$, and
if $A,A'\in S$ with $A'\ne A$ then $A'\cap A=\emptyset$.
In addition, we define $\pi(\emptyset)=\set{\emptyset}$.
\end{definition}
\begin{lemma}\label{th:cumulants}
For any index set $I$, and any $k\in (\T^d)^I$, $\sigma\in \set{\pm 1}^I$,
\begin{align} 
& \Elfin\Bigl[\prod_{i\in I} \FT{\psi}(k_i,\sigma_i)\Bigr]
= \sum_{S\in \pi(I)} \prod_{A\in S} \Bigl[ \delta_\Lambda\!\Bigl(\sum_{i\in
  A} k_i\Bigr)  C_{|A|}(k_A,\sigma_{\!A};\lambda,\Lambda) \Bigr] \, ,
\end{align}
where the sum runs over all partitions $S$ of the index set $I$,
and the shorthand notation $(k_A,\sigma_{\!A})$
refers to $(k_a,\sigma_a)_{a\in A}\in (\T^d\times \set{\pm 1})^A$, with an
arbitrary ordering of the elements $a\in A$.
\end{lemma}
\begin{proof}
Let $\E=\Elfin$.  We need to study
\begin{align} 
& \E\Bigl[\prod_{i\in I} \FT{\psi}(k_i,\sigma_i)\Bigr]
= \sum_{x\in \Lambda^I} \rme^{-\ci 2\pi \sum_i x_i\cdot k_i}
 \E\Bigl[\prod_{i\in I} \psi(x_i,\sigma_i)\Bigr]\, .
\end{align}
We denote the cumulant generating function by
$\mathcal{G}_c[f]=\ln \E[\rme^{\ci \sum_{x,\sigma}f(x,\sigma) \psi(x,\sigma)}]$,
using which 
\begin{align} 
& \E\Bigl[\prod_{i\in I} \psi(x_i,\sigma_i)\Bigr]
 = (-\ci)^{|I|} \Bigl[\prod_{i\in I} \partial_{f(x_i,\sigma_i)}\Bigr]
 \left.\rme^{\mathcal{G}_c[f]}\right|_{f=0}
\nonumber \\ & \quad
 = \sum_{S\in \pi(I)} \prod_{A\in S}
\Bigl[(-\ci)^{|A|} \prod_{i\in A} \partial_{f(x_i,\sigma_i)}
 \left.\mathcal{G}_c[f]\right|_{f=0}\Bigr]
 = \sum_{S\in \pi(I)} \prod_{A\in S}
 \E\Bigl[\prod_{i\in A} \psi(x_i,\sigma_i)\Bigr]^{\rm trunc} \, .
\end{align}
Since the measure is translation invariant, so are all of the truncated
correlation functions.  As is implicitly implied by the notation, 
they are obviously also
invariant under permutations of the index sets.  
Thus, if we choose an arbitrary ordering $i_A:\set{1,2,\ldots,|A|}\to A$ of 
the elements of  each $A\in S$, then
\begin{align} 
& \E\Bigl[\prod_{i\in I} \FT{\psi}(k_i,\sigma_i)\Bigr]
=  \sum_{S\in \pi(I)} \prod_{A\in S} \Bigl[
\sum_{x\in \Lambda^A} \rme^{-\ci 2\pi \sum_{a\in A}x_a\cdot k_a}
 \E\Bigl[\prod_{j=1}^{|A|} 
    \psi(x_{i_{A}(j)}-x_{i_{A}(1)},\sigma_{i_{A}(j)})\Bigr]^{\rm trunc}
 \Bigr]
\nonumber \\ & \quad
= \sum_{S\in \pi(I)} \prod_{A\in S}
\Bigl[ \delta_\Lambda\!\Bigl(\sum_{i\in A} k_i\Bigr)
  C_{|A|}(k_A,\sigma_A) \Bigr] \, .
\end{align}
This completes the proof of the Lemma.
\qed \end{proof}

\subsection{Main terms}\label{sec:maindiag}

Using Lemma \ref{th:cumulants} in $\Qmain$ yields a high dimensional integral
over the momenta, restricted to a certain subspace determined by
$\Delta_{n,\ell}$ and the $\delta_\Lambda$-functions arising from the cumulant
expansion.
The restrictions can be encoded in a ``Feynman diagram'', which is a planar 
graph where each edge corresponds to an independent momentum integral, and
each vertex carries the appropriate $\delta_\Lambda$-function (in physics
terminology, these can be interpreted as ``Kirchhoff's rules'' applied at the
vertex).  The explicit integral expressions are given in the following
proposition, and we will discuss their graphical representation in Section
\ref{sec:momdeltas}. 
\begin{proposition}\label{th:main1st}
For a given $N_0\ge 1$,
\begin{align}
\Qmain = \sum_{n=0}^{N_0-1}
 \sum_{\ell \in G_n} \sum_{S\in \pi(I_{0,2 n+1})}
\mathcal{F}_n^{\rm ampl}(S,\ell,t/\vep,\kappa)\, ,
\end{align}
where, setting
$\mathcal{I}''_{n} =\mathcal{I}_{n} \cup \set{(n,1)}\cup \set{(0,0)}$,
\begin{align}\label{eq:defFnampl}
& \mathcal{F}_n^{\rm ampl}(S,\ell,t/\vep,\kappa) =
(-\ci \lambda)^n
\sum_{\sigma\in \set{\pm 1}^{\mathcal{I}''_{n}}}
\int_{(\Lambda^*)^{\mathcal{I}''_{n}}} \!\rmd k\,
\Delta_{n,\ell}(k,\sigma;\Lambda)
 \nonumber \\ &  \quad \times
\prod_{A\in S}\Bigl[ \delta_\Lambda\!\Bigl(\sum_{i\in A} k_{0,i}\Bigr)
  C_{|A|}(\sigma_{0,A},k_{0,A};\lambda,\Lambda) \Bigr]
 \nonumber \\ & \quad \times
\1(\sigma_{n,1}=1)\1(\sigma_{0,0}=-1) \FT{g}(k_{n,1})^* \FT{f}(k_{n,1})
\prod_{i=1}^{n} \Bigl[ \sigma_{i,\ell_{i}} \PFone(k_{i-1;\ell_i}) \Bigr]
 \nonumber \\ & \quad \times
\int_{(\R_+)^{I_{0,n}}}\!\rmd s \,  \delta\Bigl(\frac{t}{\vep}-\sum_{i=0}^{n}
s_i\Bigr)
\prod_{i=0}^{n} \rme^{-s_i \kappa_{n-i}}
\prod_{m=0}^{n-1} \rme^{-\ci s_m \sum_{i=m+1}^{n}\Omega_{i-1;\ell_i}(k,\sigma)}
 \, .
\end{align}
\end{proposition}
\begin{proof}
The representation is a corollary of the results in Section \ref{sec:graphs},
after we relabel  $k=k_{n1}$ and $k'=k_{00}$
and set $\sigma_{n1}=1$, $\sigma_{00}=-1$ in (\ref{eq:Qmaindef}).
In the resulting formula the cluster momentum $\delta_\Lambda$-functions
enforce $\sum_i k_{0,i}=0$.  Combined with the interaction
$\delta_\Lambda$-functions this implies $k_{0,0}=-k_{n,1}$ which we have used to
simplify the final formula by changing the argument of $\FT{f}$.\qed
\end{proof}

\begin{figure}
  \centering
  \myfigure{width=0.9\textwidth}{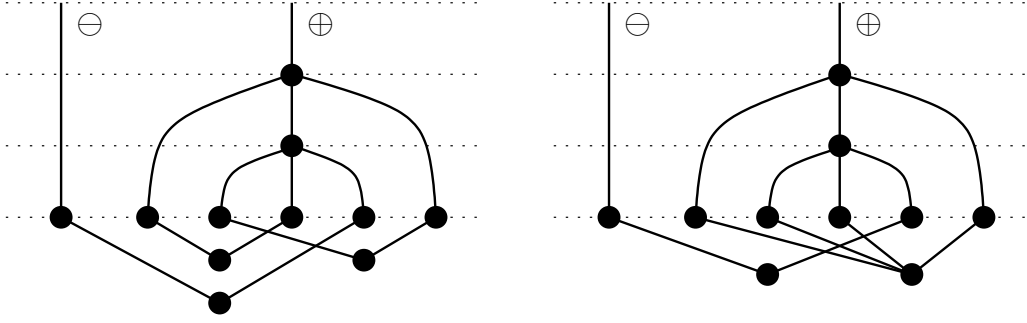}
\caption{Two diagrams representing nonzero $\mathcal{F}_2^{\rm
ampl}(S,\ell,t,\kappa)$
with interaction history $\ell=(2,1)$.
The left one has clustering $S=\set{\set{0,4},\set{1,3},\set{2,5}}$ and it
corresponds to a leading term.
The right one has $S=\set{\set{0,4},\set{1,2,3,5}}$ and corresponds to a
subleading term.
The symbol ``$\ominus$'' denotes the root of the (here trivial) minus tree,
and ``$\oplus$'' the root of the plus tree.\label{fig:maingraph}}
\end{figure}

Each choice of $n$, $S$, and $\ell$ corresponds to a unique diagram:
we take the earlier discussed ``interaction diagrams'' (as in
Fig.~\ref{fig:indexgr}),
add a ``dummy'' placeholder vertex
for each of the fields $\FT{a}_0$ at the bottom of the graph, and add a
``cluster'' vertex
for each $A\in S$ with the appropriate connections to the placeholder
vertices. Two simple examples are shown in Fig.~\ref{fig:maingraph}.
We have also added a line from the $(0,0)$-placeholder vertex to the top line, 
for reasons which will become apparent shortly.
For further use, we introduce here the concepts of ``plus'' and ``minus
tree''.  When all cluster vertices and their edges are removed, the diagram
splits into two components which are graph-theoretically trees. The left tree
(which here is a single edge connecting to the placeholder of the original
$\FT{\psi}_0(k',-1)$) is called the minus tree, and the right tree is called
the plus tree, for obvious reasons.

The integral defining the
corresponding amplitude can be constructed from
a diagram by applying the following ``Feynman rules'':
the parities of the two topmost lines are fixed to $-1$ on the left and $1$ on
the right.  The remaining parities can be computed going from top to bottom
and at each interaction vertex continuing the parity unchanged in the middle
line, and setting $-1$ on the left and $+1$ on the right.  A cluster
vertex does not affect the parities directly.
To each edge in the diagram there is
attached a momentum and they are related by Kirchhoff's rules at the
vertices: at a fusion vertex, the three momenta below need to sum to
the single momentum above, and at a cluster vertex all momenta sum
to zero.  In addition, each fusion vertex carries a factor
$-\ci \lambda \sigma \PFone$ ($\sigma$ is determined by the middle edge and
the arguments of $\PFone$ by the edges below the vertex)
and each cluster vertex a factor $C_{|A|}$ (with $\sigma$ and $k$ determined
by the edges attached to the vertex).  The total
amplitude still needs to be multiplied by
$\FT{g}(k_{n,1})^* \FT{f}(k_{n,1})$ and by the appropriate time-dependent
factor, the integrand in the last line of (\ref{eq:defFnampl}), before
integrating over $s$ and $k$.

The time-dependent factor can also be written as
\begin{align}\label{eq:defgammam}
\prod_{m=0}^n \rme^{-\ci s_m \gamma(m)},\qquad \text{where}\quad
\gamma(m) = \sum_{i=m+1}^n \Omega_i -\ci \kappa_{n-i}\, ,
\end{align}
and we recall the notation $\Omega_i =\Omega_{i-1;\ell_i}(k,\sigma)$.
The pure phase part for the time slice $m$, \itie ,
$\rme^{-\ci s_m \re \gamma(m)}$, can also be read directly from the diagram:
collect all edges which go through the time slice $m$
and for each edge $e$ add a factor $\rme^{-\ci s_m \sigma_e \omega(k_e)}$.
This follows from the following Lemma according to which inside any of the
above amplitude integrals we have
\begin{align}
\sum_{i=m+1}^{n}\Omega_{i-1;\ell_i}(k,\sigma)
= \sum_{j=1}^{2(n-m)+1} \sigma_{m,j} \omega(k_{m,j}) -
 \omega(k_{n1}) \, .
\end{align}
This yields the above mentioned factors when we follow the construction
explained earlier;
since $-\omega(k_{n1})=\sigma_{00}\omega(k_{00})$, and the corresponding edge
intersects all time slices of the diagram, also the last term comes out
correctly.

\begin{lemma}\label{th:omOmconv}
Suppose $m_0=1$, and $n\ge 0$ is given.  Then for any $\ell\in G_n$, for all
$0\le m\le n$, and with
$\sigma$ and $k$ such that $\Delta_{n,\ell}(k,\sigma;\Lambda)\ne 0$,
\begin{align}\label{eq:convertphase}
\sum_{j=1}^{2(n-m)+1} \sigma_{m,j} \omega(k_{m,j}) -
 \sum_{i=m+1}^{n} \Omega_{i-1;\ell_i}(k,\sigma)
= \sigma_{n,1} \omega(k_{n1}) \, ,
\end{align}
and
\begin{align}\label{eq:convertphase2}
\sum_{j=1}^{2(n-m)+1} \sigma_{m,j} = \sigma_{n,1}\, .
\end{align}
\end{lemma}
\begin{proof}
The proof goes via induction in $m$, starting from $m=n$ and proceeding
to smaller values.  The equation holds trivially for $m=n$, as the second sum
is not present then.  Assume that the equation holds for $m$, where
$1\le m\le n$, and to complete the induction, we need to prove that the
equation then holds for $m-1$.  Since $k,\sigma$ is consistent with
$\Delta_{n,\ell}$, we have
\begin{align}
& \sum_{j=1}^{2(n-m+1)+1} \sigma_{m-1,j} \omega(k_{m-1,j})
= \sum_{j=0}^{2} \sigma_{m-1,\ell_{m}+j} \omega(k_{m-1,\ell_{m}+j})
+ \sum_{j=1;j\ne \ell_{m}}^{2(n-m)+1} \sigma_{m,j} \omega(k_{m,j}) \, .
\end{align}
The first sum yields
$-\omega(k_{m-1,\ell_{m}})+\sigma_{m,\ell_m}  \omega(k_{m-1,\ell_{m}+1})
+ \omega(k_{m-1,\ell_{m}+2})$ which equals $\Omega_{m-1;\ell_{m}}(k,\sigma)
 + \sigma_{m,\ell_m} \omega(k_{m,\ell_m})$ .
Thus by the induction assumption,
\begin{align}
& \sum_{j=1}^{2(n-m+1)+1} \sigma_{m-1,j} \omega(k_{m-1,j}) -
 \sum_{i=m}^{n} \Omega_{i-1;\ell_i}(k,\sigma)
=  \sigma_{n,1} \omega(k_{n1}) \, ,
\end{align}
as was claimed in the Lemma.  The proof of (\ref{eq:convertphase2}) is
essentially identical,  and we will skip it.
\qed \end{proof}

\subsection{Error terms}
\label{sec:errors}

Each of the three error terms $Q^{\rm err}$ is a sum over terms of the type
\begin{align} 
\int_0^{t/\vep}\! \rmd s\,
\E\!\left[\mean{\FT{f},\FT{a}_0}^* F_s[\FT{a}_s]\right]\, ,
\end{align}
where $F_s$ contains only a finite moment of the fields $\FT{a}_s$.
We estimate it using the Schwarz inequality,
\begin{align} 
&\Bigl|\int_0^{t/\vep}\! \rmd s\,
\E\Bigl[\mean{\FT{f},\FT{a}_0}^* F_s[\FT{a}_s]\Bigr]\Bigr|^2
\le \Bigl(\int_0^{t/\vep}\! \rmd s\,
\E\Bigl[\bigl|\mean{\FT{f},\FT{a}_0}\bigr| \,
\bigl|F_s[\FT{a}_s]\bigr|\Bigr]\Bigr)^2
\nonumber \\ & \quad
\le \frac{t}{\vep} \E[|\mean{\FT{f},\FT{a}_0}|^2]\,
\int_0^{t/\vep}\! \rmd s\,
\E\bigl[|F_s[\FT{a}_s]|^2\bigr]
 \, .
\end{align}
Since
$\E[|\mean{\FT{f},\FT{a}_0}|^2] =
\int_{\Lambda^*} \rmd k\, |\FT{f}(k)|^2 W^{\lambda}_\Lambda(k)$,
the term $\E[|\mean{\FT{f},\FT{a}_0}|^2]$ remains uniformly bounded. 
Thus
$\limsup\limits_{\Lambda\to\infty}
\E[|\mean{\FT{f},\FT{a}_0}|^2] t/\vep\le
 \lambda^{-2} t c'_0 \norm{f}_2^2$, and
we need to aim at estimates for
$\sup\limits_{0\le s\le t \lambda^{-2}}\E\bigl[|F_s[\FT{a}_s]|^2\bigr]$
which decay faster than $\lambda^{4}$ in order to get a vanishing bound.

Although the Gibbs measure is not stationary with respect to $\FT{a}_t$, it
{\em is} stationary with respect to $\FT{\psi}_t$.  The
non-stationarity manifests itself only via an additional phase factor:
\begin{align}\label{eq:anonstat}
& \E\Bigl[\prod_{i\in I} \FT{a}_t(k_i,\sigma_i)\Bigr]
= \prod_{i\in I} \rme^{\ci t \sigma_i \omla(k_i)}
 \E\Bigl[\prod_{i\in I} \FT{a}_0(k_i,\sigma_i)\Bigr]
\nonumber \\ & \quad
= \rme^{\ci t \lambda R_0 \sum_i\sigma_i}
\prod_{i\in I} \rme^{\ci t \sigma_i  \omega(k_i)}
 \E\Bigl[\prod_{i\in I} \FT{a}_0(k_i,\sigma_i)\Bigr]\, .
\end{align}
The extra phase factor can always expressed in terms of the previously used
$\Omega$-factors, employing Lemma \ref{th:omOmconv}.
Applying the Lemma for $m=0$
implies that the phase factor generated by the non-stationarity of
$\FT{a}$ can be resolved by employing
\begin{align}\label{eq:s0phase}
& \prod_{j=1}^{m_0+2 n} \FT{a}_s(k_{0,j},\sigma_{0,j})
 = \rme^{\ci  s \sigma_{n,1} \omla(k_{n,1})}
 \prod_{i=1}^{n} \rme^{\ci s \Omega_{i-1;\ell_i}(k,\sigma)}
\prod_{j=1}^{m_0+2 n} \FT{\psi}_s(k_{0,j},\sigma_{0,j})\, ,
\end{align}
which will always hold inside the relevant integrals.

The following lemma gives a recipe
how the two simplex time-integrations resulting from the Schwarz inequality
can be represented in terms of a single simplex time-integration.
We begin by introducing the concept of interlacing of two sequences.
\begin{definition}
Let $n,n'\ge 0$ be integers.
A map $J:I_{n+n'}\to \set{\pm 1}$ \defem{interlaces} $(n,n')$,
if $|J^{\gets}(\set{+1})|=n$ and $|J^{\gets}(\set{-1})|=n'$.  For any such
$J$, we define further the two maps $J_\pm:I_{0,n+n'}\to \N_0$ by setting for
$\sigma\in\set{\pm 1}$, $i\in I_{0,n+n'}$,
\begin{align} 
  J_\sigma(i;J) = \sum_{j=1}^i \1(J(j)=\sigma)\, .
\end{align}
\end{definition}
Thus $J_\sigma(0;J)=0$ and else 
$J_\sigma(i;J)= |J^{\gets}(\set{\sigma})\cap I_i|$.
In addition, as $J$ interlaces $(n',n)$,
clearly $J_+:I_{0,n+n'}\to I_{0,n}$ and
$J_-:I_{0,n+n'}\to I_{0,n'}$ and both maps are increasing and onto.
We claim that with these definitions the following representation
Lemma holds, saving the proof of the Lemma until the end of this section.
\begin{lemma}\label{th:recombinationlemma}
Let $t>0$, $n,n'\ge 0$, and suppose $\gamma^+_i,\gamma^-_j\in\C$
are given for $i\in I_{0,n}$ and $j\in I_{0,n'}$. Then
\begin{align}\label{eq:recombinationlemma}
& \int_{(\R_+)^{I_{0,n}}}\!\rmd s \,  \delta\Bigl(t-\sum_{i=0}^{n} s_i\Bigr)
 \prod_{i=0}^n \rme^{-\ci s_i \gamma^+_i}
\times
\int_{(\R_+)^{I_{0,n'}}}\!\rmd s' \,  \delta\Bigl(t-\sum_{i=0}^{n'} s'_i\Bigr)
 \prod_{i=0}^{n'} \rme^{-\ci s'_i \gamma^-_i}
 \nonumber \\ & \quad
= \sum_{J\text{ interlaces }(n,n')}
 \int_{(\R_+)^{I_{0,n+n'}}}\!\rmd r \,
 \delta\Bigl(t-\sum_{i=0}^{n+n'} r_i\Bigr)
 \prod_{i=0}^{n+n'} \rme^{-\ci r_i (\gamma^+_{J_+(i;J)}+\gamma^-_{J_-(i;J)})} .
\end{align}
\end{lemma}
The Lemma might appear complicated, but it can be understood in terms of
interlacing of the two sets of time slices.  The symbolic representation in
Fig.~\ref{fig:interlacing} illustrates this point and serves as an example
of the above definitions. 
\begin{figure}
  \centering
\medskip
  \myfigure{width=0.9\textwidth}{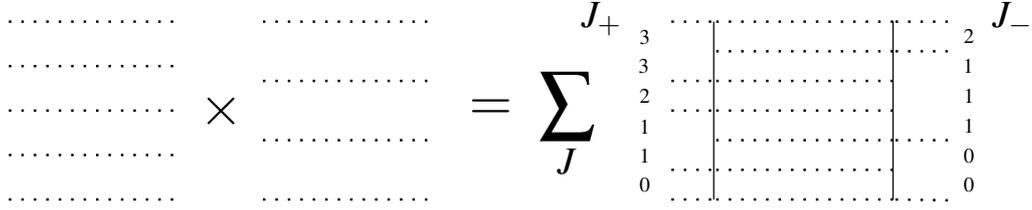}
\medskip
\caption{A symbolic representation of Lemma \ref{th:recombinationlemma} for
  $n=3$, $n'=2$, in terms of time slices. 
On the right hand side, only one example of an interlacing
($J=(+1,-1,+1,+1,-1)$) is shown. 
\label{fig:interlacing}}
\end{figure}

We recall that the $\delta$-functions in the above formula are a shorthand
notation for restricting the integration to the standard simplex of size $t$. 
The exact definition is obtained by choosing an arbitrary index $i$ and
``integrating out'' the delta function with respect to $s_i$. It is an easy
exercise to show that for the above exponentially bounded functions, an
equivalent definition is obtained by replacing the $\delta$-function by a
Gaussian approximation and then taking the variance of the Gaussian
distribution to zero.  (The latter property combined with Fubini's theorem
allows for free manipulation of the order of integration.)

Using the above observations, we can derive diagrammatic representations of the
expectation values $\E\bigl[|F_s[\FT{a}_s]|^2\bigr]$ very similar to what was
described in Section
\ref{sec:maindiag}.
We consider only the case of $\mathcal{A}_n$ in detail.  The treatment of the
remaining error terms
is very similar, and we merely quote the results in the forthcoming sections.

Let $\mathcal{I}'_{n} =\mathcal{I}_{n} \cup \set{(n,1)}=
\defset{(i,j)}{0\le i\le n, 1\le j \le m_{n-i}}$.
Then by (\ref{eq:s0phase}) and (\ref{eq:phase2}) we can write
\begin{align}\label{eq:gAnprod}
& \mean{\FT{g},\mathcal{A}_n(r,t,\cdot,1,\kappa)[\FT{a}_r]}
 = \sum_{\ell \in G_n}
\sum_{\sigma\in \set{\pm 1}^{\mathcal{I}'_{n}}}
\int_{(\Lambda^*)^{\mathcal{I}'_{n}}} \!\rmd k\,
\Delta_{n,\ell}(k,\sigma;\Lambda)
\1(\sigma_{n,1}=1) 
 \nonumber \\ & \quad \times
 \FT{g}(k_{n,1})^* \rme^{\ci r \omla(k_{n,1})} \prod_{j=1}^{m_0+2 n}
\FT{\psi}_r(k_{0,j},\sigma_{0,j})
\prod_{i=1}^{n} \Bigl[-\ci \lambda \sigma_{i,\ell_{i}}
 \PFone(k_{i-1;\ell_i})
 \Bigr]
 \nonumber \\ & \quad \times
\int_{(\R_+)^{I_{n}}}\!\rmd s \,  \delta\Bigl(t-r-\sum_{m=1}^{n} s_m\Bigr)
\prod_{m=1}^{n} \rme^{-\ci s_m \gamma^+_m}
 \, ,
\end{align}
where $\gamma^+_n=-\ci \kappa_0$ and for $1\le m\le n-1$,
\begin{align}
\gamma^+_m = \sum_{i=m+1}^n \Omega_i -\ci \kappa_{n-m}\, .
\end{align}
Now we can apply Lemma \ref{th:recombinationlemma} to study the
expectation of the square.
However, before taking the expectation value, we make a change of variables
$\sigma'_{i,j}=-\sigma_{i,2 (n-i+1)-j}$,
$k'_{i,j}=-k_{i,2 (n-i+1)-j}$, and $\ell'_i=2(n-i+1)-\ell_i$
in the complex conjugate (\itie , we swap the signs and invert the order on each
time slice).
We also define $I=I_{2 m_n}=I_{2(2 n+1)}$ to give labels to the fields
$\FT{\psi}_r$:
we denote $K=(k'_{0,\cdot},k_{0,\cdot})\in (\T^d)^{I}$ and
$o=(\sigma'_{0,\cdot},\sigma_{0,\cdot})\in \set{\pm 1}^{I}$, and thus, for
instance, $K_{m_n+1}=k_{0,1}$. Applying Lemma \ref{th:recombinationlemma},
Proposition \ref{th:PFcorr}, and the stationarity of $\FT{\psi}_s$, we obtain
\begin{align}\label{eq:def Aampl}
 & \E\Bigl[|\mean{\FT{g},
\mathcal{A}_n(s,t/\vep,\cdot,1,\kappa)[\FT{a}_s]}|^2\Bigr]
\nonumber \\ & \quad
= \sum_{J\text{ interlaces }(n-1,n-1)}
 \sum_{\ell,\ell' \in G_n}  \sum_{S\in \pi(I)}
\mathcal{A}_n^{\rm ampl}(S,J,\ell,\ell',t/\vep-s,\kappa)\, ,
\end{align}
where the amplitudes are explicitly
\begin{align} 
& \mathcal{A}_n^{\rm ampl}
=  (-\lambda^{2})^{n} \sum_{\sigma,\sigma'\in \set{\pm 1}^{\mathcal{I}'_{n}} }
\int_{(\Lambda^*)^{\mathcal{I}'_{n}}}  \!\rmd k\,
\int_{(\Lambda^*)^{\mathcal{I}'_{n}}}  \!\rmd k'\,
\Delta_{n,\ell}(k,\sigma;\Lambda) \Delta_{n,\ell'}(k',\sigma';\Lambda)
 \nonumber \\ & \quad \times
\prod_{A\in S}\Bigl[ \delta_\Lambda\!\Bigl(\sum_{i\in A} K_i\Bigr)
  C_{|A|}(K_A,o_A;\lambda,\Lambda) \Bigr]
\prod_{i=1}^{n} \Bigl[ \sigma_{i,\ell_{i}}
 \PFone(k_{i-1;\ell_i}) \sigma'_{i,\ell'_{i}} \PFone(-k'_{i-1;\ell'_i})
 \Bigr]
 \nonumber \\ & \quad \times
\1(\sigma_{n,1}=1)\1(\sigma'_{n,1}=-1) \FT{g}(k_{n,1})^* \FT{g}(-k'_{n,1})
\rme^{\ci s (\omega(k_{n,1})-\omega(k'_{n,1}))}
 \nonumber \\ & \quad \times
 \int_{(\R_+)^{I_{0,2n-2}}}\!\rmd r \,
 \delta\Bigl(\frac{t}{\vep}-s-\sum_{i=0}^{2 n-2} r'_i\Bigr)
 \prod_{i=0}^{2n-2} \rme^{-\ci r'_i
   \left(\gamma^+_{J_+(i;J)+1}+\gamma^-_{J_-(i;J)+1}\right)} .
\end{align}
In this formula, $\gamma^+_m$ are defined as before, and we set
\begin{align} 
\gamma^-_m & = \sum_{i=m+1}^{n}
\Omega_{i-1;\ell'_i}(k',\sigma')-\ci \kappa_{n-m}
\, ,
\end{align}
which can be checked to yield the correct factors by using
$\Omega(-(k_3,k_2,k_1),-\sigma)=-\Omega((k_1,k_2,k_3),\sigma)$.

The cluster $\delta$-functions imply that $\sum_{i\in I} K_i = 0$.
Applying the interaction $\delta$-func\-tions iteratively in the
direction of time then shows that the integrand is zero unless
$k_{n,1}+k'_{n,1}=0$ (modulo 1).  Therefore, $\omega(k_{n,1})=\omega(k'_{n,1})$,
and the amplitude
depends on $s$ and $t/\vep$ only via their difference $t/\vep-s$,
as implied by the notation in (\ref{eq:def Aampl}).  The final, somewhat
simplified expression, for the amplitude function is thus
\begin{align}\label{eq:Aamplsimp}
& \mathcal{A}_n^{\rm ampl}(S,J,\ell,\ell',s,\kappa)
=  (-\lambda^{2})^{n} \sum_{\sigma,\sigma'\in \set{\pm 1}^{\mathcal{I}'_{n}} }
\int_{(\Lambda^*)^{\mathcal{I}'_{n}}}  \!\rmd k\,
\int_{(\Lambda^*)^{\mathcal{I}'_{n}}}  \!\rmd k'\,
 \nonumber \\ & \quad \times
\Delta_{n,\ell}(k,\sigma;\Lambda) \Delta_{n,\ell'}(k',\sigma';\Lambda)
\prod_{A\in S}\Bigl[ \delta_\Lambda\!\Bigl(\sum_{i\in A} K_i\Bigr)
  C_{|A|}(o_A,K_A;\lambda,\Lambda) \Bigr]
 \nonumber \\ & \quad \times
\prod_{i=1}^{n} \Bigl[ \sigma_{i,\ell_{i}}
 \PFone(k_{i-1;\ell_i}) \sigma'_{i,\ell'_{i}} \PFone(-k'_{i-1;\ell'_i})
 \Bigr]
\1(\sigma_{n,1}=1)\1(\sigma'_{n,1}=-1) |\FT{g}(k_{n,1})|^2
 \nonumber \\ & \quad \times
 \int_{(\R_+)^{I_{2,2n}}}\!\rmd r \,
\delta\Bigl(s-\sum_{i=2}^{2 n} r_i\Bigr)
 \prod_{i=2}^{2n} \rme^{-\ci r_i \gamma({i;J})}\, ,
\end{align}
where, for $i=2,3,\ldots,2n$, we have
\begin{align}\label{eq:giJ}
& \gamma({i;J})
= \sum_{j=m+1}^{n} \Omega_{j-1;\ell_j}(k,\sigma)
+ \sum_{j=m'+1}^{n} \Omega_{j-1;\ell'_j}(k',\sigma')
-\ci (\kappa_{n-m}+\kappa_{n-m'})
\nonumber \\ &
=  \sum_{j=1}^{2(n-m)+1}\!\! \sigma_{m,j} \omega(k_{m,j})
+ \sum_{j=1}^{2(n-m')+1}\!\! \sigma'_{m',j} \omega(k'_{m',j})
-\ci (\kappa_{n-m}+\kappa_{n-m'})\, ,
\end{align}
with $m=m(i)=J_+(i-2;J)+1$ and $m'=m'(i)=J_-(i-2;J)+1$.
In particular, $\gamma(2 n;J)=\gamma^+_n + \gamma^-_n =-\ci 2\kappa_0$.

We can now describe the integral (\ref{eq:Aamplsimp}) using the earlier
defined diagrammatic scheme.
To make the identification more direct, we have shifted the time-indices
upwards by two:
the idea is that the first two time slices have zero length, \itie , they are
\defem{amputated}.
Formally, we could write the time-integral as
\begin{align}
  \int_{(\R_+)^{I_{0,2n}}}\!\rmd r \,
\delta\Bigl(s-\sum_{i=0}^{2 n} r_i\Bigr)
\delta(r_0) \delta(r_1)
 \prod_{i=0}^{2n} \rme^{-\ci r_i \gamma(i;J)}\, .
\end{align}
Clearly, the result is independent of how we define $\gamma(0;J)$ and
$\gamma(1;J)$.
To make the identification between an amputated amplitude
and the diagram unique, we arbitrarily require that in an amputated diagram
the first fusion always happens in the minus tree and the second fusion in the
plus tree.
The construction of the phase factor of the
time-integrand is then done using the same rules as before: for each time
slice $i$,
we collect all edges which go through the time slice
and for each edge $e$ add a factor $\rme^{-\ci r_i \sigma_e \omega(k_e)}$.
Under the above amputation condition, we arrive this way to the integrand in
(\ref{eq:Aamplsimp}).
Compared to the Feynman rules explained for the main term, we have only one
additional rule here:
in the \defem{minus} tree, the sign inside the cutoff-function is swapped,
\itie , there we use a factor $-\ci \lambda \sigma' \PFone(-k')$. 
Otherwise, the Feynman rules are identical, apart from the overall
testfunction factor
which is $|\FT{g}(k_{n,1})|^2$ here.
We have illustrated these definitions in Fig.~\ref{fig:ampdiag}.

\begin{figure}
  \centering
  \myfigure{height=0.27\textheight}{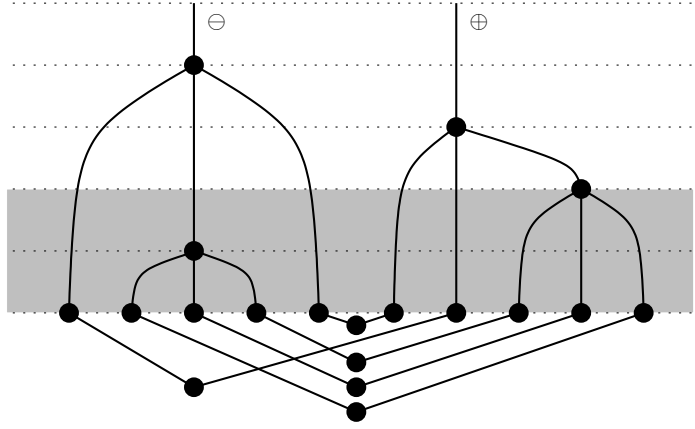}
\caption{An amputated diagram representing a nonzero
$\mathcal{A}_2^{\rm ampl}(S,J,\ell,\ell',s,\kappa)$
with $\ell'=(2,1)$, $\ell=(3,1)$, $J=(+1,-1)$, and using a pairing
$S=\set{\set{1,7},\set{2,10},\set{3,9},\set{4,8},\set{5,6}}$.
The shading on the first two time slices is used to denote the fact that these
time slices have zero length,
as explained in the text.
As before, the symbols ``$\oplus$'' and  ``$\ominus$'' denote the roots of the
plus and minus trees, respectively.\label{fig:ampdiag}}
\end{figure}

To complete the above derivation, we still need to prove the time-simplex
Lemma. 
\begin{proofof}{Lemma \ref{th:recombinationlemma}}
The Lemma is based on rearrangement of the time-integrations
by iteratively splitting one of them into two independent
parts.
The splitting will depend on the relative order of the times accumulated
from $s$ and $s'$,
which is captured by the sum over $J$ on the right hand side.
The value of $J_+(i;J)$ yields the index for the phase factor $\gamma^+$
which is ``active'' at the new time slice obtained after the splitting.
The proof below will be given mainly to show that the above definitions yield
a correct description of the result.

Suppose first that $n=0$.  Then the first factor is $\rme^{-\ci t \gamma^+_0}$. 
On the other hand,
the only admissible $J$ is then $J(i)=-1$ for all $i\in I_{n'}$, and thus
$J_+(i;J)=0$ and $J_-(i;J)=i$ for all $i\in I_{0,n'}$.  Therefore,
(\ref{eq:recombinationlemma}) holds by inspection.  A symmetrical argument
applies
to the case $n'=0$.

\enlargethispage*{1em}
Assume thus that $n'\ge 1$, and we will prove the rest by induction in $n$.
The initial case $n=0$ was checked to hold above.  Assume then that
(\ref{eq:recombinationlemma}) holds for all $n\le N$, with $N\ge 0$,
and consider
the case $n=N+1\ge 1$.
It is clear that both sides of
(\ref{eq:recombinationlemma}) are continuous in $\gamma^+_{N+1}$, and thus it
suffices to prove it assuming $\gamma^+_{N+1}\ne \gamma^+_{N}$.
Let us first concentrate on the first factor.
We change the integration variable from $s_{N}$ to $u=s_{N}+s_{N+1}$.
This shows that
\begin{align} 
& \int_{0}^\infty\! \rmd s_{N}\, \int_{0}^\infty\! \rmd s_{N+1}\,
 \delta\Bigl(t-\sum_{i=0}^{N+1} s_i\Bigr)
 \rme^{-\ci s_N \gamma^+_N -\ci s_{N+1} \gamma^+_{N+1}}
 \nonumber \\ & \quad
 = \int_{0}^\infty\! \rmd s_{N+1}\, \int_{s_{N+1}}^\infty\! \rmd u\,
 \delta\Bigl(t-u-\sum_{i=0}^{N-1} s_i\Bigr)
 \rme^{-\ci (u-s_{N+1}) \gamma^+_N -\ci s_{N+1} \gamma^+_{N+1}}
 \nonumber \\ & \quad
 = \int_{0}^\infty\! \rmd u\,
 \delta\Bigl(t-u-\sum_{i=0}^{N-1} s_i\Bigr)  \rme^{-\ci u \gamma^+_{N}}
 \int_{0}^u\! \rmd s_{N+1}\, \rme^{-\ci s_{N+1} ( \gamma^+_{N+1}- \gamma^+_N )
 }
 \nonumber \\ & \quad
 =  \frac{\ci}{\gamma^+_{N+1}- \gamma^+_N}
 \int_{0}^\infty\! \rmd u\,
 \delta\Bigl(t-u-\sum_{i=0}^{N-1} s_i\Bigr)
\Bigl(\rme^{-\ci u \gamma^+_{N+1}}
-\rme^{-\ci u \gamma^+_{N}} \Bigr) .
\end{align}
The induction assumption can be applied to both terms separately, which
proves that (\ref{eq:recombinationlemma}) is equal to
\begin{align} 
& \frac{\ci}{\gamma^+_{N+1}- \gamma^+_N}
 \sum_{J\text{ interlaces }(N,n')}
 \int_{(\R_+)^{I_{0,N+n'}}}\!\rmd r \,
 \delta\Bigl(t-\sum_{i=0}^{N+n'} r_i\Bigr)
 \nonumber \\ & \qquad \times
 \Bigl( \prod_{i=0}^{N+n'}
\left. \rme^{-\ci r_i (\gamma^+_{J_+(i;J)}+\gamma^-_{J_-(i;J)})}
 \right|_{\gamma^+_{N}\to\gamma^+_{N+1}} -
 \prod_{i=0}^{N+n'} \rme^{-\ci r_i
   (\gamma^+_{J_+(i;J)}+\gamma^-_{J_-(i;J)})}
 \Bigr) .
\end{align}
For a fixed $J$, let $j_0=\min \defset{i\in I_{0,N+n'}}{J_+(i;J)=N}$, \itie ,
$j_0$ denotes the last appearance of $+1$ in $J$.
Then $N\le j_0\le N+n'$.
The difference in the brackets can then be expressed as
\begin{align} 
& \prod_{i=0}^{N+n'} \rme^{-\ci r_i\gamma^-_{J_-(i;J)}}
 \prod_{i=0}^{j_0-1} \rme^{-\ci r_i \gamma^+_{J_+(i;J)}}
 \Bigl( \rme^{-\ci  \gamma^+_{N+1} \sum_{i=j_0}^{N+n'} r_i} -
\rme^{-\ci   \gamma^+_{N}\sum_{i=j_0}^{N+n'} r_i} \Bigr)
 \nonumber \\ & \quad
= (-i) (\gamma^+_{N+1}- \gamma^+_N)
 \prod_{i=0}^{N+n'} \rme^{-\ci r_i\gamma^-_{J_-(i;J)}}
 \prod_{i=0}^{j_0-1} \rme^{-\ci r_i \gamma^+_{J_+(i;J)}}
 \nonumber \\ & \qquad \times
\left. \int_0^{u}\! \rmd s\,
 \rme^{-\ci s \gamma^+_{N}}
 \rme^{-\ci (u-s)\gamma^+_{N+1}}\right|_{u=\sum_{i=j_0}^{N+n'}r_i}\, .
\end{align}
Set $S_\ell =\sum_{i=j_0}^{\ell-1}r_i$.
The final integral is split according to the position of $s$ in
the sequence $(S_\ell)_{\ell=j_0,\ldots,N+n'+1}$.  This yields
\begin{align} 
\sum_{\ell =j_0}^{N+n'} \int_0^{\infty}\! \rmd s\,
\1(S_{\ell}\le s \le S_{\ell+1})
 \rme^{-\ci s \gamma^+_{N}}  \rme^{-\ci (S_{N+n'+1}-s)\gamma^+_{N+1}}\, .
\end{align}

Given a map $J$ and $\ell\in \set{j_0(J),N+n'}$, we define a map
$J'=J'_{\ell,J}:I_{N+1+n'}\to \set{\pm 1}$ by the rule
\begin{align} 
 J'(i) = \begin{cases}
 J(i),& \text{if }i\le \ell,\\
 +1,& \text{if }i=\ell+1,\\
 -1,& \text{if }i> \ell+1.
\end{cases}
\end{align}
Obviously, $J'$ interlaces $(N+1,n')$, and
the maps $J_\pm(\cdot;J')$
then satisfy $J_\pm (i;J')=J_\pm(i;J)$ for $i\le \ell$
and $J_+ (i;J')=N+1$, $J_- (i;J')=J_-(i-1;J)$ for
$i> \ell$.  Conversely, if $J''$ is an
arbitrary map interlacing $(N+1,n')$ then there are unique $\ell$ and $J$
such that $J''=J'_{\ell,J}$, determined by the choices $\ell=j_0(J'')$,
and $J$ obtained from $J''$ by canceling $\ell$.
Therefore,
\begin{align} 
 \sum_{J'\text{ interlaces }(N+1,n')} F(J')
 = \sum_{J\text{ interlaces }(N,n')} \sum_{\ell=j_0(J)}^{N+n'} F(J'_{\ell,J})\, .
\end{align}

Thus we only need to prove that the remaining integrals are equal, \itie ,
that the integral on the right hand side of (\ref{eq:recombinationlemma}) for
$J\to J'=J'_{\ell,J}$ and $n\to N+1$, is equal to
\begin{align}\label{eq:nearlythere}
& \int_{(\R_+)^{I_{0,N+n'}}}\!\rmd r \,
\delta\Bigl(t-\sum_{i=0}^{N+n'} r_i\Bigr)
 \prod_{i=0}^{N+n'} \rme^{-\ci r_i\gamma^-_{J_-(i;J)}}
 \prod_{i=0}^{j_0-1} \rme^{-\ci r_i \gamma^+_{J_+(i;J)}}
\nonumber \\ & \quad  \times
 \int_0^{\infty}\! \rmd s\,
\1(S_{\ell}\le s \le S_{\ell+1})
 \rme^{-\ci s \gamma^+_{N}}  \rme^{-\ci (S_{N+n'+1}-s)\gamma^+_{N+1}}\, .
\end{align}
To see this, let us change the integration variables
$(r_i,s)_i$ to $(r'_j)_j$ by using  $r'_j=r_j$ for $j<\ell$,
$r'_j=r_{j-1}$ for $\ell<j\le N+1+n'$, and $r'_\ell = s - S_{\ell}$,
$r'_{\ell+1}= S_{\ell} + r_\ell -s$.  Since $S_{\ell}$ does not depend on
$r_\ell$, the Jacobian can straightforwardly be checked to be equal to one,
and the effect of $\1(S_{\ell}\le s \le S_{\ell+1})$ is simply to restrict the
integration region to $r'\in (\R_+)^{I_{0,N+1+n'}}$.
On the other hand, $r_\ell= r'_\ell+r'_{\ell+1}$ and thus
\begin{align} 
\sum_{i=0}^{N+n'} r_i \gamma^-_{J_-(i;J)}
= \sum_{i=0}^{N+1+n'} r'_i \gamma^-_{J_-(i;J')}
\end{align}
and, as $s=\sum_{i=j_0}^{\ell} r'_i$, then also
\begin{align} 
s \gamma^+_N + (S_{N+n'+1}-s)\gamma^+_{N+1}
= \sum_{i=j_0}^{\ell} r'_i \gamma^+_{N}
+ \sum_{i=\ell+1}^{N+1+n'} r'_i \gamma^+_{N+1}
= \sum_{i=j_0}^{N+1+n'} r'_i \gamma^+_{J_+(i;J')} \,  .
\end{align}
Thus
relabeling of the integration variables now proves that (\ref{eq:nearlythere})
is equal to the
integral on the right hand side of (\ref{eq:recombinationlemma}).  This
finishes the proof of the Lemma.\qed
\end{proofof}

\section{Resolution of the momentum constraints}
\label{sec:momdeltas}

One important element for our estimates is disentangling the
complicated momentum dependencies into more manageable form which allows
iteration of a finite collection of bounds.  For this we have to carefully
assign which of the momenta are freely integrated over, and which are used to
integrate out the $\delta$-functions and thus attain a linear dependence on
the free integration variables.  We begin from a diagram as described in the
previous section, which can represent either a main term or an amputated
amplitude.  To make full use of graph invariants,
we then add one more $\delta$-function to the integrand:
we multiply it by a factor
\begin{align} 
1=\int_{\Lambda^*}\! \rmd k_{e_0}\,
 \delta_{\Lambda}(k_{e_0}-k-k')  \, ,
\end{align}
where $k$ is the outgoing momentum at the root of the plus tree, and $k'$ at
the root of the minus tree.
(Using the notations of the previous section, we thus have $k=k_{n,1}$,
$k'=k_{0,0}$ for a main term,
and $k=k_{n,1}$, $k'=k'_{n,1}$ for an amputated term.)
The factor will facilitate the analysis of the momentum constraints, as
without it there would be one free integration variable which is not
associated with a loop in the corresponding graph.  This would lead to
unnecessary repetition in the oncoming proofs in form of spurious ``special
cases'', which can now be avoided at the cost of introducing a ``spurious edge''
into the graph.  The additional $\delta$-function can then be accounted for by
introducing two additional vertices and one extra edge to the graph: one
``fusion vertex'' which connects the two edges related to $k$ and $k'$, and
one vertex to the top of the graph, so that $e_0$ is the edge connecting the
two new vertices.
(See Fig.~\ref{fig:graphex1} for an illustration.)

\begin{figure}[t]
  \centering
  \myfigure{height=0.3\textheight}{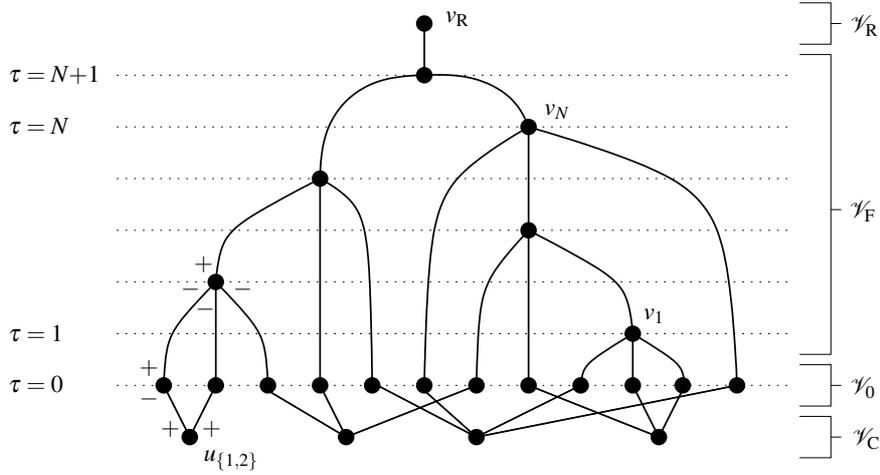}
  \caption{An example of a momentum graph $\mathcal{G}$ constructed by the
algorithm
  described in Section \ref{sec:momdeltas}.  The vertices of the graph are
  denoted by ``bullets'', connected by the edges, the bending of which serves
  only an illustrative purpose.   The graph corresponds to a case
  with $n=3$, $n'=2$ interactions,
  and  a cluster decomposition
  $S=\set{\set{1,2},\set{3,4,7},\set{5,6,9,12},\set{8,10,11}}$,
  using an enumeration of the vertices in $\Dverts$ from left to right
  in the figure.
  Examples of the notations used for the different vertices
  are also given, in particular, we have denoted the division of
  the vertices into the four disjoint sets on the right.  The horizontal
  dotted lines divide the graph into the time slices, and the labels on the
left show some values of the
  function $\tau$.   Finally, the signs near the vertices on the bottom left
  half of the graph denote how the edges at the vertex
  are divided between the sets $\edges_\pm(v)$.
  \label{fig:graphex1}}
\end{figure}

The diagram obtained this way is called the \defem{momentum graph}
associated with the original amplitude,
and it consist of $\mathcal{G}=(\verts,\edges)$, where $\verts$ collects the
vertices and $\edges$ the edges of the graph.  There is also additional
structure arising from the construction of the graph and related to the
different roles the vertices play.  In particular, the fusion vertices have a
natural time-ordering determined by $\ell$, $\ell'$ and $J$, and encoded in
the way we have drawn the diagrams.

Let us first summarize the construction of the momentum graph, and introduce
related notation for later use.  The total time $t$
and the partial time-integration variables $\kappa_i$ are parameters which do
not affect the momentum structure, and we assume them to be fixed to some
allowed value in the following.
The amplitude depends on
$S$, the cluster partitioning of the initial fields, and
the number of interactions in the plus and minus trees, $n$ and $n'$,
as well as the related collisions histories determined by $\ell$, $\ell'$ and
$J$.
(For a main term graph $n'=0$, and $\ell'$ and $J$ are then not relevant.)
We let $N=n+n'$ denote the total number of interactions, and consider here only
the non-trivial case $N\ge 1$.
Given these parameters, we can construct the momentum graph $\mathcal{G}$
by the following iteration procedure.

At each iteration step,
for the given previous graph $\mathcal{G}=(\verts ,\edges )$ with
$\verts\ne \emptyset$
we construct a new graph $\mathcal{G}'=(\verts' ,\edges' )$
by either ``attaching a new edge'' to some given
$v\in\verts$, or by ``joining the vertices'' $v,v'\in\verts$, $v'\ne v$.
Explicitly, in the first case when a new edge is attached to $v\in\verts$, we
choose a new vertex 
label $u\not \in \verts$, and define $\verts'=\verts\cup \set{u}$
and $\edges'=\edges\cup\set{\set{v,u}}$.   In our iteration scheme,
the new vertex label $u$ will be a dummy variable, which we can choose to
relabel later.
In the second case, when two existing vertices $v,v'$ are joined, we define
$\verts'=\verts$ and $\edges'=\edges\cup\set{\set{v,v'}}$.
The iterative construction will thus imply a natural order for the edges of
the graph:
\defem{we will say in the following that $e< f$ if the edge $e$ is created
before $f$.}
This defines then a complete order $e\le f$ on $\edges$.  We will use the
creation
order also to label the edges: $e_i$ is the edge which is created in the $i$:th
iteration step.

We begin with $\mathcal{G}\upn{0}=(\verts\upn{0} ,\edges\upn{0} )$
where $\verts\upn{0}=\set{\rootv,v_{N+1}}$
and $\edges\upn{0}=\set{e_0}$, $e_0=\set{\rootv,v_{N+1}}$.
We next go through the list
$i=1,2,\ldots,N_k$, where $N_k$ is the total number of $k,k'$-integrals.
At each iteration step, we add the corresponding edge to $\edges$.
In the first two iteration steps we attach two new
edges, labeled $e_1$ and $e_2$, to $v_{N+1}$.
The edge $e_1$ begins the minus tree associated with the $k'$-integrals,
and the edge $e_2$ begins the plus tree associated with the $k$-integrals.
If the last interaction (as determined by $J$) is in the minus tree
we next choose $e_1$, and otherwise choose $e_2$, and relabel
the unlabeled vertex in it by $v_N$.

The next three iteration steps are to attach three new edges to $v_N$,
left to right in the picture (and thus having parities $-1$, $\sigma$, $1$).
The interaction history is determined by $\ell,\ell',J$, and is used
(backwards in time) for
choosing a unique edge with an unlabeled vertex in the following steps.
We pick the appropriate edge and label the (unique) unlabeled
vertex in this edge as $v_{N-1}$, and the next three iteration steps consist of
attaching three new
edges to $v_{N-1}$, left to right.  This procedure of attaching triplets of
edges is
iterated altogether $N$ times and results in a tree
starting from $\rootv$.

The resulting vertex set is composed of $N+2$ labeled and of $2 N+2$
unlabeled vertices.  The labeled vertices belong either to
$\Rverts=\set{\rootv}$ or to $\Fverts=\set{v_j}_{j=1,\ldots,N+1}$,
which we call the {\em root\/} and the {\em fusion
vertex set,\/} respectively.  The term \defem{interaction vertex} refers to an 
interaction vertex in the original diagram.  The set of interaction vertices
is thus $\Iverts=\Fverts\setminus\set{v_{N+1}}$.
We collect the remaining vertices
to $\Dverts$, and call this the \defem{initial time vertex} set.
Each $v\in \Dverts$ is associated with a definite
$\psi$-factor in the initial time expectation
value, and $S$ can thus be identified with a unique partition of
$\Dverts$ into clusters.  For every cluster $A$ in $S$,
we associate an independent label $u_A$. The set
$\Cverts=\set{u_A}_{A\in S}$ is called the \defem{cluster vertex} set.
The final graph $\graph$ is defined to have a vertex set
$\verts = \Rverts\cup \Fverts \cup \Dverts \cup \Cverts$.
In the final iteration steps, we add edges by going through the initial time
vertices, left to right,
and for each vertex $v$ joining it to $u_{A(v)}$ where
$A(v)\in S$ is the unique cluster containing $v$.

This yields an unoriented graph 
$\graph=\graph(S,J,n,\ell,n',\ell')=(\verts,\edges)$
representing the corresponding amplitude.
The vertices have a natural time-order given by
$\tau:\verts \to [0,N+2]$, which we define by setting for $v\in \verts$
\begin{align} 
  \tau(v) = \begin{cases} N+2,& \text{if } v\in \Rverts, \\
    j,& \text{if there is } j\in\set{1,\ldots,N+1} \text{ such that } v=v_j, \\
    0,& \text{if } v\in \Dverts \cup \Cverts .
\end{cases}
\end{align}
We extend the time-ordering to the edges by defining
$\taup(e)=\max\defset{\tau(v)}{v\in e}$ for $e\in \edges$.  It is obvious from
our construction
that $e\le f$ implies $\taup(e)\ge \taup(f)$.

For any $v\in \verts$, let
$\edges(v)=\defset{e\in\edges}{v\in e}$ denote the set of edges attached to
$v$. To each edge $e\in\edges$ we have associated an integration over a
variable $k_e$.  These variables are not independent since, apart from the
root vertex,
each vertex has a $\delta$-function associated to it.
Explicitly, for $v\not\in \Rverts$, there is a factor (Kirchhoff rule)
\begin{align} 
  \delta_\Lambda\Big(\sum_{e\in\edges_+(v)} k_e-\sum_{e\in\edges_-(v)} k_e\Big),
\end{align}
with $\edges(v)=\edges_+(v)\cup \edges_-(v)$.  How the edges are split between
the two sets depends on the type of vertex.  If $v\in \Cverts$, then
$\edges_-=\emptyset$ and $\edges_+=\edges(v)$.  Otherwise,
$\edges_+(v)=\set{e}$, where $e$ is the first edge attached to $v$,
and $\edges_-(v)=\edges(v)\setminus\set{e}\ne \emptyset$.
We have illustrated these definitions in the example graph in
Fig.~\ref{fig:graphex1} (The graph has $n',n>0$ with $n'\ne n$, and as such is
not related to any of the present  amplitudes.
However, we use the more general graph to show that the scheme does not depend
on the special relation between $n$ and $n'$.)

Our aim is next to ``integrate out'' all the constraint $\delta$-functions.
We do this by associating with every vertex a unique edge attached to it which
we use for the integration.
As long as we use each edge not more than once, this results in a complete
resolution of the momentum constraints.
The edges used in the integration of the $\delta$-functions are called
\defem{integrated},
and the remaining edges are called \defem{free}.  We use the notation
$\edges'$ for the collection of integrated edges and $\fedges$ for the free
edges. The following theorem shows that there is a way of achieving such a
division of edges which respects their natural time-ordering.
\begin{theorem}\label{th:intdeltas}
Consider a momentum graph $\mathcal{G}$.
There exists a complete integration of the
momentum constraints, determined by a certain unique spanning tree of the
graph, such that for any free edge $f$ all $k_e$ with
$e<f$ are independent of $k_f$.
In addition, all free edges \defem{end} at a fusion vertex:
if $f$ is free,
there is a fusion vertex $v\in \Fverts$ and $v'\in \verts$ such that
$\tau(v)>\tau(v')$ and $f=\set{v,v'}$.
\end{theorem}

From now on, we assume that the momentum constraints are integrated out using
the unique construction in Theorem \ref{th:intdeltas}.  For any fusion vertex
$v\in \Fverts$, we call the
number of free edges in $\edges_-(v)$
the \defem{degree} of the fusion vertex, and denote this by $\deg v$.
The following theorem summarizes how the integrated edges
ending at an interaction vertex depend on its free momenta.
\begin{proposition}\label{th:momatintv}
The degree of a fusion vertex belongs to $\set{0,1,2}$.
If $v\in\Iverts$ is a degree one interaction vertex, then
$\edges_-(v)=\set{f,e,e'}$
where $f$ is a free edge, and $k_e=-k_f+G$, $k_{e'}=G'$, where $G$ and $G'$
are independent of $k_f$.
If $v$ is a degree two interaction vertex, then $\edges_-(v)=\set{f,f',e}$
where $f,f'$ are free edges, and $k_e=-k_f-k_{f'}+G$, where $G$ is
independent of $k_f$ and $k_{f'}$.
\end{proposition}

We will need other similar properties of the integrated momenta, to be given
later in this section.
However, let us first explain how the constraints are removed.
\begin{proofof}{Theorem \ref{th:intdeltas} and Proposition \ref{th:momatintv}}
We construct a spanning tree for $\graph$ which provides a
recipe for integration of the vertex $\delta$-constraints and leads to the
properties stated in Theorem \ref{th:intdeltas}.
We first construct an unoriented tree $\Ttree=(\Tverts,\Tedges)$
from $\mathcal{G}$,
and then define an oriented tree $\oTtree=(\Tverts,\oTedges)$
by assigning an orientation to each of the edges in $\Tedges$.

Let $\Ttree\upn{0}=(\Tverts\upn{0},\Tedges\upn{0})$,
with $\Tverts\upn{0}=\emptyset=\Tedges\upn{0}$.
We go through all edges in $\edges$ in the opposite order they were
created, \itie ,
decreasing with respect to their order.
At the iteration step $l$, let $e$ denote the corresponding edge,
and consider the previous graph $\Ttree\upn{l-1}$.
If adding the edge $e$ to $\Ttree\upn{l-1}$ would
create a loop, we define
$\Ttree\upn{l}=\Ttree\upn{l-1}$.  Otherwise, we define $\Ttree\upn{l}$ as the
graph
resulting from this addition, \itie ,  we define
$\Tverts\upn{l}=\Tverts\upn{l-1}\cup e$, and
$\Tedges\upn{l}=\Tedges\upn{l-1}\cup \set{e}$.
Since in the first case necessarily
$e\subset \Tverts^{\ell-1}$, we will always have $e\subset \Tverts^{\ell}$,
and thus no vertex in $e$ can be lost in the iteration step.

Let $\Ttree=(\Tverts,\Tedges)$ denote the graph obtained after the final
iteration step.
By construction, at each step $\Ttree\upn{\ell}$ is a forest, and thus so is
$\Ttree$. Moreover, since $\graph$ is connected, $\Ttree$ is actually a tree.
Since every vertex in $\verts$ is contained in some edge,
we also have $\Tverts=\verts$.
In addition, $\Tedges\subset \edges$, and every
$e\in \edges\setminus\Tedges$ has the following property:
adding it to $\Tedges$ would make a unique
loop {\em composed out of edges $\set{e'}$
each of  which satisfies $e'\ge e$,\/}  and thus also
$\taup (e')\le \taup (e)$. (The loop is
unique since $\Ttree$ itself has no loops.)

Next we create $\oTtree$ by
assigning an orientation to the edges of $\Ttree$.  We root the tree
at $\rootv$.  This is achieved by the following
algorithm: we first note that for any vertex $v$ there is
a unique path connecting it to $\rootv$.  We orient the edges of the path
so that it starts from $v$ and ends in $\rootv$.
This is iterated for all vertices in the tree.  Although
it is possible that two different vertices share edges along the path, these
edges are assigned the same orientation at all steps of the algorithm.
(If two such paths share any vertex, then the paths
must coincide past this vertex; otherwise there would be a
loop in the graph.)
This results in an oriented graph in which for every
$v\in\verts\setminus\Rverts$ there is a {\em unique\/} edge 
$E(v)\in \edges(v)$ pointing {\em out\/}
of the vertex.  In addition, the map 
$E:\verts\setminus\Rverts\to \edges$
is one-to-one. Thus we can integrate all the momentum $\delta$-functions,
by using the variable $k_{E(v)}$ for the $\delta$-function at the
vertex $v\in\verts\setminus\Rverts$.
We have depicted the oriented tree resulting from the graph of
Fig.~\ref{fig:graphex1} in Fig.~\ref{fig:graphex2}.

\begin{figure}
  \centering
  \myfigure{height=0.3\textheight}{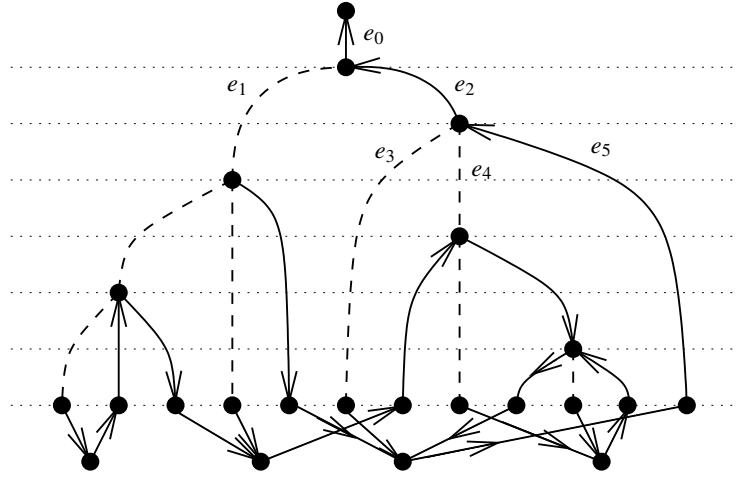}
  \caption{The oriented spanning tree $\oTtree$
    corresponding to the  graph $\mathcal{G}$ given in Fig.~\ref{fig:graphex1}.
    The edges in the complement of the tree
    have also been depicted by dashed lines.
    The enumeration $(e_\ell)$ of the edges corresponds to the one explained
    in the text; the spanning tree is constructed by adding the edges in
    the graph in {\em decreasing\/} order.
  \label{fig:graphex2}}
\end{figure}

After the above integration steps,
all the constraints have been resolved, and the set of remaining
integration variables will consist of $k_e$ with 
$e\in \edges\setminus \Tedges$.
These are all free integration variables, and thus
$\fedges:=\edges\setminus \Tedges$ is the set of free edges,
and $\edges':=\Tedges$ is the set of integrated edges.
Obviously, one has to add at least all edges attached to a cluster vertex
before a loop can be created, and thus no such edge is free.  Also, the
addition of the last edge $e_0$
never creates a loop.  All remaining edges end at a fusion vertex, and thus
this is true also of all free edges.

In order to conclude the proof of Theorem \ref{th:intdeltas},
we need to find out how the integrated momenta depend on the free ones.
For later use, let us spell  out also this fairly standard part in detail.
For any $v\in \verts$, let $\fedges(v)$ collect
the free edges attached to $v$,
$\fedges(v)=\edges(v)\cap \fedges$.
Let us also associate for any
$v\in \verts\setminus\Rverts$ an ``edge parity'' mapping
$\sigma_v:\edges(v)\to \set{-1,+1}$ defined by
\begin{align} 
  \sigma_v(e) = \begin{cases}
    +1, & \text{if } e\in \edges_+(v),\\
    -1, & \text{if } e\in \edges_-(v).
  \end{cases}
\end{align}
\begin{lemma}\label{th:mesigns}
If $e=\set{v,v'}\in \edges$ does not intersect $\Rverts$,
then $m(e)=-\sigma_v(e)\sigma_{v'}(e)=1$.
\end{lemma}
\begin{proof}
Assume first that $v\in \Cverts$.  Then $\sigma_v(e)=1$ and $v'\in \Dverts$.
Since $\edges(v')$ then contains only two elements, of which $e$ is created
later,
we have $\sigma_{v'}(e)=-1$, and thus $m(e)=1$.
Assume then $e\cap \Cverts = \emptyset$. If
$v\in \Dverts$, then $v'\in \Fverts$, and thus
$e$ is the earlier of the two edges attached to $v$ and $\sigma_v(e)=1$.
However, then it is one of the three later edges attached to $v'$,
and thus $\sigma_{v'}(e)=-1$.  This implies $m(e)=1$.
Since $m(e)$ is symmetric under the exchange of $v$ and $v'$,
we can now assume that $e\cap \Cverts \cap\Dverts = \emptyset$.
Then both $v,v'\in \Fverts$ and it follows from the construction
that $\sigma_v(e) \sigma_{v'}(e) = -1$.  This proves that for any edge $e$
with $e\cap \Rverts = \emptyset$, $m(e)=1$.
\qed \end{proof}

Consider then an  integrated variable $k_e$, with $e\in \edges'$.
The edge $e$ has been assigned an
orientation, say $e=(v_1,v_2)$, going from the vertex $v_1$ to the vertex
$v_2$. Let $\mathcal{P}(v)$, $v\in \verts$, denote the collection of the
vertices $v'$
for which there exists a path from $v'$ to $v$ in the
{\em oriented\/} tree $\oTtree$.  In particular, we include here the trivial
case $v'=v$. We claim that then
\begin{align}\label{eq:kesol}
  k_e = \sum_{v\in \mathcal{P}(v_1)}  \sum_{f\in \fedges (v)}
 \left( - \sigma_{v_1}(e) \sigma_{v}(f) \right) k_{f} .
\end{align}

This can be proven by induction in a degree $j$ associated with an
oriented edge $e=(v_1,v_2)\in \oTedges$: $j$ is defined as the maximum of the
number of
vertices in an oriented path
from any leaf to $v_1$  (note that such paths always
exist).  For $j=1$, $v_1$ is itself a leaf, and thus
$\mathcal{P}(v_1)=\set{v_1}$ and
$\fedges (v_1)=\edges(v_1)\setminus\set{e}$.
Also $v_1\not\in \Rverts$, since the edge $e_0$ is always oriented as
$(v_{N+1},\rootv)$.  Thus there is a
$\delta$-function associated with $v_1$, and it enforces
\begin{align} 
  \sum_{e'\in \edges(v_1)} \sigma_{v_1\!}(e') k_{e'} =0.
\end{align}
The designated integration of this $\delta$-function yields, with $v=v_1$,
\begin{align} 
 k_e = -\sigma_v(e) \sum_{e'\in \edges(v)\setminus\set{e}} \sigma_{v}(e')
 k_{e'} = \sum_{e'\in \fedges (v)}  ( -\sigma_v(e) \sigma_{v}(e')) k_{e'} ,
\end{align}
and therefore (\ref{eq:kesol}) holds for $j=1$.
Assume then that (\ref{eq:kesol}) holds for any edge up to degree $j\ge 1$,
and  suppose $e=(v_1,v_2)$ is an edge with a degree $j+1$.
Again $v_1\not\in \Rverts$, and the corresponding $\delta$-function implies
that, with $v=v_1$,
\begin{align} 
 k_e = \sum_{e'\in \edges(v)\setminus\set{e}} (-\sigma_{v}(e)
 \sigma_{v}(e')) k_{e'}.
\end{align}
In the sum, an edge $e'$ is either free or it
must have a degree of at most $j$, as otherwise $e$
would have a degree of at least $j+2$.  Thus by the induction assumption,
\begin{align}\label{eq:kesol2}
& k_e =
\sum_{f\in \fedges (v)} (-\sigma_v(e) \sigma_{v}(f)) k_{f}
+ \sum_{e'\in \edges(v)\setminus\set{e}\setminus \fedges (v)}
(-\sigma_v(e) \sigma_{v}(e'))
\nonumber \\ & \qquad
\times
\sum_{v'\in \mathcal{P}(V_1(e'))}  \sum_{f\in \fedges (v')}
 \left( -\sigma_{V_1(e')}(e') \sigma_{v'}(f)  \right) k_{f},
\end{align}
where $e'=(V_1(e'),v)$ and $V_1(e)$ denotes the first vertex of an oriented
edge $e$. Therefore,
\begin{align} 
& (-\sigma_v(e) \sigma_{v}(e'))
 \left( -\sigma_{V_1(e')}(e') \sigma_{v'}(f)  \right)
 = -\sigma_v(e) m(e') \sigma_{v'}(f)
 = -\sigma_v(e) \sigma_{v'}(f) \, .
\end{align}
Now (\ref{eq:kesol2})  can be checked to coincide with (\ref{eq:kesol}).
This completes the induction step, and thus proves  (\ref{eq:kesol}).

Consider then a free integration variable corresponding to
$f_0=\set{u,u'}\in\fedges$, where we can choose $\tau(u)>\tau(u')$.
Then $f_0\cap \Rverts =\emptyset$, $\taup (f_0)=\tau(u)$, and
$\sigma_u(f_0)=-1$, $\sigma_{u'}(f_0) = +1$.
The unique oriented paths from $u$ and $u'$ to
the root of the tree must coincide starting from a unique vertex $v_0$, which
can {\em a priori\/}
also be either $u$ or $u'$.  In addition, the paths before $v_0$ cannot
have any common vertices.  Suppose $e=(v_1,v_2)$ belongs to
the path from $u$ to $v_0$ in $\oTtree$.  Then $k_e$ depends on
$k_{f_0}$, and using (\ref{eq:kesol}) we find that
$k_e=\sigma_{v_1}(e) k_{f_0} + \cdots$.  Similarly, if $e$  belongs to
the path from $u'$ to $v_0$, then $k_e=-\sigma_{v_1}(e) k_{f_0} + \cdots$.
For any $e$ which comes after $v_0$ in the path, both terms will be present,
and they cancel each other.  Resorting to (\ref{eq:kesol}) thus proves
that only those $k_e$ whose edges are contained in either of the
paths $u\to v_0$ and $u'\to v_0$ depend on the free variable $k_{f_0}$.
However, as these edges, together with $f_0$, would
form a loop in $\mathcal{G}$, it follows from the
construction of $\Ttree$ that for any such edge $e$ we have
$e\ge f_0$.  This proves that  if $f\in \fedges$ and $e\in\edges$ with $e<f$,
$k_e$ is either free (and thus independent of $k_f$) or by the above result
does not depend on $k_f$.  This completes the proof of
Theorem \ref{th:intdeltas}.

Proposition  \ref{th:momatintv} is now a corollary of the above results.
The degree of the top fusion vertex is obviously less than two,
and since for any vertex adding its first edge cannot create a loop,
also the degree of all interaction vertices is less than or equal to two.
To prove the second claim,
let us assume $u$ is an interaction vertex and consider the three edges
belonging to
$\edges_-(u)$, one of which is $f_0$.
If $v_0\ne u$, then there is a non-trivial path from $u$ to
$v_0$, and we can assume that $e=(u,v_2)$
is the first edge along this path.  Since then $e>f_0$, we have
$e\in \edges_-(u)$,and thus $k_e=-k_{f_0}+\cdots$.
If $v_0= u$, there is a non-trivial path from $u'$ to $u$, and let $e=(v_1,u)$
be the last edge in that path.  By $e>f_0$,
again we then have $e\in \edges_-(u)$.  Since $m(e)=1$, we find
$k_e=-\sigma_{v_1}(e) k_{f_0} + \cdots=\sigma_{u}(e) k_{f_0} + \cdots=
-k_{f_0} + \cdots$.
Finally, consider the third edge $e'\in \edges_-(u)$.
This cannot belong to either of the paths from $u\to v_0$ and $u'\to v_0$,
and thus $k_{e'}$ is always independent of $k_{f_0}$.
If $e'$ is integrated, the degree of $u$ is one, and we have proved the
statement made in the Proposition.
If $e'$ is a free edge, the degree of $u$ is two.  If we then apply the above
result to
$e'$ instead of $f_0$, we can conclude
that $k_e=-k_{e'}-k_{f_0}+\cdots$, where the remainder is independent of both
$k_{e'}$ and $k_{f_0}$.
(Note that $e$ is then the only integrated edge in $\edges_-(u)$, and must
therefore contain both of the free variables.)
This completes also the proof of 
Proposition \ref{th:momatintv}.\qed
\end{proofof}

In the following,
the term {\em oriented path} refers to a path in $\oTtree$.
Without this clarifier,  a path always refers to an
unoriented path in a subgraph of $\mathcal{G}$.
The following Lemma improves on (\ref{eq:kesol}) and
yields the exact dependence of integrated momenta on the free ones.
\begin{lemma}\label{th:kesol3}
For any integrated edge $e=(v_1,v_2)\in \edges'$,
let $P=\mathcal{P}(v_1)$ denote the collection of vertices
such that there is an oriented path from the vertex to
$v_1$.  Then
\begin{align}\label{eq:kesol3}
&  k_e = \sum_{v\in P}  \sum_{f=\set{v,v_f}\in \fedges (v)}
   \1(v_f\not\in P)
 \left( - \sigma_{v_1}(e) \sigma_{v}(f) \right) k_{f} .
\nonumber \\ & \quad
 = - \sigma_{v_1}(e) \sum_{f\in \fedges}
   \1(\exists v\in f\cap P\text{ and }f\cap P^c\ne \emptyset )
   \sigma_{v}(f) k_{f} .
\end{align}
In addition, any $f=\set{v,v'}\in\edges$, such that $f\ne e$,
$v\in P$, and $v'\not \in P$, is free.
\end{lemma}
\begin{proof}
By (\ref{eq:kesol}) the result in (\ref{eq:kesol3})
holds without the characteristic function
$\1(v_f\not\in \mathcal{P}(v_1))$.  Consider thus
$v,v'\in \mathcal{P}(v_1)$, $v'\ne v$, such that
$f=\set{v,v'}$ is free.  Then $\sigma_{v}(f)=-\sigma_{v'}(f)$, and thus
$-\sigma_{v_1}(e) \sigma_{v}(f)-\sigma_{v_1}(e) \sigma_{v'}(f)=0$.
Now sums over edges in $\fedges (v)$ and $\fedges (v')$ appear in
(\ref{eq:kesol}), and thus the terms proportional to $k_f$ in these sums
cancel each other.  This proves (\ref{eq:kesol3}).

To prove the last statement let us assume the converse.
We suppose $f$ is not free, which implies
that $f$ has a representative in $\oTtree$.
Suppose first that it is $(v,v')$.
There is a unique oriented path from $v'$ to the
root of the corresponding oriented tree.  Since $v'\not\in P$,
the path does not contain $v_1$.  This however is not possible
because there is also an oriented path from $v$ to $v_1$ to the root
(otherwise $\Ttree$ contains a loop).

Therefore, we only need to consider
$f=(v',v)$. Then there is an oriented path from
$v$ to $v_1$, and thus also an oriented path from $v'$ to $v_1$.
This contradicts $v'\not\in P$, and thus we can conclude that $f$ must be
free. \qed 
\end{proof}

\begin{corollary}\label{th:zerok}
For any edge $e\in \edges$, there is a unique collection of free edges
$\fedges_e$, and of $\sigma_{e,f}\in \set{\pm 1}$, $f\in\fedges_e$,
such that
\begin{align} 
k_e =\sum_{f\in \fedges_e} \sigma_{e,f} k_f\, .
\end{align}
In addition, $k_e$ is independent of all free momenta if and only if
$k_e=0$.  This is equivalent to $\fedges_e=\emptyset$, which occurs if
and only if the number of connected components increases by one
when the edge $e$ is removed from $\mathcal{G}$.
\end{corollary}
\begin{proof}
If $e\in \fedges$, we choose $\fedges_e=\set{e}$ and $\sigma_{e,e}=1$.
Otherwise, the existence part follows from the Lemma.
Suppose there are two such expansions given by $\fedges_e$,
$\sigma_{\cdot,e}$ and $\fedges'_e$, $\sigma'_{\cdot,e}$.
If $\fedges'_e\ne \fedges_e$, the difference of the expansions would
contain some free momenta with coefficients $\pm 1$, and
if $\fedges'_e= \fedges_e$ but $\sigma'_{\cdot,e}\ne \sigma_{\cdot,e}$,
some free momenta would appear in the difference with coefficients
$\pm 2$.  This proves that the expansion is unique.

Obviously, $k_e$ is a constant if and only if $\fedges_e=\emptyset$,
when $k_e=0$.
If $e$ is free, then $\fedges_e$ is not empty. Since
the spanning tree is then not affected by removal of $e$,
the number of connected components remains unchanged by the removal.
Therefore, the Corollary holds in this case.

Else $e=(v_1,v_2)$ is an integrated edge and we can apply Lemma
\ref{th:kesol3}.  Denote $P=\mathcal{P}(v_1)$, and
suppose there is a path from $v_1$ to $v_2$ which does not contain
$e$.  In this case, removing $e$ from $\graph$ does not create any new
components.
Along this path there is an edge $f=\set{v,v'}$
such that $v\in P$ but $v'\not\in P$.  Since $f\ne e$, by
Lemma \ref{th:zerok}, $f\in \fedges$.  This implies that $k_e$
depends on $k_f$, and is not uniformly zero, in accordance with the
Corollary.

Finally, assume that every path from $v_1$ to $v_2$ contains $e$.
This implies that $v_1$ and $v_2$ belong to different components if
$e$ is removed from  $\graph$.  However, then there still
must be a path from any vertex to either $v_1$ or $v_2$, and the number of
components is thus exactly two.  Suppose $f=\set{v,v'}\in \fedges_e$, and
choose $v\in P$, $v'\not \in P$.  Following the oriented paths in the opposite
direction, we obtain paths $v_1\to v$, $\rootv\to v_2$ which do not contain
$e$.  Since the oriented path $v'\to \rootv$ does not
contain $e$, we can join the three segments with $f$ into a path $v_1\to v_2$
which avoids $e$. 
This contradicts the assumption, and thus now $\fedges_e=\emptyset$.
This completes the proof of the Corollary.
\qed \end{proof}

Since removing $e_0$ from $\graph$ isolates $\rootv$,
the Corollary implies that always $k_{e_0}=0$, \itie , the sum of the top momenta
of plus and minus trees is zero.  We have already used this property in
the derivation of the amplitudes.

\begin{lemma}\label{th:nokkdiff}
Suppose $f,f'$ are the two free edges ending at
a degree two interaction vertex
$v_0\in \Iverts$.  Let $e=(v_1,v_2)\in \edges'$ be an integrated edge.
Then $k_e=F_e(k_f,k_{f'})+G_e$
where $G_e$ is independent of $k_f,k_{f'}$ and
$(k,k')\mapsto F_e(k,k')$ is one of the following seven functions:
$0$, $\pm k$, $\pm k'$, $\pm (k+k')$.

Let $v\in \Iverts$, and suppose $e,e'\in \edges_-(v)$, $e\ne e'$.
Then $k_e+k_{e'}=F(k_f,k_{f'})+G$
where $G$ is independent of $k_f,k_{f'}$ and
$(k,k')\mapsto F(k,k')$ is also one of the above seven functions.
If $v=v_0$, the choice is reduced to one of the functions
$-k$, $-k'$, and $k+k'$.
\end{lemma}
\begin{proof}
There are $w,w'\in \Iverts\cup\Cverts$ such that $k_f=\set{v,w}$,
$k_{f'}=\set{v,w'}$.
Then $\sigma_w(f)=1=\sigma_{w'}(f')$ and
$\sigma_v(f)=-1=\sigma_{v}(f')$.  We express $k_e$ using (\ref{eq:kesol}),
which shows that $k_e=F_e+G_e$ where
$F_e = -\sigma_{v_1}(e)(-o_v (k_f+k_{f'})+o_{w} k_f +o_{w'} k_{f'})$, and
$o_x$ is one, if there is an oriented path from the vertex
$x$ to $v_1$, and zero
otherwise.  Checking all combinations produces the list of seven functions
stated in the Lemma.

Let us then consider the second statement.
If $v=v_0$, then either $e$ or $e'$ is the unique integrated edge in 
$\edges_-(v)$, or $e,e'$ are equal to the free momenta $f,f'$.
In the first case, we can apply Proposition \ref{th:momatintv} which shows that
either $F=-k$ or $F=-k'$ will work, and in the second case we can
choose $F=k+k'$.  We can thus assume $v\ne v_0$.
If both $e,e'$ are free, then $F=0$ works.  If only one of the edges is free,
then the previous result implies the existence of the decomposition.

Thus we can assume that both
edges are integrated, and $e=\set{v,w_1}$, $e'=\set{v,w'_1}$,
where $w_1\ne w'_1$.
Suppose first that $e$ points out
from $v$, \itie , $e=(v,w_1)$.  Then $e'$ points in, $e'=(w'_1,v)$,
and any oriented path to $w'_1$ extends into an
oriented path to $v$.  Since $\sigma_{w'_1}(e')=1$ and
$\sigma_{v}(e)=-1$, we have
$k_e+k_{e'}=F_e + F_{e'} +G_e + G_{e'}$
with $F_e + F_{e'}=-(-o_v (k_f+k_{f'})+o_{w} k_f +o_{w'} k_{f'})
-o_v (k_f+k_{f'})+o_{w} k_f +o_{w'} k_{f'}
-o'_v (k_f+k_{f'})+o'_{w} k_f +o'_{w'} k_{f'}$, where $o_x$ is one, if there
is an oriented path from $x$ to $w'_1$, and $o'_x$ is one if there
is an oriented path from $x$ to $v$ which does not go via $w'_1$.
Thus $F_e + F_{e'}=-o'_v (k_f+k_{f'})+o'_{w} k_f +o'_{w'} k_{f'}$
is also of the stated form.

In the remaining case, we can assume that both $e$ and $e'$ point into $v$:
$e=(w_1,v)$, $e'=(w'_1,v)$.  If there is an oriented path
from a vertex $x$ to $w_1$, there cannot be an oriented path from
$x$ to $w'_1$, and vice versa.  Since, in addition,
$\sigma_{w'_1}(e')=1=\sigma_{w_1}(e)$, we have
$k_e+k_{e'}=F_e + F_{e'} +G_e + G_{e'}$
with $F_e + F_{e'}=-(-o_v (k_f+k_{f'})+o_{w} k_f +o_{w'} k_{f'})$
where $o_x$ is one if there is an oriented path from $x$ to either
$w_1$ or $w'_1$, and zero otherwise.  This completes the proof of the
Lemma.
\qed \end{proof}

\begin{proposition}\label{th:pairmom}
Suppose $e, e'\in \edges$, $e\ne e'$, are such that $\fedges_e=\fedges_{e'}$.
Then either $k_e=0=k_{e'}$ or
removal of $e$ and $e'$ from $\graph$ splits it into exactly two
components.
\end{proposition}
\begin{proof}
If $\fedges_e=\fedges_{e'}=\emptyset$, then by Corollary \ref{th:zerok}
$k_e=0=k_{e'}$, and the first alternative holds.  If both $e$ and $e'$ are
free, then $\fedges_e=\fedges_{e'}$ implies $e'=e$, which was not allowed.
We can thus assume that $k_e,k_{e'}\ne 0$, and that at least one of $e,e'$
is integrated.
Let $\graph''$ denote the graph obtained by removing $e$ and $e'$
from $\graph$.

Suppose first that one edge is free, but the other is not.
By symmetry, we can choose $e'$ to be the free one,
and assume $e=(v_1,v_2)$.
Let $P=\mathcal{P}(v_1)$.  Since $\fedges_e=\fedges_{e'} =\set{e'}$, 
there are $v\in P$ and $v'\not\in P$ such that $e'=\set{v,v'}$.
Removal of $e'$ does not change the number of components,
since the spanning tree is not affected by it.
We thus need to show that the subsequent removal of $e$ will split the
component.
Suppose $f=\set{w,w'}$ is an edge such that $w\in P$ and $w'\not\in P$.
If $f\ne e$, $f$ is free by Lemma \ref{th:kesol3}, and thus $f\in \fedges_e$
which implies $f=e'$.
Thus $e,e'$ are the only edges with this property, and every path
from the component containing $P$ to the one containing $P^c$ must use either
$e$ or $e'$.  On the other hand, the vertices in $P$ (respectively $P^c$) can
be connected without using $e$ or $e'$, and thus $\graph''$ has
exactly two components.

Thus we can assume that neither $e$ nor $e'$ is free.
We identify $e'=(v'_1,v'_2)$ and $e=(v_1,v_2)$ in
$\oTedges$. Let $P=\mathcal{P}(v_1)$,
and $P'=\mathcal{P}(v'_1)$.
Suppose first that the oriented path from $v_1$ to root
contains $v_1'$.  Then $P\subset P'$, and
we claim that $P\cup (P')^c$ and
$P'\setminus P$ span independent connected components in $\mathcal{G}''$.
Suppose
$f=\set{w,w'}$ is free and $w\in P' \setminus P$.
If $w'\in P$, then $f\in \fedges_e\setminus \fedges_{e'}$,
but this is empty by assumption.
Similarly, if $w'\not\in P'$, then
$f\in \fedges_{e'}\setminus \fedges_{e}=\emptyset$.
Thus necessarily $w'\in  P' \setminus P$.
This implies that cutting $e$ and $e'$ isolates
$P' \setminus P$ from both $P$ and $(P')^c$. $P' \setminus P$
is not empty since it contains at least $v_1'$, and
for any
$w\in P'\setminus P$ there is an oriented path from $w$ to $v'_1$ which, by
the above results, cannot contain $e$ nor (obviously) $e'$.  Thus
$P' \setminus P$ spans
a connected component in $\mathcal{G}''$.  Both $P$ and $(P')^c$
are connected in  $\mathcal{G}''$
(note that  every vertex in $(P')^c$
has a path to the root which does not go via $v_1'$).
On the other hand, there must be a free edge connecting $P$ and $(P')^c$,
which thus is different from $e$ or $e'$,
as else $k_{e}=0$.  Therefore, $P\cup (P')^c$ spans the second, and last,
component in $\mathcal{G}''$.
We have proven the result for the case
the oriented path from
$v_1$ to root contains $v'_1$ and obviously this also proves the result in
the case if the path from $v'_1$ to root contains $v_1$.

Thus we can assume that the oriented paths from $v_1$ and $v_1'$ to root are
not contained in each other.  This implies that $P\cap P'=\emptyset$.
If there exists a path from $P$ to $(P\cup P')^c$ avoiding $e$,
there is a free edge between these sets which thus belongs to
$\fedges_{e}\setminus \fedges_{e'}=\emptyset$.
Similarly, there cannot be any path from
$P' $ to $(P\cup P')^c$ avoiding $e'$.  Since $k_e\ne 0$, there must thus be
a free edge from $P$ to $P'$.  Therefore, $P\cup P'$
and $(P\cup P')^c$ span disjoint connected components in $\mathcal{G}''$.
This completes the proof of the Proposition.
\qed \end{proof}

\begin{lemma}\label{th:pairmom2}
Suppose $e, e'\in \edges$, $e\ne e'$.  Then $\fedges_e=\fedges_{e'}$
if and only if there is $\sigma\in \set{\pm 1}$ such that
$k_e=\sigma k_{e'}$ independently of the free momenta.
\end{lemma}
\begin{proof}
If $k_e=\sigma k_{e'}$, then by
uniqueness of the representation in Corollary \ref{th:zerok}
$\fedges_e=\fedges_{e'}$.
If $e$ and $e'$ are both free, then
$\fedges_e=\fedges_{e'}$ implies $e'=e$ which is not allowed.
If one of them is free, say $e'$,
then $\fedges_e=\fedges_{e'}$ and (\ref{eq:kesol3})
imply $k_e= \pm k_{e'}$.

Thus we can assume that $e$ and $e'$ are not
free, and set $e=(v_1,v_2)$, $e'=(v'_1,v'_2)$.
We also denote $P=\mathcal{P}(v_1)$, $P'=\mathcal{P}(v'_1)$.
Suppose first that $P\subset P'$.  Then for any
$f\in \fedges_e$ there are $v\in P$, $v'\not\in P$
such that $f=\set{v,v'}$.  Since also
$f\in \fedges_{e'}$ then necessarily $v\in P'$, $v'\not\in P'$,
and the factor $\sigma_v(f)$ is the same in the representation
(\ref{eq:kesol3})
both of $k_e$ and of $k_{e'}$.  This implies that
$k_{e'}=\sigma_{v'_1}(e')\sigma_{v_1}(e)k_e$, in accordance with the Corollary.
By symmetry, the same results also holds if $P'\subset P$.

If the oriented path from $v_1$ to the root is contained in the oriented path
from $v'_1$ to the root, then $P\subset P'$, and if the path from $v'_1$ is
contained in the path from $v_1$, then $P'\subset P$.
Thus we can assume the converse, which clearly implies $P\cap P'=\emptyset$.
Then if $f\in \fedges_e=\fedges_{e'}$, we have $f=\set{v,v'}$ where
$v\in P$, $v'\not \in P$ and  $v\not \in P'$, $v'\in P'$.
Thus by
Lemma \ref{th:mesigns} $\sigma_v(f)=-\sigma_{v'}(f)$, and we can conclude from
(\ref{eq:kesol3}) that $k_{e'}=-\sigma_{v'_1}(e')\sigma_{v_1}(e)k_e$.  This
concludes the proof of the Lemma.
\qed \end{proof}

\begin{definition}
  Consider $V_1, V_2 \subset \Dverts$.  We say that there is a
  \defem{connection} between them  
  if there is a cluster which connects them, \itie , if there is $A\in S$ such
  that 
 $A\cap V_1\ne\emptyset$ and $A\cap V_2\ne\emptyset$.  If there is no
 connection between $V_1$ and $V_2$, we say that  
 $V_1$ is \defem{isolated} from $V_2$. 
\end{definition}
If $V_1$ is isolated from $V_2$, then obviously also $V_2$ is isolated from
$V_1$.  Clearly, $V_1$ is isolated from $V_2$ if and only if there is no path
connecting them in the graph which includes only edges intersecting
$\Cverts$. 

\begin{corollary} \label{th:kkpl}
Let $v\in \Iverts$, and suppose $e,e'\in \edges_-(v)$, $e\ne e'$.
Then $k_e+k_{e'}$ is independent of all free momenta if and only if
the initial time vertices at the bottom of the interaction trees starting from
$e$ and $e'$ are
isolated from the rest of the initial time vertices.
In this case, $k_e+k_{e'}=0$.
\end{corollary}
\begin{proof}
Suppose $k_e+k_{e'}$ is independent of all free momenta.
By Corollary \ref{th:zerok} this is possible only if, in fact,
$k_{e'}=-k_e$. In particular, then $\fedges_e=\fedges_{e'}$ and by Proposition
\ref{th:pairmom} either $k_e=0=k_{e'}$,
or removing $e,e'$ splits a connected component.

Denote the set of initial time vertices
at the bottom of the interaction subtree starting from $e$ ($e'$) by  $D_e$
($D_{e'}$). If $k_e=0$, then
$k_e=0=k_{e'}$ which implies that $D_e$ and $D_{e'}$ are separately
isolated from the rest of the initial time vertices, and the theorem holds.
Otherwise,
we can assume that removing $e,e'$ splits the graph into two components.
Thus there can be no connection from $D_e\cup D_{e'}$ to its complement in
$\Dverts$. This proves the ``only if'' part of the theorem.

For the converse, suppose $D_e\cup D_{e'}$ is isolated from the rest of the
initial time vertices.  If there is no connection between $D_e$ and $D_{e'}$
then any path from $D_e$ to the root
must go via $e$, which implies that $k_e=0$.
Similarly, then $k_{e'}=0$, and thus also $k_e+k_{e'}=0$.
If there is a connection between $D_e$ and $D_{e'}$, then the larger of the
edges $e$, $e'$ is integrated, the other is free, and they sum to zero.
This completes the proof of the theorem.
\qed \end{proof}

\begin{corollary}\label{th:zeroisirr}
If the momentum graph has an edge $e\ne e_0$ such that
$k_e=0$ identically, $S$ contains a cluster with odd number of elements.
\end{corollary}
\begin{proof}
Suppose there is an edge $e$ such that $k_e=0$ identically.
If $e\in\edges_-(v)$ for some fusion vertex $v$, then the argument used in
the proof of Corollary \ref{th:kkpl} shows that the subtree spanned by $e$
must have isolated clustering.  This implies that the size of at least 
one of the clusters is odd. 
Since $e\ne e_0$, we can then assume that $e$ contains a cluster vertex.
However, since every cluster
has a size of at least two, removal of one such edge cannot split the
graph.  This contradicts $k_e=0$.
\qed \end{proof}

The following theorem proves that the number of free momenta is independent of
the choice of the spanning tree.  It is a standard result and included here
mainly for the sake of completeness.
\begin{proposition}\label{th:loopsinv}
Let $\Ttree_1=(\verts,\edges_1)$ and $\Ttree_2=(\verts,\edges_2)$ be spanning
forests of a graph $\mathcal{G}=(\verts,\edges)$.  Then
$|\edges_2|=|\edges_1|$.
\end{proposition}
\begin{proof}
We make the proof by
induction in $|\edges_2\setminus \edges_1|$.  If this number is zero, then
$\edges_2\subset\edges_1$, and as $\edges_1$ cannot contain any loops,
we have $\edges_2=\edges_1$, and the theorem holds.

Make the induction assumption that the theorem holds up to $N\ge 0$.
Consider $\edges_2$ such that $|\edges_2\setminus \edges_1|=N+1$.
Then there is $f_0\in \edges_2\setminus \edges_1$.
Adding $f_0$ to $\Ttree_1$ creates a unique loop. Let
$f'_i$, $i=1,\ldots,n$,
count the momenta along this loop which do not belong to $\edges_2$,
\itie , which belong to $\edges_1\setminus \edges_2$. Then $n\ge 1$,
as otherwise $\Ttree_2$ would contain a loop.  On the other hand,
adding one of $f'_i$ to $\Ttree_2$ also creates a unique loop.
If none of these new loops contains $f_0$, then $\edges_2$ has a loop:
if $f_0=\set{a,b}$, one can start from $a$, follow the first loop,
and go around each $f'_i$ along the new loops, arriving finally to $b$.
Thus we can assume that $f'\in \edges_1\setminus \edges_2$
is such that it belongs to the loop created by $f_0$ in $\Ttree_1$,
and $f_0$ belongs to the loop created by $f'$ in $\Ttree_2$.

We then set $\edges_3=(\edges_2\cup \set{f'})\setminus \set{f_0}$,
and consider $\Ttree_3=(\verts,\edges_3)$.  Since the removal of $f_0$ cuts
the unique loop generated by $f'$ in $\edges_2$, $\Ttree_3$ has no loops,
\itie , it is also a forest.  If $g\in \edges_3^c$, then $g\ne f'$ and
either $g\in \edges_2^c$ or $g=f_0$.  Adding $f_0$ to $\Ttree_3$ creates a
loop by construction.  Else $g\in \edges_2^c$, and adding it to $\Ttree_2$
creates a unique loop.  If $f_0$ is not along this loop, it is composed of
edges in $\edges_3$, and adding $g$ to $\Ttree_3$ creates a loop.
If $f_0$ is along this loop, we can avoid it by using the loop created
by the addition of $f'$, and construct a loop out of edges in
$\set{g}\cup\edges_3$.
Thus $\Ttree_3$ is a spanning forest.  Since
$|\edges_3\setminus\edges_1|=N$, we can apply the induction assumption to it,
which shows that
$|\edges_2|=|\edges_3|=|\edges_1|$. This
completes the induction step.
\qed \end{proof}

\begin{proposition}\label{th:noffreemom}
A momentum graph has exactly $2 N+2-|S|$ free momenta, where $N=n+n'$.
\end{proposition}
\begin{proof}
We construct a second spanning tree by first going from top to bottom, then
adding
the edges containing the cluster vertices, going from left to
right (this is exactly the opposite order in which the spanning tree was
constructed before).
Clearly, the spanning tree then contains all edges in the interaction tree,
and exactly one edge per cluster in $S$ (the first edge connects
the cluster vertex to the tree, but every further edge would create a loop).
Thus there are altogether
$\sum_{A\in S} (|A|-1)=2 N+2-|S|$ free edges attached to the cluster vertices.
By Proposition \ref{th:loopsinv}, the number of free momenta
is independent of the choice of the spanning tree, and thus the result holds
also for the first spanning tree.
\qed \end{proof}

\section{Expansion parameters and classification of graphs}
\label{sec:classification}

\begin{definition}[Expansion parameters]\label{th:defkappaetc}
Let $\delta$ be a constant for which the
dispersion relation $\omega$
satisfies the dispersion bound (DR3), and
$\gamma$ be a constant for which the dispersion relation
satisfies the crossing bounds in (DR4).  We define
\begin{align}\label{eq:defgammap}
b= \frac{3}{4}, \quad \gamma' = \min (\frac{1}{4},2 \gamma,2\delta),\quad
a_0  = \frac{\gamma'}{24},\quad\text{and}\quad
b_0 =16 \Bigl( 3 + \frac{1}{a_0}\Bigr).
\end{align}
For any $\lambda>0$ let us then define
\begin{align}\label{eq:chooseN0}
\vep=\lambda^2\quad\text{and}\quad
  N_0(\lambda)
  = \max\Bigl(1,\Bigl\lfloor\,\frac{a_0 \, |\ln \lambda|}{\ln
    \sabs{\ln \lambda}}\,
  \Bigr\rfloor\Bigr)
\, ,
\end{align}
where $\lfloor x\rfloor$ denotes the integer part of $x\ge 0$.
Let also, with $N_0=N_0(\lambda)$,
\begin{align}\label{eq:defkappa}
\kappa'(\lambda)=\lambda^2 N_0^{b_0}\qand
\kappa_n(\lambda) = \begin{cases}
0,& 0\le n< N_0/2\, ,\\
\kappa'(\lambda), & N_0/2 \le n\le N_0\, .
\end{cases}
\end{align}
\end{definition}
The definition of $b$, associated with the removal of the singular manifold,
coincides with the one given earlier in Section \ref{sec:graphs}.

With this choice of parameters, 
we have $N_0\to \infty$, $\max_n\kappa_n\to 0$ for $\lambda\to 0^+$, and
\begin{align}\label{eq:N0limit1}
\frac{N_0(\lambda) \ln \sabs{\ln\lambda}}{|\ln \lambda|} \to a_0\qand
\frac{N_0(\lambda) \ln N_0(\lambda)}{|\ln \lambda|} \to a_0  \, .
\end{align}
If $c,c'>0$, $n_1,n_2,n_3\in \N_+$, and $p_1,p_2\in \R$
are some fixed given constants, then using $n!\le n^n$ we easily find that
\begin{align}\label{eq:N0limit2}
  c^{N_0}  \lambda^{p_1} N_0^{p_2 N_0+c'}
 ((n_1 N_0)!)^{n_2} \sabs{\ln \lambda}^{n_3 N_0+c'} \to 0\, ,
\end{align}
as soon as the inequality $p_1- a_0( p_2+ n_1 n_2 +n_3)>0$ is satisfied.
The decay is then actually powerlaw in $\lambda$, with the supremum of the
power determined by the above difference.
For instance, with our choices of $a_0,b_0$,
we have up to a powerlaw $\lambda^2$ decay in
\begin{align}\label{eq:N0lima}
  c^{N_0}  \lambda^{-2} N_0^{-b_0\frac{1}{4} N_0+4 N_0+c'}
 (4 N_0)! \sabs{\ln \lambda}^{4 N_0+c'} \to 0\, ,
\end{align}
and up to a powerlaw $\lambda^{\gamma' \! /2}$ decay in
\begin{align}\label{eq:N0limb}
  c^{N_0}  \lambda^{\gamma'} N_0^{4 N_0+c'}
 (4 N_0)! \sabs{\ln \lambda}^{4 N_0+c'} \to 0\, .
\end{align}

Consider a generic momentum graph, defined using parameters
$(S,J,n,\ell,n',\ell')$.
We integrate out all the momentum constraints using the spanning
tree which respects the time-ordering, as explained in the previous section.
We recall also the definition of a degree of an interaction vertex (we stress
here that this concept is not a graph invariant, and thus depends on the way
we have constructed the spanning tree).
By Proposition \ref{th:momatintv}, the degree counts the number of free
momenta ending at the vertex, and it belongs to $\set{0,1,2}$.
The following terminology will be used from now on:
\begin{definition}
Consider a time slice $i\in I_{0,n+n'}$ in a momentum graph.
If it has exactly zero length, it is called \defem{amputated}.
If it ends in an interaction vertex of degree $1$ or $2$ it is called
{\em short\/}.  Otherwise, it is called {\em long\/}.
\end{definition}
By this definition, the time slice $i=n+n'$ is always long.

The graph is called
\begin{description}
 \setlength{\itemsep}{0pt}
 \setlength{\parsep}{0pt}  %
\item[{\em irrelevant,}\/] if the amplitude corresponding to the graph is
  identically zero. Otherwise it is {\em relevant.}
\item[{\em pairing,}\/] if for every $A\in S$ we have $|A|=2$.
Otherwise it is {\em non-pairing.}
\item[{\em higher order,}\/] if it is a relevant non-pairing graph.
\item[{\em fully paired,}\/] if it is a pairing graph
and has no interaction vertices of degree one.
A pairing graph which is not fully paired is called {\em partially paired.}
\end{description}
Clearly, if there is $A\in S$ such that $|A|$ is odd, the graph is
irrelevant.

\begin{figure}
  \centering
  \medskip
  \myfigure{width=0.9\textwidth}{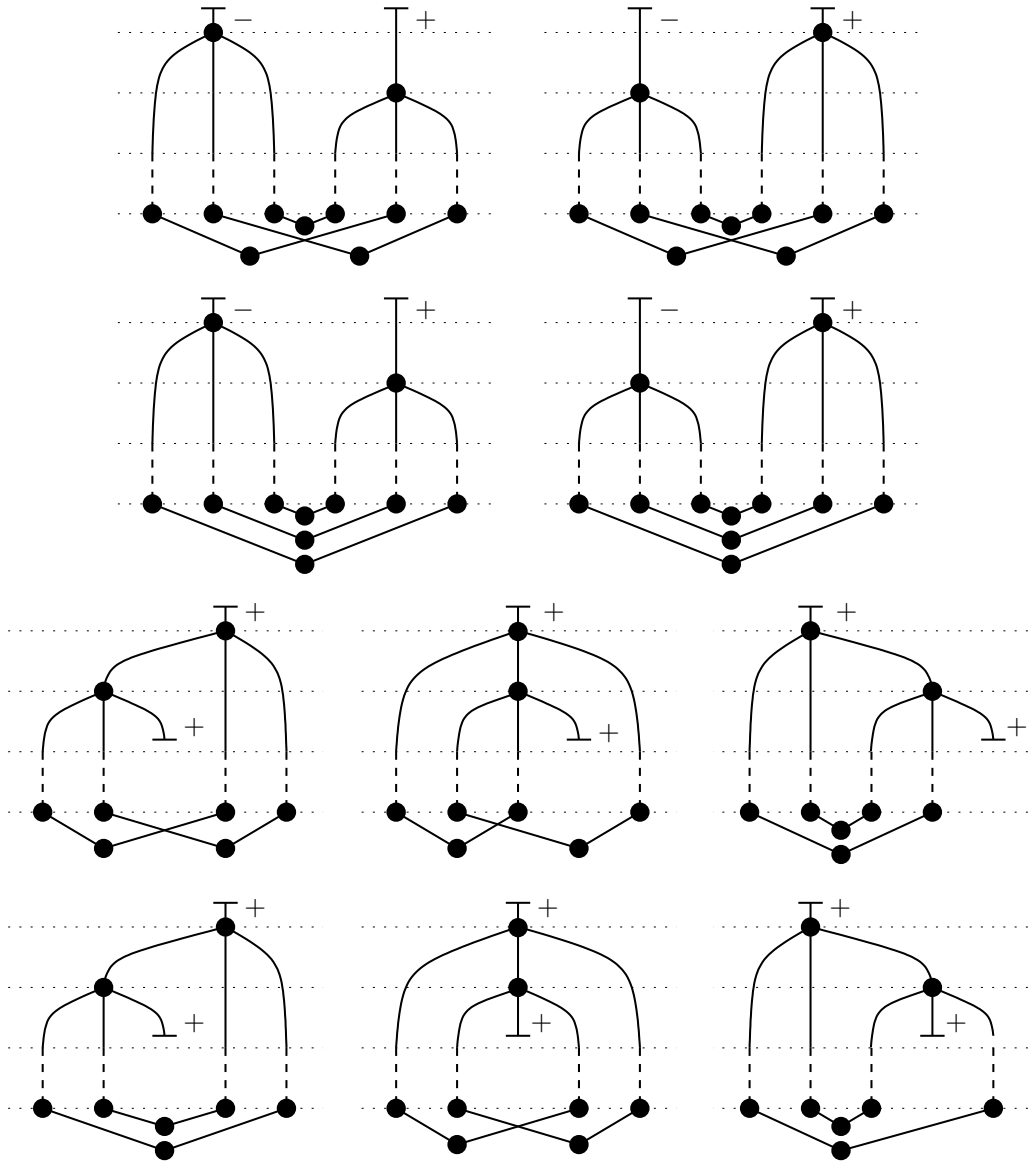}
  \medskip
\caption{Half of the leading motives: the first four depict ``gain motives'',
and the lower six ``loss motives''.  The ``truncated'' lines denote the places 
where the motive is attached to a graph with the appropriate parities shown
next to the line.
The dashed lines will always extend to the initial time vertices, \itie , these
parts of the edges will stretch over several time slices when more vertices
are added below the motive.
The remaining leading motives can be obtained from the above by inverting the
parities of all edges, and then inverting the order of the edges below
the interaction vertices.\label{fig:leadmot}}
\end{figure}

\begin{figure}
  \centering
  \myfigure{width=0.95\textwidth}{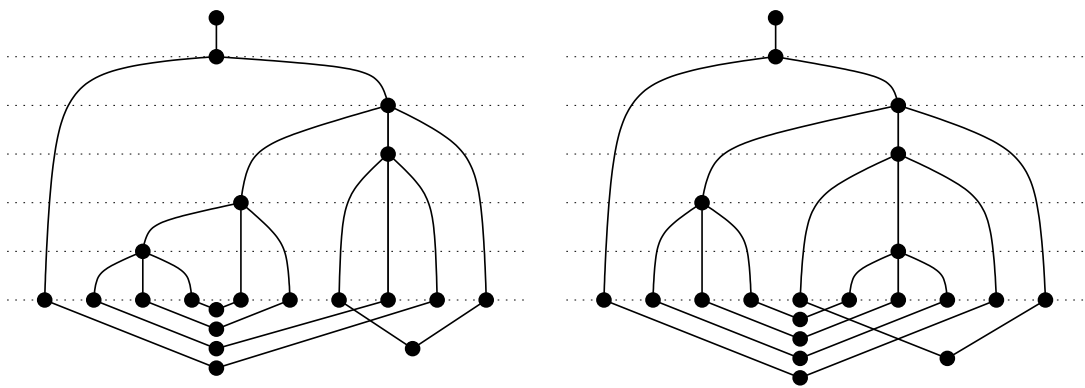}
  \medskip
\caption{Two examples of leading graphs, from the second level of iteration of
  leading motives. Both are obtained starting from the left graph in
  Fig.~\ref{fig:maingraph}.  In the first graph we have added the inverse of
  the third loss motive to the second line (the line has negative parity).  In
  the second graph we have split the corresponding pairing using the third
  gain motive.\label{fig:leaditer}} 
\end{figure}

Every fully  paired graph has thus
interaction vertices only of degree $0$ and $2$.
Such graphs can be obtained for instance by iteration of
the graph ``motives'' depicted in Fig.~\ref{fig:leadmot}.
These motives are called {\em leading motives\/}
or {\em immediate recollisions\/}.  The latter name comes from the fact that
the motive does not change the ``incoming'' momentum.
The motives can be iteratively attached to a graph in two ways: the gain
motives can replace any pairing cluster
(the order of parities of the cluster determines which four of the total eight
gain motives can be used: the first four in Fig.~\ref{fig:leadmot} are used
for replacing $(-1,1)$-pairings),
and the loss motives can be attached to any line with the correct parity
(the bottom six in Fig.~\ref{fig:leadmot} can be attached to a line of
parity $1$).
Any graph which is obtained by such an iteration starting from the simple
graph corresponding to $n=0=n'$ (a single loop) is called a leading term
graph. 

A straightforward induction shows that a leading term graph is fully paired.  
Furthermore,
the fully paired graphs are classified into three categories,
depending on properties of the phase factors of long time slices:
a fully paired graph is called
\begin{description}
 \setlength{\parsep}{0pt}  
 \setlength{\itemsep}{0pt}
\item[{\em leading,}\/] if it is formed by iteration of leading motives.
\item[{\em crossing or nested,}\/] otherwise.
\end{description}
The precise classifications require technical definitions, to be given in
Section \ref{sec:fullypaired}.  We only mention here that
in a nested graph
the first short time slices (at the bottom of the graph)
consist of leading motives ``nested'' inside another leading motive.
(This explains the name, already used in \cite{erdyau99} for a similar
construction, although
it would be more precise to call our nested graphs as ones which
{\em begin\/} with a nest.)

\subsection{Iterative cluster scheme}\label{sec:iterclsuters}

An important estimate of the magnitude of the amplitude associated with a
momentum graph,
is the total number of interaction vertices of the various degrees.
We denote by $n_i$ the number of interaction vertices of degree $i$.
Then for instance $n_0+n_1+n_2=n+n'$, and the following Lemma captures other
basic relations between these numbers for relevant graphs.
\begin{lemma}\label{th:ivdegrees}
Consider a relevant graph, and let $r=N+1-|S|$, $N=n+n'$.  Then $\Nnp\le r \le
N$, where
$\Nnp:=|\defset{A\in S}{|A|>2}|$ is the number of clusters which are not
pairs,  and $r=0$ if and only
if $S$ is a pairing.  In addition,
$n_2-n_0=r$, $n_0=\frac{1}{2}(N-r-n_1)$,  and $0\le n_1\le N$,
$0\le n_0\le \lfloor\frac{N-r}{2}\rfloor$.
\end{lemma}
\begin{proof}
The graph is relevant and thus has no odd clusters.  Then
a cluster in $S$ is either a pair, or has a size of at least $4$.  Therefore,
$2 N +2 = \sum_{A\in S}|A|= 2 |S| + \sum_{A\in S}(|A|-2)
\ge 2|S|+ 2 \Nnp$.  This implies $r\ge \Nnp$
and $r\le N$ is obvious.  Also $r=0$ if and only if $S$ is a pairing.
If there is a cluster which contains initial time vertices from both plus and
minus trees,
then adding the second edge ($e_1$) to the top fusion vertex $v_{N+1}$ creates
a loop. Otherwise, the initial time vertices of plus and minus trees are
isolated from each other, which implies that both contain at least
one odd cluster.  Therefore, in a relevant graph
exactly one free momenta is not attached to an interaction vertex.
By Proposition \ref{th:noffreemom} thus  $2 N+1-|S|=2 n_2+n_1$,
and clearly also $N=n_0+n_1+n_2$. We substitute $|S|=N+1-r$, and find the
stated formulae for $n_2$ and $n_0$.  The upper bound for $n_0$ then arises,
as $n_1\ge 0$ and $n_0$ needs to be an integer.
\qed \end{proof}

We will also use the corresponding cumulative counters: we let
$n_j(i)$, $j=0,1,2$, denote the number of interaction vertices of degree $j$
below and including $v_i$.
\begin{proposition}\label{th:ivordering}
Consider an interaction vertex $v_i$, $1\le i \le N$, in a relevant graph.
Then always
\begin{align} 
n_2(i)\le r + n_0(i)\qand
n_0(i)\ge \frac{i-(n_1+r)}{2}\, ,
\end{align}
where $r=N+1-|S|$ and $n_1=n_1(N)$, as in Lemma \ref{th:ivdegrees}.
\end{proposition}

The proof of Proposition \ref{th:ivordering} is based on one more construction
related to the momentum graph,
which we call the \defem{iterative cluster scheme}.  Since the scheme will
reappear later, let us first explain it in detail.

Let us consider the evolution of the cluster structure while
the spanning tree is being built.  We define $S^{(0)}=S$
and let $S^{(i)}$ denote a clustering of the edges intersecting the
time slice $i$, induced by the following iterative procedure
where the interaction vertices are added to the graph, one by one
from bottom to top.  The addition of the vertex $v_i$ will thus fuse the three
edges in $\edges_-(v_i)$ into the one in $\edges_+(v_i)$.  All the three
``old'' edges belong to some clusters in $S^{(i-1)}$ while the ``new'' edge
does not appear there.  We construct $S^{(i)}$ by first joining all clusters
in $S^{(i-1)}$ which intersect $\edges_-(v_i)$,
and then replacing the three edges by the unique new one in $\edges_+(v_i)$. 
The rest of the clusters are kept unchanged.

If two of the three old edges belong to the same cluster in $S^{(i-1)}$, then
adding the second one would create a loop  in the
construction of the spanning tree.
Similarly, if all three edges go into
the same cluster, then this creates two separate loops.  Therefore, this also
determines the degree of the added interaction vertex:  if the vertex joins
three separate old clusters, it has degree $0$, if it joins two clusters, it
has degree $1$, and if all edges belong to the same previous cluster, it has
degree $2$.  An explicit application of the scheme is presented in
Fig.~\ref{fig:iterclusters}.

\begin{figure}
  \centering
   \medskip
  \myfigure{width=0.9\textwidth}{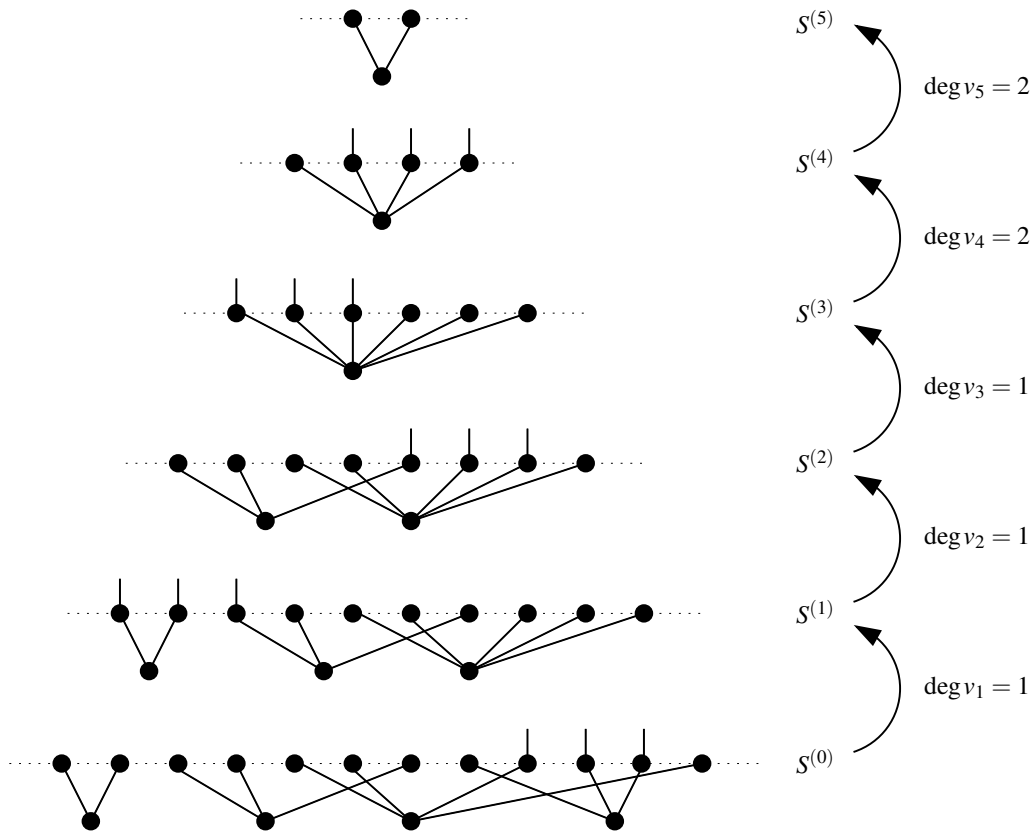}
  \medskip
  \caption{An application of the iterative cluster scheme to the example in
    Fig.~\ref{fig:graphex1}. (This graph is somewhat more general than what is
    discussed in the text since it contains clusters of odd size.) 
    For simplicity, we have also denoted at each level which of the edges will
    be fused in the next iteration step. 
    As explained in the text, the degree of the fusion vertex can also be
    deduced from the scheme. 
    These have also been shown and they can be checked to 
    coincide with the degrees which are apparent from the spanning tree 
    in Fig.~\ref{fig:graphex2}.\label{fig:iterclusters}}
\end{figure}

By going through all the alternatives, we then find that in the
iteration step the number of clusters changes as
$|S^{(i)}|=|S^{(i-1)}|-2+\deg{v_i}$.  In addition, the structure of the
clustering is conserved, in the sense that each  $S^{(i)}$ contains only even
(and non-empty) clusters.  

\begin{proofof}{Proposition \ref{th:ivordering}}
Consider the iterative cluster scheme, where
for each $i$, the set $S^{(i)}$ is a partition of $2(N-i)+2$ elements.
Since all clusters have size of at least two, we thus have
$|S^{(i)}|\le N-i+1$.  On the other hand, here
\begin{align} 
|S^{(i)}| = |S|-\sum_{j=1}^i (2-\deg v_j) =
|S|-2 i + n_1(i)+ 2 n_2(i)\, .
\end{align}
Since by construction $i=n_0(i)+n_1(i)+n_2(i)$, we have proven that
\begin{align} 
n_2(i) \le N+1-|S|+i-n_1(i)-n_2(i)= r+n_0(i),
\end{align}
as claimed in the Proposition.  But then also
\begin{align} 
i\le 2 n_0(i)+n_1(i)+r\le 2 n_0(i)+n_1+r\, ,
\end{align}
from which the second inequality follows.\qed
\end{proofof}

\section{Main lemmata}
\label{sec:lemmas}

We have collected here the main technical tools and results to be used in the
proof of the main estimates.

\subsection{Construction of momentum cutoff functions}
\label{sec:cutoff}

We first explain the construction of the cutoff function $\PFzero$, and prove
that it satisfies Proposition \ref{th:PFcorr} which was used in Section
\ref{sec:graphs} in the derivation of the basic Duhamel formulae.  We recall
that $b=\frac{3}{4}$.

Let $\Msing$ denote the singular manifold in Assumption
\ref{th:disprelass}.
Then there are $\Ns > 0$ and smooth closed
one-dimensional submanifolds $M_j$, $j=1,\ldots,\Ns$,
of $\T^d$
such that $\Msing=\bigcup_{j=1}^{\Ns} M_j$.  Since the manifold $M_j$ is
actually compact, for each $j$ there exists $\vep_j>0$
such that the map $k\mapsto d(k,M_j)$ is smooth in the neighborhood
$U_j:= \defset{k\in \smash{\T^d}}{d(k,M_j)<\vep_j}$ of $M_j$.
We define $\vep_0=\min_j \vep_j$, when $\vep_0>0$, and consider an arbitrary
$\vep$ such that $0<\vep< \vep_0$.  We recall here that
$\Msing\ne \emptyset$ since at least $0\in \Msing$.

We choose an arbitrary one-dimensional smooth
``step-function'' $\varphi$.  Explicitly, we assume that $\varphi\in
C^\infty(\R)$
is symmetric, $\varphi(-x)=\varphi(x)$, monotone on $[0,\infty)$,
and $\varphi(x)=0$ for $|x|\ge 1$ and $\varphi(x)=1$ for
$|x|\le \frac{1}{2}$.   In particular, then $\varphi(0)=1$.
We define further, for $0<\vep< \vep_0$, $j=1,\ldots,\Ns$,
the functions $f^{j}:\T^d \to [0,1]$ by
\begin{align} 
f^{j}(k;\vep) = \varphi\Bigl(\frac{d(k,M_j)}{\vep}\Bigr),
\quad k\in\T^d\, .
\end{align}
Then $f^{j}(k;\vep) =0$ for all $d(k,M_j)\ge \frac{\vep}{\vep_j}\vep_j$,
where $\frac{\vep}{\vep_j}<1$.  Thus,
by construction, $f^j$ is smooth on $\T^d$,
and we can find a constant $C$ independent of $\vep$ such that
$|\nabla f^j(k;\vep)|\le \frac{C}{\vep}$ for all $j,k,\vep$.
In addition, we have
$f^{j}(k;\vep)=1$ if $k\in M_j$.

Next we construct $d$-dimensional cut-off functions.
Let $\lambda'_0 = \min(1,\lambda_0,\vep_0^{1/b})$, and define
for all $0<\lambda<\lambda'_0$ the functions
$\Fone,\Fzero:\T^d\to \R$ by
\begin{align} 
\Fone(k) = \prod_{j=1}^{\Ns} \left(
  1-f^{j}(k;\lambda^b) \right), \quad
\Fzero = 1-\Fone \, .
\end{align}
\begin{lemma}\label{th:Foneprop}
There is a constant $C_1\ge 1$ such that for any $0<\lambda<\lambda'_0$,
\begin{jlist}
\item $0\le \Fone,\Fzero \le 1$.
\item If $k\in \Msing$, then $\Fone (k)=0$ and $\Fzero (k)=1$.
\item If $d(k,\Msing)\ge \lambda^b$, then $\Fone (k)=1$ and $\Fzero (k)=0$.
\item $\Fone,\Fzero$ are smooth, and
$|\nabla\Fone(k)|, |\nabla\Fzero(k)| \le C_1 \lambda^{-b}$, for all $k$.
\item  $0\le \Fone (k)\le C_1 \lambda^{-b} d(k,\Msing)$
for all $k\in \T^d$.
\item There is a constant $C$ such that
\begin{align} 
\int_{\T^d} \rmd k\, \Fzero(k)
\le \int_{\T^d} \rmd k\, \1(d(k,\Msing)<\lambda^b) \le C \lambda^{b(d-1)} \, .
\end{align}
\end{jlist}
\end{lemma}
\begin{proof}
The first four items
follow from the above-mentioned properties of $f^j$.
For the fifth item, fix $k$ and let $l=d(k,\Msing)$.
If $l\ge \lambda^b$, then $\Fone(k)=1$, and the inequality holds trivially
for any $C_1\ge 1$.  If $l< \lambda^b$, then $l<(\lambda'_0)^b\le\vep_0$.
Since $\Msing$ is a compact set, there are $j$ and $k'\in M_j$, such that
$l=d(k,k')$.  In addition, there is a smooth path
$\gamma:[0,1]\to U_j$ from $k'$ to $k$ such that
$d(k,k')=\int_0^1 \rmd s |\gamma'(s)|$.
Then $\Fone(k)=\Fone(k)-\Fone(k')=\int_0^1 \rmd s \frac{\rmd}{\rmd s}
\Fone(\gamma(s))$.  Using the chain rule, and then applying the result in 
item 4,
shows that $\Fone (k)\le C_1 \lambda^{-b} l$ also in this case.
The last estimate follows by first estimating
$\Fzero(k)\le \1(d(k,\Msing)<\lambda^b)$, and then using the
compactness of the manifold and the fact that it has maximally
codimension $d-1$.
\qed \end{proof}

Now we are ready to define  the $3 d$-dimensional cut-off functions
introduced in Section \ref{sec:graphs}.  We define
$\PFone,\PFzero:(\T^d)^3 \to [0,1]$  by
\begin{align} 
\PFone(k_1,k_2,k_3) = \Fone (k_1+k_2)\Fone (k_2+k_3) \Fone (k_3+k_1)\, , 
\quad \PFzero=1-\PFone .
\end{align}
\begin{proofof}{Proposition \ref{th:PFcorr}}
Inequality (\ref{eq:PFzeroineq}) follows from the previous properties,
since for any $0\le a_i\le 1$, $i=1,2,3$, it holds that
$1-\prod_{i=1}^3 (1-a_i) \le a_1+a_2+a_3$.
The other points are obvious corollaries of Lemma \ref{th:Foneprop}.\qed
\end{proofof}

\subsection{From phases to resolvents}

The following result generalizes the standard formula used in connection with
time-dependent perturbation expansions.
\begin{theorem}\label{th:resolvents}
Let $I$ be a non-empty finite index set, assume $t>0$, and let 
$\gamma_i\in D$, 
$i\in I$, with $D\subset \C$ compact.  Suppose $A$ is a non-empty subset 
of $I$. We choose an additional time index label $*$, \itie , assume 
$*\not\in I$,
and let $A^{\rm c}=I\setminus A$, and $A'=A^{\rm c}\cup\set{*}$.  Then
for any path $\Gamma_D$ going once anticlockwise around $D$, we have
\begin{align}\label{eq:phtores}
&\int_{(\R_+)^{I}}\!\rmd s \,  \delta\Bigl(t-\sum_{i\in I} s_i\Bigr)
 \prod_{i\in I}\rme^{-\ci \gamma_i s_i}
\nonumber \\ & \quad
= -\oint_{\Gamma_D} \frac{\rmd z}{2\pi}
\int_{(\R_+)^{A'}}\!\rmd s \,  \delta\Bigl(t-\sum_{i\in A'} s_i\Bigr)
 \prod_{i\in A'} \left.\rme^{-\ci \gamma_i s_i}\right|_{\gamma_*=z}
\prod_{i\in A} \frac{\ci}{z-\gamma_i} \, .
\end{align}
\end{theorem}
\begin{proof}
Let us first consider the case $A=I$.  Then $A'=\set{*}$, and by definition the
``$s$-integral'' on the right hand side yields a factor $\rme^{-\ci z t}$. 
Therefore, in this case the formula is equal to the standard formula (whose
proof under the present assumptions can be found for instance from Lemma 4.9 in
\cite{ls05}).  If $A\ne I$, there is
$i_0\in A^{\rm c}$.  Resorting to the definition of the time-integration as an
integral over a standard simplex,
it is straightforward to prove that now
\begin{align}
& \int_{(\R_+)^{I}}\!\rmd s \,  \delta\Bigl(t-\sum_{i\in I} s_i\Bigr)
 \prod_{i\in I}\rme^{-\ci \gamma_i s_i}
\nonumber \\ & \quad
= \int_0^t  \!\rmd s_{i_0} \rme^{-\ci \gamma_{i_0} s_{i_0}}
\Bigl[\int_{(\R_+)^{I'}}\!\rmd s \,  \delta\Bigl(t-s_{i_0}-\sum_{i\in I'}
s_i\Bigr)
 \prod_{i\in I'}\rme^{-\ci \gamma_i s_i} \Bigr] \, ,
\end{align}
where $I'=I\setminus\set{i_0}$.
Therefore, we can perform an induction in the number of elements in 
$A^{\rm c}$, starting from $|A^{\rm c}|=0$.
Applying the above formula, induction assumption, and then Fubini's theorem
shows that (\ref{eq:phtores}) is valid for all $A$.
\qed \end{proof}

\subsection{Cluster combinatorics}

\begin{lemma}\label{th:clustercomb}
There is a constant $c$ such that for all $N>0$, $0<\lambda<\lambda_0$,
  \begin{align}\label{eq:ccmbbound}
 \sum_{S\in\pi(I_{N})}
 \prod_{A\in S}  \sup_{\Lambda,k,\sigma}|C_{|A|}(k,\sigma;\lambda,\Lambda)|
 \le c^{N} N!\, .
  \end{align}
If the sum is restricted to non-pairing $S$, then the bound can be improved by
a factor of $\lambda$.
\end{lemma}
\begin{proof}
Any $S$ which is non-pairing either has an odd cluster, or
contains a cluster of size of at least four.  If there is an odd cluster,
the corresponding $C_{|A|}$ term is zero, and thus any positive bound works
for them.  We cancel all partitions containing a singlet, and
use the bound in (\ref{eq:Cnbound}) for all clusters
which are not pairs.  As proven in Section \ref{sec:cumulants},
the constant can be adjusted so that
for pairs we can use (\ref{eq:Cnbound}) without the factor of $\lambda$.
Let $\pi'(I_N)$ consist of all partitions of $I_N$ which do not contain
singlets, \itie , of $S\in \pi(I_N)$
such that $|A|\ge 2$ for all $A\in S$.  Then
\begin{align} 
&  \sum_{\substack{S\in\pi(I_{N}),\\
     S\text{ not a pairing}}}
 \prod_{A\in S}  \sup_{k,\sigma}|C_{|A|}(k,\sigma;\lambda,\Lambda)|
\le  \lambda
\sum_{S\in\pi'(I_{N})}
 \prod_{A\in S} \left( (c_0)^{|A|} |A|!\right)
\end{align}
and
\begin{align} 
&
\sum_{S\in\pi(I_{N})}
 \prod_{A\in S}  \sup_{k,\sigma}|C_{|A|}(k,\sigma;\lambda,\Lambda)|
\le
\sum_{S\in\pi'(I_{N})}
 \prod_{A\in S} \left( (c_0)^{|A|} |A|!\right)
\nonumber \\ & \quad
\le (c_0)^{N} \sum_{m=1}^{\lfloor N/2\rfloor}
 \sum_{S\in\pi(I_{N})}   \1(|S|=m)
 \prod_{A\in S} |A|! \, .
\end{align}
A combinatorial computation along the proof of Lemma C.4 in \cite{ls05}
shows that
\begin{align} 
& \sum_{S\in\pi(I_{N})}
  \1(|S|=m) \prod_{A\in S} |A|!
=  \frac{N!}{m!}
 \sum_{n\in \N_+^m} \1\Bigl(\sum_{j=1}^m n_j = N\Bigr)
= \frac{N!}{m!} \binom{N-1}{m-1}
\nonumber \\ & \quad
\le \frac{N!}{m!} (N-1)^{m-1}\, .
\end{align}
The sum over $m$ from $1$ to $\infty$ of the last bound is bounded by
$N! \rme^{N}$.  This proves that (\ref{eq:ccmbbound}) holds
with $c=c_0 \rme$.
\qed \end{proof}

\subsection{Integrals over free momenta}
\label{sec:freeint}

\begin{proposition}\label{th:defendelta}
Suppose the assumption (DR\ref{it:DRdisp}) holds with constants
$C,\delta>0$, and assume $f \in \ell_1((\Z^d)^3)$.
Then for all
$s\in \R$, $k_0\in \T^d$, and $\sigma,\sigma'\in\set{\pm 1}$,
\begin{align}\label{eq:defendelta}
&\left| \int_{(\T^d)^2} \!\rmd k'\rmd k\,
 \rme^{\ci s (\omega(k) + \sigma' \omega(k')+ \sigma \omega(k_0-k-k'))}
  \FT{f}(k,k',k_0-k-k')
\right|
\le C\norm{f}_1 \sabs{s}^{-1-\delta}\, .
\end{align}
In particular,
\begin{align}\label{eq:leadphasebnd}
&\left| \int_{(\T^d)^2} \!\rmd k'\rmd k\,
 \rme^{\ci s (\omega(k) + \sigma' \omega(k')+ \sigma \omega(k_0-k-k'))}
\right| \le C\sabs{s}^{-1-\delta}\, .
\end{align}
\end{proposition}

In particular, the Proposition implies that $\Gamma(k_1)$  in
(\ref{eq:defGamma}) is well defined.
Adapting the proof of Proposition A.1 in \cite{ls05}, the
Proposition also shows that our assumptions on the dispersion relation 
$\omega$ guarantee that the map
\begin{align}
 F\mapsto \lim_{\beta\to 0^+}
    \int_{(\T^d)^2} \!\rmd k_2\rmd k_3\, \frac{\beta}{\pi}
\frac{1}{(\omega_1+\omega_2-\omega_3-\omega_4)^2+\beta^2}
F(k_2,k_3,k_1-k_2-k_3) \, ,
\end{align}
where $\omega_4=\omega(k_1-k_2-k_3)$,
defines for all $k_1\in \T^d$ a bounded positive Radon measure on $(\T^d)^3$
which we denote by
$\rmd k_2\rmd k_3\rmd k_4\,  \delta(k_1+k_2-k_3-k_4)
\delta(\omega_1+\omega_2-\omega_3-\omega_4)$.
In addition, if $F\in L^2((\T^d)^3)$ has summable Fourier transform, we also
have
\begin{align}
& \int_{(\T^d)^3} \!\rmd k_2\rmd k_3\rmd k_4\,
 \delta(k_1+k_2-k_3-k_4)
\delta(\omega_1+\omega_2-\omega_3-\omega_4)  F(k_2,k_3,k_4)
\nonumber \\ & \quad
= \int_{-\infty}^{\infty}\!\frac{\rmd s}{2\pi} \Bigl[ \int_{(\T^d)^3} \!\rmd
k_2\rmd k_3\rmd k_4\,
 \delta(k_1+k_2-k_3-k_4)
 \rme^{\ci s (\omega_1+\omega_2-\omega_3-\omega_4)}
 F(k_2,k_3,k_4) \Bigr] \, .
\end{align}
This gives a precise meaning to the ``energy conservation'' $\delta$-function
in (\ref{eq:Gamma2}), and proves the equality.
We wish to stress here that this $\delta$-function is a non-trivial
constraint, and can produce
non-smooth behavior even for smooth dispersion relations.

\begin{proofof}{Proposition \ref{th:defendelta}}
Since $f \in \ell_1((\Z^d)^3)$, we have as an absolutely convergent sum,
\begin{align}
\FT{f}(k,k',k_0-k-k') = \sum_{x_1,x_2,x_3\in \Z^d} \!\!
\rme^{-\ci 2\pi (k\cdot (x_1-x_3)+k'\cdot (x_2-x_3)+k_0\cdot x_3)} f(x_1,x_2,x_3)\, .
\end{align}
We insert this in the integrand and use Fubini's theorem to exchange the order
of $x$-sum and $k,k'$-integrals.
The resulting convolution integral over $k,k'$ can be expressed in terms of
$p_t(x)$, which is the inverse Fourier transform of
$k\mapsto\rme^{-\ci t \omega(k)}$.  This proves that
\begin{align}\label{eq:convsplit}
& \int_{(\T^d)^2} \!\rmd k'\rmd k\,
 \rme^{\ci s (\omega(k) + \sigma' \omega(k')+ \sigma \omega(k_0-k-k'))}
  \FT{f}(k,k',k_0-k-k')
\nonumber \\ & \quad
= \sum_{x_1,x_2,x_3\in \Z^d} f(x_1,x_2,x_3)
\nonumber \\ & \qquad\times
\sum_{y\in \Z^d}
\rme^{-\ci 2\pi k_0\cdot (y+x_3)}  p_{-s}(y+x_3-x_1)
p_{-\sigma' s}(y+x_3-x_2) p_{-\sigma s}(y)\, .
\end{align}
Thus by H\"{o}lder's inequality and the property
$\norm{p_{-s}}_3 = \norm{p_{s}}_3$,
its absolute value is bounded by
$\norm{f}_1\norm{p_{s}}_3^3\le \norm{f}_1 C\sabs{s}^{-1-\delta}$.
This proves (\ref{eq:defendelta}).  Equation (\ref{eq:leadphasebnd}) follows
then by applying the result to
$f(x_1,x_2,x_3)= \prod_{i=1}^3 \1(x_i=0)$.\qed
\end{proofof}

\begin{lemma}[Degree one vertex]\label{th:degoneest}
For any $k_0\in \T^d$,
$\alpha\in \R$, $|\beta|>0$, $0<\lambda\le \lambda'_0$,
and $\sigma,\sigma'\in\set{\pm 1}$,
\begin{align}\label{eq:degoneest}
\int_{\T^d} \!\rmd k\,
\frac{\Fone(\sigma' k_0)}{|\omega(k) + \sigma \omega(k_0-k)-\alpha 
+\ci \beta|}
\le C\lambda^{-b} \sabs{\ln |\beta|}^2\, ,
\end{align}
where $C$ depends only on $\omega$ and the basic cutoff function $\varphi$.
\end{lemma}
\begin{proof}
The left hand side of (\ref{eq:degoneest}) does not depend on the sign of
$\beta$, and thus it suffices to consider $\beta>0$.
The result holds trivially for any $C\ge 1$ if $|\beta|\ge 1$.
Furthermore, if
we change the integration variable from $k$ to $k'=\sigma' k$,
the left hand side becomes
\begin{align} 
& \int_{\T^d} \!\rmd k'\,
\frac{\Fone(\sigma' k_0)}{|\omega(\sigma' k') +
\sigma \omega(\sigma'(\sigma' k_0-k'))-\alpha +\ci \beta|}
 \nonumber \\ & \quad
=
\int_{\T^d} \!\rmd k\,
\frac{\Fone(\sigma' k_0)}{|\omega(k) +
\sigma \omega(\sigma' k_0-k)-\alpha +\ci \beta|} \, .
\end{align}
Thus it is enough to prove the theorem for $\sigma'=1$.

Let us thus assume $0<\beta\le 1$, $\sigma'=1$.
We apply Lemma \ref{th:Foneprop} to the left hand side,
which proves that it is then bounded by
\begin{align}
C_1 \lambda^{-b} d(k_0,\Msing)
\int_{\T^d} \!\rmd k\,
\frac{1}{|\omega(k) + \sigma \omega(k_0-k)-\alpha+\ci \beta|} \,.
\end{align}
In particular, if $d(k_0,\Msing)=0$ the left hand side is zero, and the
bound (\ref{eq:degoneest}) holds trivially.
Let us thus assume $k_0\not\in \Msing$.
By Lemma 4.21 in \cite{ls05}, for any real $r,\beta$,
\begin{align}\label{eq:restophases}
\frac{1}{|r+\ci \beta|}
 = \sabs{\ln\beta}
 \int_{-\infty}^{\infty} \!\rmd s\, \rme^{\ci s r} F(s;\beta)
\end{align}
where $F(s;\beta)\ge 0$ is such that
$F(s;\beta)\le \rme^{-\beta|s|} + \1(|s|\le 1) \ln |s|^{-1}$.
The bound is uniformly integrable in $s$.
Applying this representation and then Fubini's theorem shows that
\begin{align}
& \int_{\T^d} \!\rmd k\,
\frac{1}{|\omega(k) + \sigma \omega(k_0-k)-\alpha+\ci \beta|}
\nonumber \\ & \quad
\le \sabs{\ln \beta} \int_{-\infty}^{\infty} \!\rmd s\, F(s;\beta)
 \left| \int_{\T^d} \!\rmd k\,  \rme^{\ci s (\omega(k) + \sigma
     \omega(k_0-k))} \right|
\nonumber \\ & \quad
\le \sabs{\ln \beta}
\int_{-\infty}^{\infty} \!\rmd s\, F(s;\beta)
\frac{C\sabs{s}^{-1} }{d(k_0,\Msing)}
\nonumber \\ & \quad
\le \sabs{\ln \beta}\frac{C}{d(k_0,\Msing)}
\Bigl( \int_{-1}^{1} \!\rmd s\, (1+\ln |s|^{-1})
+ 2 \int_{1}^{\infty} \!\rmd s\, \frac{1}{s} \rme^{-\beta s} \Bigr)
\nonumber \\ & \quad
\le \sabs{\ln \beta}^2\frac{C'}{d(k_0,\Msing)}\, ,
\end{align}
where in the second inequality we have used assumption
(DR\ref{it:DRinterf}).  Collecting the estimates together yields the bound in
(\ref{eq:degoneest}).
\qed \end{proof}

\begin{lemma}[Degree two vertex]\label{th:degtwoest}
For any $k_0\in \T^d$, $\alpha\in \R$, $|\beta|>0$,
and $\sigma,\sigma'\in\set{\pm 1}$,
\begin{align}\label{eq:degtwoest}
\int_{(\T^d)^2} \!\rmd k'\rmd k\,
\frac{1}{|\omega(k) + \sigma' \omega(k')
+ \sigma \omega(k_0-k-k')-\alpha +\ci \beta|}
\le C \sabs{\ln |\beta|}\, ,
\end{align}
where $C$ depends only on $\omega$.
\end{lemma}
\begin{proof}
Again it suffices to consider $\beta>0$.
We apply the same representation of the resolvent term as in the proof of the
previous Lemma.  This shows that
\begin{align}
&\int_{(\T^d)^2} \!\rmd k'\rmd k\,
\frac{1}{|\omega(k) + \sigma' \omega(k')
+ \sigma \omega(k_0-k-k')-\alpha +\ci \beta|}
\nonumber \\ & \quad
\le \sabs{\ln \beta} \int_{-\infty}^{\infty} \!\rmd s\, F(s;\beta)
 \left| \int_{(\T^d)^2} \!\rmd k'\rmd k\,
 \rme^{\ci s (\omega(k) + \sigma' \omega(k')+ \sigma \omega(k_0-k-k'))}
\right|\, .
\end{align}
Applying Proposition \ref{th:defendelta} to the absolute value
shows that a constant $C$ for (\ref{eq:degtwoest})
can be found.
\qed \end{proof}

\section{Partially paired and higher order graphs}
\label{sec:higherorder}

In this section we consider relevant graphs which are either higher order,
when they necessarily contain a cluster $A'\in S$ with $|A'|\ge 4$, or they
are pairing and contain an interaction vertex of degree one.  We will show
that the contribution of these
graphs is negligible in all error terms and in the main term.
In addition, the related estimates will suffice to prove
that also all other contributions to the amputated and constructive
interference error terms are negligible.  We will use the notations introduced
in the earlier sections, in particular, in Section \ref{sec:iterclsuters}.

\begin{lemma}[Basic $\mathcal{A}$-estimate]\label{th:basicAest}
There is a constant $C>0$ such that for any (amputated)
momentum graph $\graph(S,J,n,\ell,n,\ell')$, $1\le n\le N_0$, and
$s>0$ we have
\begin{align}\label{eq:basicAest}
& \limsup_{\Lambda\to \infty}\,
\lambda^{2n} \sum_{\sigma,\sigma'\in \set{\pm 1}^{\mathcal{I}'_{n}} }
\1(\sigma_{n,1}=1)\1(\sigma'_{n,1}=-1)
\int_{(\Lambda^*)^{\mathcal{I}'_{n}}}  \!\rmd k\,
\int_{(\Lambda^*)^{\mathcal{I}'_{n}}}  \!\rmd k'\,
 \nonumber \\ & \quad \times
\Delta_{n,\ell}(k,\sigma;\Lambda) \Delta_{n,\ell'}(k',\sigma';\Lambda)
\prod_{A\in S} \delta_\Lambda\!\Bigl(\sum_{i\in A} K_i\Bigr)
\prod_{i=1}^{n} \Bigl[ \PFone(k_{i-1;\ell_i}) \PFone(-k'_{i-1;\ell'_i})
 \Bigr]
 \nonumber \\ & \quad \times
\left|
 \int_{(\R_+)^{I_{2,2n}}}\!\rmd r \,
\delta\Bigl(s-\sum_{i=2}^{2 n} r_i\Bigr)
 \prod_{i=2}^{2n} \rme^{-\ci r_i \gamma({i;J})}
\right|
 \nonumber \\ &
\le \rme^{s\lambda^2} \frac{(s\lambda^2)^{\tilde{n}_0-n'_0}
}{(\tilde{n}_0-n'_0)!}
\lambda^{2+\tilde{n}_2+(1-b)\tilde{n}_1-\tilde{n}_0} N_0^{-b_0 n'_0}
C^{1+\tilde{n}_1+\tilde{n}_2}
\sabs{\ln n} \sabs{\ln \lambda}^{1+\tilde{n}_2+2 \tilde{n}_1}
 \, ,
\end{align}
where $\tilde{n}_j=n_j-n_j(2)$ denotes the number of non-amputated interaction
vertices of degree $j$,
and $n'_0\ge 0$ counts the number of degree zero interaction vertices
$v_i$ with $2< i\le 2 n-N_0+1$.
\end{lemma}
In the above, $n'_0=0$, unless $n\ge (N_0+1)/2$, when
$n'_0=n_0(2 n-N_0+1)-n_0(2)$.
\begin{proof}
We first perform the sums over $\sigma$ and $\sigma'$, which have only one
non-zero term, the one with the appropriate propagation of parities.
We resolve the momentum constraints as explained in Section
\ref{sec:momdeltas}, \itie ,  we integrate out all the $\delta_\Lambda$-terms
using the time-ordered spanning tree.  We rewrite the remaining (free)
$k,k'$-integrals as in (\ref{eq:LamtoLeb}) and
thus convert all sums into Lebesgue
integrals over step functions. Since the resulting integrand is uniformly
bounded, using dominated convergence theorem we can take the limit
$\Lambda\to \infty$ inside the integrals. This proves the existence of the
limit, and the
resulting formula is, in fact, identical to the one obtained by
replacing in the left hand side of (\ref{eq:basicAest}) every $\Lambda^*$ by
$\T^d$
and all $\delta_\Lambda$ by $\delta=\delta_{\T^d}$.
However, we continue using the time-ordered resolution of momentum constraints
also after the continuum limit $\Lambda\to\infty$ has been taken.

There are total $N=2 n$ interaction vertices,
and let $A_j$, $j=0,1,2$, denote the collection of time slice indices
$2\le i< N$ such that $\deg v_{i+1} =j$.   Some of the sets can be empty,
but they are disjoint and their union is
$\set{2,3,\ldots,N -1}$.
Let further $B=\defset{i\in A_0}{i\le N-N_0}$ (which can be
empty).  For every $i\in B$ we thus have $\lfloor i/2\rfloor \le n-(N_0/2)$.
Set $\gamma_i=\gamma(i;J)$.  Then $-2 \kappa'\le \im \gamma_i \le 0$
and $|\re\gamma_i|\le 2 N\norm{\omega}_\infty$ for all $i$.
We can thus apply Lemma \ref{th:resolvents} with
$A=\set{N}\cup A_1\cup A_2$,
and using the path $\Gamma_N$ depicted in Fig.~\ref{fig:gnpath}.
Since then  $A'=\set{*}\cup A_0$, we find
\begin{align} 
& \left|  \int_{(\R_+)^{I_{2,2n}}}\!\rmd r \,
\delta\Bigl(s-\sum_{i=2}^{2 n} r_i\Bigr)
 \prod_{i=2}^{2n} \rme^{-\ci r_i \gamma_i}
\right|
\nonumber \\ & \quad
\le
\oint_{\Gamma_N} \frac{|\rmd z|}{2\pi}
\int_{(\R_+)^{A'}}\!\rmd r \,  \delta\Bigl(s-\sum_{i\in A'} r_i\Bigr)
\left|\rme^{-\ci r_* z}\right|
\prod_{i\in A_0} \left|\rme^{-\ci r_i  \gamma_i}\right|
\prod_{i\in A} \frac{1}{|z-\gamma_i|} \, .
\end{align}

If $i\in B$, we have $\im (-\gamma_i)= \kappa_{n-m} + \kappa_{n-m'}$
with $m+m'$ equal to $2 + J_+(i{-}2;J) + J_-(i{-}2;J)=i$.  Thus then
$\min(m,m')\le \lfloor i/2\rfloor \le n-(N_0/2)$, and therefore,
$\im (-\gamma_i) \ge \kappa_{n-\min(m,m')}=\kappa'=\lambda^2 N_0^{b_0}$.
In general, $\im (-\gamma_i) \ge 0$, and we obtain the estimates
\begin{align}\label{eq:intrtrick}
& \int_{(\R_+)^{A'}}\!\rmd r \,  \delta\Bigl(s-\sum_{i\in A'} r_i\Bigr)
\left|\rme^{-\ci r_* z}\right|
\prod_{i\in A_0} \left|\rme^{-\ci r_i  \gamma_i}\right|
\nonumber \\ & \quad
\le \rme^{s (\im z)_+} \int_{(\R_+)^{B}}\!\rmd r \,
\prod_{i\in B} \rme^{-\kappa' r_i}
 \int_{(\R_+)^{A'\setminus B}}\!\rmd r \,
\delta\Bigl(s-\sum_{i\in B} r_i-\sum_{i\in A'\setminus B} r_i\Bigr)
\nonumber \\ & \quad
\le \rme^{s (\im z)_+} (\kappa')^{-|B|} \frac{s^{\tilde{n}}}{\tilde{n}!} \, ,
\end{align}
where $(\cdot)_+$ was defined in (\ref{eq:softpiold}) and
$\tilde{n}=|A'\setminus B|-1=|A_0\setminus B|=|A_0|-|B|=\tilde{n}_0-n'_0$.
Since $\gamma_{2 n-2}=-\ci 2 \kappa_0=0$, this shows that
\begin{align} 
& \left| \int_{(\R_+)^{I_{2,2n}}}\!\rmd r \,
\delta\Bigl(s-\sum_{i=2}^{2 n} r_i\Bigr)
 \prod_{i=2}^{2n} \rme^{-\ci r_i \gamma_i}
\right|
\nonumber \\ & \quad
\le  (\kappa')^{-n'_0} \frac{s^{\tilde{n}_0-n'_0}}{(\tilde{n}_0-n'_0)!}
\oint_{\Gamma_N} \frac{|\rmd z|}{2\pi}
\frac{\rme^{s (\im z)_+}}{|z|}
\prod_{i\in A_1\cup A_2} \frac{1}{|z-\gamma_i|} \, .
\end{align}

\begin{figure}
\centering
  \myfigure{width=0.8\textwidth}{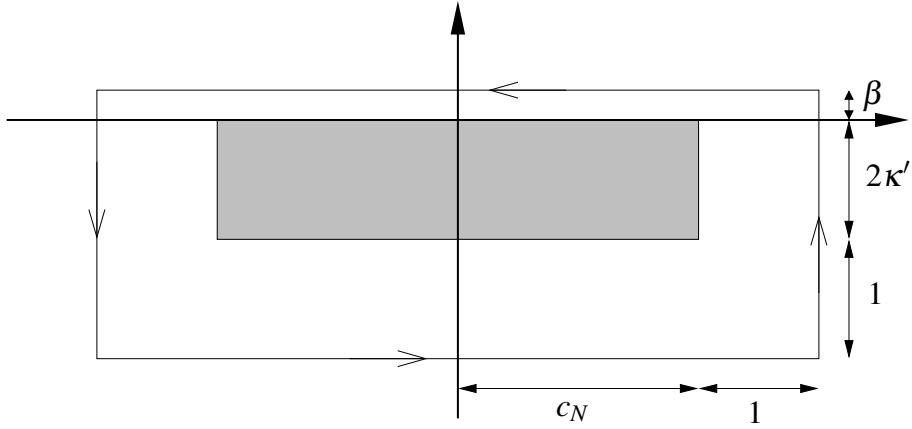}
  \medskip
\caption{Integration path $\Gamma_N$.  
Here $c_N=2 (N+1) \norm{\omega}_\infty$,
$\beta=\lambda^2$, and the shaded area contains all possible values of
$\gamma(i;J)$ for momentum graphs with $N$ interaction vertices.
\label{fig:gnpath}}
\end{figure}

We then estimate $\PFone(k_{0;\ell_1})\PFone(-k'_{0;\ell_1}) \le 1$
to remove the dependence on the two ``amputated'' interaction vertices
at the bottom of the interaction trees.  If there is any free momenta
associated with these vertices, they will be integrated over next, resulting
in an irrelevant factor of $1$.

Each of the resolvents in $\frac{1}{|z-\gamma_i|}$,
$i\in A_1\cup A_2$, depends only on the free momenta associated with edges
ending on a time slice $i'\ge i$.  Consider a degree one vertex in the 
plus-tree.
By Proposition \ref{th:momatintv}, there is a permutation $\pi$
of $\set{0,1,2}$
such that $\tilde{k}=k_{i-1,\ell_i+\pi(1)}$ is the free momentum,
and neither $k_{i-1,\ell_i+\pi(3)}$ nor
$k_0=k_{i-1,\ell_i+\pi(2)}+k_{i-1,\ell_i+\pi(1)}$
depend on $\tilde{k}$. We then estimate $\PFone(k_{i-1;\ell_i}) \le \Fone
(k_0)$.
Analogously, for every degree one vertex in the minus tree, we can estimate
$\PFone(-k'_{i-1;\ell'_i}) \le \Fone (-k'_0)$.
We remove all remaining $\PFone$, which are thus attached to a degree zero or
two vertex, by
the trivial estimate, $\PFone\le 1$.

After this we can use the estimates given in Lemmata
\ref{th:degoneest} and \ref{th:degtwoest} to iterate through the free momentum 
integrals in the direction of time,
\itie , from bottom to top in the graph.  At each iteration step, only one
resolvent factor depends on the corresponding free momenta, and the remaining
free momenta only affect the value of
``$\alpha$'' in the resolvent factor.  The estimates can be used with
$\beta=\lambda^2$
for those $z\in \Gamma_N$ in the top horizontal part of the path; we can
ignore the imaginary part of $\gamma_i$, since this is always negative, and
thus would only increase the ``$\beta$'' in the Lemmas
and lower the value of the resolvent factor.
For the remaining $z$
we have $|z-\gamma_i|\ge 1$, and the upper bounds remain valid also for these
values of $z$, after we adjust the constant
so that $C\ge 1$.
After the last iteration step, there is one free momentum integral left,
provided that there are any free momenta attached to the
top fusion vertex.  However, since the remaining integrand is
momentum-independent, this integral yields a trivial factor $1$, and can thus
be ignored.
The only remaining integral is over $z$.  This we estimate by
\begin{align} 
\oint_{\Gamma_N} \frac{|\rmd z|}{2\pi}
\frac{\rme^{s (\im z)_+}}{|z|}
\le C \rme^{s\lambda^2} \sabs{\ln N} \sabs{\ln \beta}\, ,
\end{align}
where $C$ is a constant which depends only on $\norm{\omega}_\infty$.
Collecting the estimates together yields the upper bound in
(\ref{eq:basicAest}); the power of $\lambda$ arising from the estimates
is $2 n-2(\tilde{n}_0-n'_0)-2 n'_0-b \tilde{n}_1$ which we have simplified to
$2+\tilde{n}_2+(1-b)\tilde{n}_1-\tilde{n}_0$ using
$2 n-2=\tilde{n}_0+\tilde{n}_1+\tilde{n}_2$.
\qed \end{proof}

For the following result, let us recall the definition of the time-dependent
exponents in a main term,
$\gamma(i)$ in equation (\ref{eq:defgammam}).  In the analysis of the partial
integration error term, Section \ref{sec:partierr}, we will need a
generalization of this phase factor to a case with interactions also in the
minus tree.  To this end we define
\begin{align}\label{eq:defgtilde}
\tilde\gamma({i;J})=\gamma^+_j+\gamma^-_{j'},\quad\text{with}\quad
j=J_+(i;J),\ j'=J_-(i;J)\, ,
\end{align}
where thus $j,j'\in\set{0,1,\ldots,n}$.  For the main term, with $n'=0$,
we have $J(i)=+1$, for all $i$, and $\gamma(i)=\tilde\gamma({i;1})$. 
Although the functional dependence on the mapping $J$ is different between
$\tilde{\gamma}$ and the amputated $\gamma$, in both cases the correct
exponential can be read off from our momentum graphs by summing over
$\sigma_e\omega(k_e)$, for all edges $e$ which intersect the corresponding
time slice.  Therefore, we will make no distinction between the amputated and
non-amputated exponentials and denote both by $\gamma(i;J)$.

\begin{lemma}[Basic $\mathcal{F}$-estimate]\label{th:basicFest}
There is a constant $C>0$ such that for any
momentum graph $\graph(S,J,n,\ell,n',\ell')$, $n,n'\ge 0$, and
$s>0$ we have
\begin{align}\label{eq:basicFest}
& \lim_{\Lambda\to \infty}
\lambda^{n+n'} \sum_{\sigma\in \set{\pm 1}^{\mathcal{I}'_{n}} }
\sum_{\sigma'\in \set{\pm 1}^{\mathcal{I}'_{\smash[t]{n'}}} }
\1(\sigma_{n,1}=1)\1(\sigma'_{n',1}=-1)
 \nonumber \\ & \quad \times
\int_{(\Lambda^*)^{\mathcal{I}'_{n}}}  \!\rmd k\,
\int_{(\Lambda^*)^{\mathcal{I}'_{\smash[t]{n'}}}}  \!\rmd k'\,
\Delta_{n,\ell}(k,\sigma;\Lambda) \Delta_{n',\ell'}(k',\sigma';\Lambda)
\prod_{A\in S} \delta_\Lambda\!\Bigl(\sum_{i\in A} K_i\Bigr)
 \nonumber \\ & \qquad \times
\prod_{i=1}^{n} \PFone(k_{i-1;\ell_i})
\prod_{i=1}^{n'} \PFone(-k'_{i-1;\ell'_i})
\left| \int_{(\R_+)^{I_{0,n+n'}}}\!\rmd r \,
\delta\Bigl(s-\sum_{i=0}^{n+n'} r_i\Bigr)
 \prod_{i=0}^{n+n'} \rme^{-\ci r_i \gamma({i;J})}\right|
 \nonumber \\ &
\le \rme^{s\lambda^2} \frac{(s\lambda^2)^{n_0} }{(n_0)!}
\lambda^{n_2+(1-b)n_1-n_0} C^{1+n_1+n_2}
\sabs{\ln (n+n'+1)} \sabs{\ln \lambda}^{1+n_2+2 n_1}
 \, ,
\end{align}
where $n_i$ denotes the number of interaction vertices of degree $i$.
\end{lemma}
\begin{proof}
There are altogether $N=n+n'$ interaction vertices, and
for $j=0,1,2$  we set
$A_j=\defset{0\le i< N}{\deg v_{i+1} =j}$  and 
$B=\emptyset$.
With these adjustments, we can derive the
bound as in the proof of Lemma \ref{th:basicAest}.  (Choosing $B=\emptyset$
implies $n'_0=0$ and the estimate thus
ignores any additional decay arising from factors with $\im \gamma_i<0$.)
The resulting power of $\lambda$ is $n+n'-2 n_0-b n_1$ which we have
simplified using $n+n'=n_0+n_1+n_2$.
\qed \end{proof}

\subsection{Amputated error term}
\label{sec:proveerramp}

\begin{proposition}\label{th:amperr}
Suppose $t>0$ and $0<\lambda<\lambda'_0$ are given, and define
$N_0$ and $\kappa$, as in Definition \ref{th:defkappaetc}.
There is a constant
$C>0$ depending only on $f$ and $g$, and
$c>0$ depending only on $\omega$ such that
  \begin{align}\label{eq:amperr}
 &  \limsup_{\Lambda\to\infty}   |Q^{\rm err}_{\rm amp}[g,f](t)|^2
\le C t^2 \rme^t \sabs{c t}^{N_0} N_0^{2 N_0+2} (4 N_0)!
\sabs{\ln \lambda}^{4 N_0+2}
\Bigl( \lambda + \lambda^{-2} N_0^{-b_0 N_0/4}\Bigr) \, ,
  \end{align}
as soon as $N_0(\lambda)\ge 56$.
\end{proposition}
Since $\gamma'\le 1$, we can apply the limits (\ref{eq:N0lima}) and 
(\ref{eq:N0limb}) here and conclude from the above bound that
$\limsup_{\Lambda\to\infty} |Q^{\rm err}_{\rm amp}[g,f](t)|\to 0$ as
$\lambda\to 0$.

\begin{proof}
By Proposition \ref{th:Qmainerr}, and according to the discussion in Section
\ref{sec:errors}, we have
\begin{align} 
&
\limsup_{\Lambda\to\infty} |Q^{\rm err}_{\rm amp}[g,f](t)|^2 \le C\norm{f}_2^2
t^2 \lambda^{-4}
\nonumber \\ & \times
\sup_{0\le s\le t\lambda^{-2}}
\limsup_{\Lambda\to\infty}\!\!\!
\sum_{J\text{ interlaces }(n-1,n-1)}
 \sum_{\ell,\ell' \in G_n}  \sum_{S\in \pi(I_{4 n+2})}
|\mathcal{A}_n^{\rm ampl}(S,J,\ell,\ell',s,\kappa)|\, ,
\end{align}
where $n=N_0$.  We have here first applied dominated convergence
to move the limit $\limsup_{\Lambda\to\infty}$ inside the $s$-integral which
was then estimated trivially.  The bound for domination is contained in the
following. 

The $\limsup$ can be bounded by first employing
(\ref{eq:Aamplsimp}) and then Lemma \ref{eq:basicAest}.
This yields the bound
\begin{align}
& \norm{\FT{g}}_\infty^2
\sum_{J\text{ interlaces }(n-1,n-1)}
 \sum_{\ell,\ell' \in G_n}
\sum_{S\in \pi(I_{4 n+2})}
\prod_{A\in S} \sup_{\Lambda,k,\sigma} |C_{|A|}(\sigma,k;\lambda,\Lambda)|
 \nonumber \\ & \qquad \times
 \rme^{s\lambda^2} \frac{(s\lambda^2)^{\tilde{n}_0-n'_0} }{(\tilde{n}_0-n'_0)!}
\lambda^{2+\tilde{n}_2-\tilde{n}_0+(1-b)\tilde{n}_1} n^{-b_0 n'_0}
C^{1+\tilde{n}_1+\tilde{n}_2}
\sabs{\ln n} \sabs{\ln \lambda}^{1+\tilde{n}_2+2 \tilde{n}_1}
\nonumber \\ & \quad
\le \norm{\FT{g}}_\infty^2
\sum_{J\text{ interlaces }(n-1,n-1)}
 \sum_{\ell,\ell' \in G_n}
\sum_{S\in \pi(I_{4 n+2})}
\prod_{A\in S} \sup_{\Lambda,k,\sigma} |C_{|A|}(\sigma,k;\lambda,\Lambda)|
 \nonumber \\ & \qquad \times
 \rme^{t} \sabs{c t}^{n}
\sabs{\ln \lambda}^{2 + 4 n}
\lambda^{2+\tilde{n}_2-\tilde{n}_0+(1-b)\tilde{n}_1} n^{-b_0 n'_0}\, .
\end{align}
We have used here
$\tilde{n}_0\le n_0 \le n$ which is implied by Lemma \ref{th:ivdegrees}.
We set $r=2n+1-|S|$ as in Proposition \ref{th:ivordering}, and note that here
$\tilde{n}_j=n_j-n_j(2)\ge 0$ and $n'_0=n_0(n{+}1)-n_0(2)$.
By Lemma \ref{th:ivdegrees},
$\tilde{n}_2-\tilde{n}_0= r+n_0(2)-n_2(2)
\ge r-2$, and $\tilde{n}_1\ge n_1-2$.  If $r+n_1\ge 24$,
we thus have
$\tilde{n}_2-\tilde{n}_0+(1-b)\tilde{n}_1\ge r +\frac{1}{4} n_1 -3\ge 3$.
In such cases we estimate
$\lambda^{2+\tilde{n}_2-\tilde{n}_0+(1-b)\tilde{n}_1} n^{-b_0 n'_0}\le
\lambda^5$.

If $r+n_1< 24$, we get from Proposition \ref{th:ivordering},
the estimates $n'_0\ge (n-n_1-r)/2-2$ and $\tilde{n}_2-\tilde{n}_0\ge 0$.
Thus for
any $n\ge 56$,  we have
$n'_0\ge  (n/2)-14\ge n/4$, and therefore also
$\lambda^{2+\tilde{n}_2-\tilde{n}_0+(1-b)\tilde{n}_1} n^{-b_0 n'_0}\le
\lambda^2 n^{-b_0 n/4}$.

The number of terms in the sum over $J$
is less than $2^{2 n-2}$ and the number of terms in the sum over $\ell$ 
is equal to 
$( 2n{-}1)!!\le (2 n)^n$.
For the sum over $S$ we can apply Lemma \ref{th:clustercomb}.
Collecting all the estimates together proves that (\ref{eq:amperr}) holds,
after readjustment of the constants $c$ and $C$.
\qed \end{proof}

\subsection{Constructive interference error terms}
\label{sec:proveerrcut}

\begin{proposition}\label{th:cuterr}
Suppose $t>0$ and $0<\lambda<\lambda'_0$ are given, and define
$N_0$ and $\kappa$, as in Definition \ref{th:defkappaetc}.
There is a constant
$C>0$ depending only on $f$ and $g$, and
$c>0$ depending only on $\omega$ such that
  \begin{align}\label{eq:cuterr}
 &  \limsup_{\Lambda\to\infty}   |Q^{\rm err}_{\rm cut}[g,f](t)|^2
\le C t^2 \rme^t \sabs{c t}^{N_0} N_0^{2 N_0+4} (4 N_0)!
\sabs{\ln \lambda}^{4 N_0+3} \lambda^{1/4} \,.
  \end{align}
\end{proposition}
Since $\gamma'\le \frac{1}{4}$, this estimate proves that
$\limsup_{\Lambda\to\infty} |Q^{\rm err}_{\rm amp}[g,f](t)|\to 0$ as
$\lambda\to 0$.

\begin{proof}
We again denote
$\mathcal{I}'_{n} =\mathcal{I}_{n} \cup \set{(n,1)}=
\defset{(i,j)}{0\le i\le n, 1\le j \le m_{n-i}}$,
and define $I=I_{2 m_n}=I_{4 n +2}$ to give labels to the final $\FT{a}$,
as before.  By expanding the pairing truncations in $\Ptrunc$
to individual components, and then applying the cluster expansions in
Lemma \ref{th:cumulants}, we find that the effect of the extra terms in the
truncations is the cancellation of all those terms from the main cumulant
expansion which contain one of the corresponding pairings.  None of the other 
clusterings is affected.  The remainder of the analysis is completely
analogous to that used for $\mathcal{A}_{n}$ in Section \ref{sec:errors}, and
it shows that
\begin{align} 
&\E\Bigl[|\mean{\FT{g},
\mathcal{Z}_n(s,t/\vep,\cdot,+1,\kappa)[\FT{a}_s]}|^2\Bigr]
\nonumber \\ & \quad
= \sum_{J\text{ interlaces }(n-1,n-1)} \
\sum_{\ell,\ell' \in G_n}\ \sum_{S\in \pi(I_{4 n +2})}
\mathcal{Z}_n^{\rm ampl}(S,J,\ell,\ell',t/\vep-s,\kappa)\, .
\end{align}
Here
$\mathcal{Z}_n^{\rm ampl}(S,J,\ell,\ell',s,\kappa)=0$
for any partition $S$ which contains a pairing of
any two edges attached to one of the ``truncated'' vertices
(\itie , if there are $A\in S$ and $i\in \set{1,2}$ such that $|A|=2$ and
$\edges_+(u)\subset\edges_-(v_i)$ for every $u\in A$).
For all other $S$ we have
\begin{align}\label{eq:Znampl}
& \mathcal{Z}_n^{\rm ampl}(S,J,\ell,\ell',s,\kappa)
=  (-\lambda^{2})^{n} \sum_{\sigma,\sigma'\in \set{\pm 1}^{\mathcal{I}'_{n}} }
\int_{(\Lambda^*)^{\mathcal{I}'_{n}}}  \!\rmd k\,
\int_{(\Lambda^*)^{\mathcal{I}'_{n}}}  \!\rmd k'\,
 \nonumber \\ & \quad \times
\Delta_{n,\ell}(k,\sigma;\Lambda) \Delta_{n,\ell'}(k',\sigma';\Lambda)
\prod_{A\in S}\Bigl[ \delta_\Lambda\!\Bigl(\sum_{i\in A} K_i\Bigr)
  C_{|A|}(o_A,K_A;\lambda,\Lambda) \Bigr]
 \nonumber \\ & \quad \times
 \sigma_{1,\ell_{1}}
 \PFzero(k_{0;\ell_1}) \sigma'_{1,\ell'_{1}}
 \PFzero(-k'_{0;\ell'_1})
\prod_{i=2}^{n} \Bigl[ \sigma_{i,\ell_{i}}
 \PFone(k_{i-1;\ell_i}) \sigma'_{i,\ell'_{i}}
 \PFone(-k'_{i-1;\ell'_i})
 \Bigr]
 \nonumber \\ & \quad \times
\1(\sigma_{n,1}=1)\1(\sigma'_{n,1}=-1) |\FT{g}(k_{n,1})|^2
  \int_{(\R_+)^{I_{2,2n}}}\!\rmd r \,
\delta\Bigl(s-\sum_{i=2}^{2 n} r_i\Bigr)
 \prod_{i=2}^{2n} \rme^{-\ci r_i \gamma({i;J})}
\, ,
\end{align}
where $\gamma({i;J})$ is given by (\ref{eq:giJ}).  We use the Schwarz
inequality in the sum over $n$, and then proceed as in the proof of
Proposition \ref{th:amperr}.  Then using Proposition \ref{th:Qmainerr} yields
the estimate 
\begin{align} 
&
\limsup_{\Lambda\to\infty} |Q^{\rm err}_{\rm cut}[g,f](t)|^2 \le
N_0 C\norm{f}_2^2 t^2 \lambda^{-4}
\sup_{0\le s\le t\lambda^{-2}}
  \sum_{n=1}^{N_0}
\sum_{J\text{ interlaces }(n-1,n-1)}
\nonumber \\ & \qquad \times
 \sum_{\ell,\ell' \in G_n}  \sum_{S\in \pi(I_{4 n+2})}
\limsup_{\Lambda\to\infty}  |\mathcal{Z}_n^{\rm
ampl}(S,J,\ell,\ell',s,\kappa)|\, .
\end{align}

Compared to the amputated error term, the terms in the sum have the following
improved
upper bounds whose proof will be given at
the end of this section.
\begin{lemma}[Basic $\mathcal{Z}$-estimate]\label{th:basicZest}
If we change the term
$\PFone(k_{0;\ell_1}) \PFone(-k'_{0;\ell'_1})$ on
the left hand side of (\ref{eq:basicAest})
to $\PFzero(k_{0;\ell_1}) \PFzero(-k'_{0;\ell'_1})$,
the estimate can either be improved
by a factor of $C \sabs{\ln \lambda} \lambda^{z_0}$, $z_0=bd-2$,
or the corresponding $\mathcal{Z}_n^{\rm ampl}$ is zero.
If $\tilde{n}_1=0$, then the estimate is valid also with
$z_0=(d-2) b$.  If any of the amputated vertices has degree two,
the estimate is valid with  $z_0=(d-1) b$.
\end{lemma}

The $\limsup$-factor can then be bounded by first employing
(\ref{eq:Znampl}) and then using Lemma \ref{th:basicZest}.
We neglect the extra decay provided by the partial time integration, and
estimate $n^{-b_0 n'_0}\le 1$.  Simplifying the expression somewhat along the
lines used in the proof of Proposition \ref{th:amperr}, we thus have
the following bound for the  $\limsup$-term:
\begin{align}
& \norm{\FT{g}}_\infty^2
 \sum_{n=1}^{N_0} \sum_{J\text{ interlaces }(n-1,n-1)}
 \sum_{\ell,\ell' \in G_n}
\sum_{S\in \pi(I_{4 n+2})}
\prod_{A\in S} \sup_{\Lambda,k,\sigma} |C_{|A|}(\sigma,k;\lambda,\Lambda)|
 \nonumber \\ & \qquad \times
 \rme^{t} \sabs{c t}^{n}
\sabs{\ln \lambda}^{3 + 4 n}
\lambda^{2+\tilde{n}_2-\tilde{n}_0+(1-b)\tilde{n}_1+db}
 \nonumber \\ & \qquad \times
\1(\mathcal{Z}_n^{\rm ampl}\ne 0)\times\begin{cases}
\lambda^{-b}, & \text{if }n_2(2)>0;\\
\lambda^{-2 b}, & \text{if }n_2(2)=0, \tilde{n}_1=0;\\
\lambda^{-2}, & \text{otherwise.}\\
\end{cases}
\end{align}
By Proposition \ref{th:ivordering},
here always  $\tilde{n}_2-\tilde{n}_0\ge 0$ and $\tilde{n}_1\ge 0$.
Thus if $n_2(2)>0$, the power of $\lambda$ can be bounded from above by
$\lambda^{2+b(d-1)}\le \lambda^{4+\frac{1}{4}}$.

We can thus assume that $n_2(2)=0$.  Then
$\tilde{n}_2-\tilde{n}_0=n_2-n_0-n_2(2)+n_0(2)=r+n_0(2)$.
If also $n_0(2)=0$, then necessarily $n_1(2)=2$, \itie , that
both amputated interaction vertices have exactly one free
momentum attached to them.  By the iterative cluster scheme,
this implies that there is a cluster
$A_0\in S$ such that exactly two of the edges in the first interaction
vertex, the amputated minus vertex, connect to it.
If $A_0$ is a pairing, then $\mathcal{Z}_n^{\rm ampl}=0$ by
definition. Otherwise, $|A_0|\ge 4$, and thus then $r\ge 1$.
Therefore, in all cases we can conclude that now
either $\mathcal{Z}_n^{\rm ampl}=0$ or
$\tilde{n}_2-\tilde{n}_0\ge 1$.

Therefore, if $n_2(2)=0$ and $\tilde{n}_1=0$,
the power of $\lambda$ is bounded by
$\lambda^{2+1+b(d-2)}\le \lambda^{4+\frac{1}{2}}\le \lambda^{4+\frac{1}{4}}$.
If $n_2(2)=0$ and $\tilde{n}_1>0$, then the bound
$\lambda^{2+1+1-b+db-2}=\lambda^{2+b(d-1)}\le \lambda^{4+\frac{1}{4}}$
can be used.  Therefore, whatever the clustering, a bound
$\lambda^{4+\frac{1}{4}}$ is always available.
The rest of the sums can be bounded exactly as in the
proof of Proposition \ref{th:amperr}, apart from the first sum over $n$
which yields an additional factor $N_0$.
Collecting all the estimates together proves that (\ref{eq:cuterr}) holds.
\qed \end{proof}

\begin{proofof}{Lemma \ref{th:basicZest}}
The statement with $z_0=0$
is a corollary of the proof of Lemma \ref{th:basicAest},
since the estimate $\PFzero(k_{0;\ell_1}) \PFzero(-k'_{0;\ell'_1})\le 1$
allows to remove these terms at the right place in the proof.
However, we can improve on the estimate by using the fact that
$\PFzero(k_{0;\ell_1})$ enforces particular sums of momenta to lie very close
to the singular manifold.

Let us first consider the case were one of the amputated vertices has degree
two. If this is the amputated minus vertex, we find from
Proposition \ref{th:PFcorr} that
\begin{align} 
&\PFzero(-k'_{0;\ell'_1})
\le  \sum_{e_1,e_2 \in \edges_-(v_1);\ e_1< e_2}
\1\!\left(d(-(k_{e_1}+k_{e_2}),\Msing)<\lambda^b\right)\, .
\end{align}
By Lemma \ref{th:nokkdiff}, for any choice of $e_1$, $e_2$
here $-(k_{e_1}+k_{e_2})$ depends on the free momenta $k,k'$ of $v_1$
either as $k$, $k'$, or $-(k+k')$.  Thus when we first integrate over $k$ and
then over $k'$, in one of the integrals
we can apply Lemma \ref{th:Foneprop}, according to which
\begin{align}\label{eq:Fzerovol}
\sup_{k_0\in \T^d} \int_{\T^d} \rmd k\, \1\!\left(d(\pm
k+k_0,\Msing)<\lambda^b\right)
\le C \lambda^{b(d-1)} \, .
\end{align}
After this we can estimate $\PFzero(k_{0;\ell_1})\le 1$, and then continue as
in the proof of Lemma \ref{th:basicAest}.

If the amputated minus vertex does not have degree two, we estimate
trivially $\PFzero(-k'_{0;\ell'_1})\le 1$ and integrate over any free momenta
attached to it, which yields an irrelevant factor of one for the iterative
bound.  We then estimate the extra factor in the amputated plus vertex by
\begin{align}\label{eq:ampv2est}
&\PFzero(k_{0;\ell_1})
\le  \sum_{e_1,e_2 \in \edges_-(v_2);\ e_1< e_2}
\1\!\left(d(k_{e_1}+k_{e_2},\Msing)<\lambda^b\right)\, .
\end{align}
If the amputated plus vertex, $v_2$, has degree two, we again
gain a factor $C\lambda^{b(d-1)}$ from performing the two free integrations
attached to it.  After this the proof can proceed as in Lemma
\ref{th:basicAest}.
Thus if either of the amputated vertices has degree two, we have
proven a gain by the stated factor with $z_0=b(d-1)$.
This proves the last statement made in the Lemma.

To prove the other two statements, it is sufficient to consider
the term corresponding to some fixed pair $e_1< e_2$, 
$e_1,e_2\in \edges_-(v_2)$ in the bound (\ref{eq:ampv2est}).
If $k_{e_1}+k_{e_2}$ is independent of all free momenta, then by Proposition
\ref{th:kkpl}, the two initial time vertices belonging to $e_1\cup e_2$
must be isolated from the rest of the initial time vertices.  This is possible
only if they are paired, but then
$\mathcal{Z}_n^{\rm ampl}=0$ by definition.  Thus we can assume
that there is some free momenta on which $k_{e_1}+k_{e_2}$ depends.
Of the corresponding free edges,
let $f_0$ denote the one added first (\itie , it is the maximum in the ordering
of edges).

We next estimate all $\PFone$-factors as in the proof of Lemma
\ref{th:basicAest}, with one exception:
if the fusion vertex at which $f_0$ ends is a degree two interaction vertex,
we will need the corresponding
$\PFone$-factor, and this is kept unchanged.
Then we use the estimates in the proof of Lemma \ref{th:basicAest}, and
iterate through the interaction
vertices until the vertex at which $f_0$ ends is reached.
If $f_0$ is attached to either an
amputated interaction vertex or the top fusion vertex, then using
(\ref{eq:Fzerovol})
we gain an improvement with $z_0=b(d-1)$, which is the best bound of all
the three possibilities.

If $f_0$ is attached to a degree one non-amputated
interaction vertex,
we first remove the corresponding ``resolvent'' factor
using the trivial $L^\infty$
estimate and then apply (\ref{eq:Fzerovol}).  Since in the proof of Lemma
\ref{th:basicAest} this term would be estimated by
$C \sabs{\ln \lambda}^2 \lambda^{-b}$, we gain
an improvement by a factor of $C \lambda^{b-2+b(d-1)}$,
as compared to Lemma \ref{th:basicAest}.
This yields the worst bound with $z_0=b d-2$.

Otherwise,
$f_0$ is attached to a non-amputated interaction vertex
of degree two.  Let the two free momenta be denoted by $k_1$ and $k_2$,
and the third integrated momenta by $k_3$.  In addition, denote
$k_0=k_1+k_2+k_3$ which is independent of $k_1$ and $k_2$.
By Lemma \ref{th:nokkdiff}, $k_{e_1}+k_{e_2}=\pm k_i+k_0'$, for some
choice of sign and $i\in \set{1,2,3}$, where  $k_0'$
is independent of $k_1$, $k_2$.  Thus we need to consider
\begin{align}\label{eq:degtwoest2}
&\int_{(\T^d)^2} \!\rmd k_1\rmd k_2\,
\1\!\left(d(\pm k_i+k_0',\Msing)<\lambda^b\right)
 \nonumber \\ & \quad \times
\frac{\PFone(\pm(k_1,k_2,k_0-k_1-k_2))}{|\omega(k_1) + \sigma' \omega(k_2)
+ \sigma \omega(k_0-k_1-k_2)-\alpha +\ci \beta|}\, .
\end{align}
(The $\PFone$-factor is present here, as the only one which was not estimated
trivially. 
$\PFone$ is also clearly invariant under permutation of its arguments.)
If $i=1$, we estimate $\PFone\le \Fone(\pm(k_0-k_1))$, and
change the integration variable $k_1$ to $k=\pm k_1+k_0'$.
Then first applying Lemma \ref{th:degoneest}  to the $k_2$
integral and then (\ref{eq:Fzerovol}) to the $k$ integral,
we find that the integral is bounded by
$C\sabs{\ln \beta}^2\lambda^{-b+b(d-1)}$.   Analogous chance of variables can
be done to show that the bound is valid also if $i=2$ or $i=3$.
Thus we get an improvement by
a factor of
$C\sabs{\ln \beta}\lambda^{b(d-2)}$ compared to the estimate used in the
proof of Lemma \ref{th:basicAest}.

After one of the above estimates, we can finish the iteration of the
interaction vertices, and complete
the rest of the proof as in Lemma \ref{th:basicAest}.
If $\tilde{n}_1=0$, then there are no non-amputated
degree one vertices.  Thus either of the two better
estimates apply, and $z_0=b (d-2)> bd-2$ can be used.
Otherwise, we need to resort to the worst estimate
with $z_0=bd-2$. This completes the proof of the Lemma.\qed
\end{proofof}

\subsection{Partial time-integration error terms} \label{sec:partierr}


\begin{proposition}\label{th:ptierr1}
Suppose $t>0$ and $0<\lambda<\lambda'_0$ are given, and define
$N_0$ and $\kappa$, as in Definition \ref{th:defkappaetc}.
There is a constant
$C>0$ depending only on $f$ and $g$, and
$c>0$ depending only on $\omega$ and $\lambda'_0$ such that
  \begin{align}\label{eq:ptierr1}
 &  \limsup_{\Lambda\to\infty}   | Q^{\rm err}_{\rm pti}[g,f](t)|^2
\le C t^2 \rme^t \sabs{c t}^{N_0} N_0^{2 N_0+5+2 b_0} (4 N_0)!
\sabs{\ln \lambda}^{4 N_0+2} \lambda^{1/4}
\nonumber \\ & \quad
+ C t^2 N_0^{2+2 b_0}
\sup_{\substack{0\le s\le t\lambda^{-2},\\ N_0/2\le n < N_0}} \ 
\sum_{J\text{ interlaces }(n,n)} \
\sum_{\ell,\ell' \in G_n}\ \sum_{S\in \pi(I_{4 n +2})}
|\mathcal{G}_n^{\rm pairs}(S,J,\ell,\ell',s,\kappa)|
 \, ,
\end{align}
where $\mathcal{G}_n^{\rm pairs}(S,J,\ell,\ell',s,\kappa)=0$, if the
graph defined by $J,S,\ell,\ell'$ is not fully paired, and otherwise
it is equal to
\begin{align}\label{eq:defGnpairs}
& \mathcal{G}_n^{\rm pairs}(S,J,\ell,\ell',s,\kappa)
=  (-\lambda^{2})^{n} \sum_{\sigma,\sigma'\in \set{\pm 1}^{\mathcal{I}'_{n}} }
\int_{(\T^d)^{\mathcal{I}'_{n}}}  \!\rmd k\,
\int_{(\T^d)^{\mathcal{I}'_{n}}}  \!\rmd k'\,
\Delta_{n,\ell}(k,\sigma) \Delta_{n,\ell'}(k',\sigma')
 \nonumber \\ & \quad \times
\prod_{A=\set{i,j}\in S}
\Bigl[ \delta(K_i+K_j) \1(o_i=-o_j) W(K_i) \Bigr]
\prod_{i=1}^{n} \Bigl[ \sigma_{i,\ell_{i}}
 \PFone(k_{i-1;\ell_i}) \sigma'_{i,\ell'_{i}}
 \PFone(-k'_{i-1;\ell'_i})
 \Bigr]
 \nonumber \\ & \quad \times
\1(\sigma_{n,1}=1)\1(\sigma'_{n,1}=-1) |\FT{g}(k_{n,1})|^2
\int_{(\R_+)^{I_{0,2n}}}\!\rmd r \,
\delta\Bigl(s-\sum_{i=0}^{2 n} r_i\Bigr)
 \prod_{i=0}^{2n} \rme^{-\ci r_i \gamma({i;J})} \, .
\end{align}
\end{proposition}
\begin{proof}
The error term $\mathcal{G}_n$ was defined in (\ref{eq:defGnviaA}),
where we can directly apply (\ref{eq:gAnprod}). Comparing the result to the
definition of $\mathcal{F}_n$ in (\ref{eq:defFn}) shows that
\begin{align} 
\mean{\FT{g},\mathcal{G}_{n}(s,t,\cdot,1,\kappa)[\FT{a}_s]}
 = \mean{\FT{g},\rme^{\ci  s \omla}
\mathcal{F}_n(t-s,\cdot,1,\kappa)[\FT{\psi}_s]} \, .
\end{align}
Thus in this case
\begin{align} 
&\E\Bigl[|\mean{\FT{g},
\mathcal{G}_n(s,t/\vep,\cdot,1,\kappa)[\FT{a}_s]}|^2\Bigr]
= \E\Bigl[|\mean{
\rme^{-\ci  s \omla} \FT{g},
\mathcal{F}_n(t/\vep-s,\cdot,1,\kappa)[\FT{\psi}_0]}|^2\Bigr]
\nonumber \\ & \quad
= \sum_{J\text{ interlaces }(n,n)} \
\sum_{\ell,\ell' \in G_n}\ \sum_{S\in \pi(I_{4 n +2})}
\mathcal{G}_n^{\rm ampl}(S,J,\ell,\ell',t/\vep-s,\kappa)\, ,
\end{align}
where
\begin{align}\label{eq:defGnampl}
& \mathcal{G}_n^{\rm ampl}(S,J,\ell,\ell',s,\kappa)
=  (-\lambda^{2})^{n} \sum_{\sigma,\sigma'\in \set{\pm 1}^{\mathcal{I}'_{n}} }
\int_{(\Lambda^*)^{\mathcal{I}'_{n}}}  \!\rmd k\,
\int_{(\Lambda^*)^{\mathcal{I}'_{n}}}  \!\rmd k'\,
 \nonumber \\ & \quad \times
\Delta_{n,\ell}(k,\sigma;\Lambda) \Delta_{n,\ell'}(k',\sigma';\Lambda)
\prod_{A\in S}\Bigl[ \delta_\Lambda\!\Bigl(\sum_{i\in A} K_i\Bigr)
  C_{|A|}(o_A,K_A;\lambda,\Lambda) \Bigr]
 \nonumber \\ & \quad \times
\prod_{i=1}^{n} \Bigl[ \sigma_{i,\ell_{i}}
 \PFone(k_{i-1;\ell_i}) \sigma'_{i,\ell'_{i}}
 \PFone(-k'_{i-1;\ell'_i})
 \Bigr]
\1(\sigma_{n,1}=1)\1(\sigma'_{n,1}=-1) |\FT{g}(k_{n,1})|^2
 \nonumber \\ & \quad \times
\int_{(\R_+)^{I_{0,2n}}}\!\rmd r \,
\delta\Bigl(s-\sum_{i=0}^{2 n} r_i\Bigr)
 \prod_{i=0}^{2n} \rme^{-\ci r_i \gamma({i;J})}\, .
\end{align}
Thus the amplitude differs from the
``amputated'' amplitudes $\mathcal{A}_n$ and $\mathcal{Z}_n$ by containing
also the propagators associated with the first two interactions.
We recall the discussion about the
definition of the non-amputated exponentials $\gamma(i;J)$ in
(\ref{eq:defgtilde}).

As in the proof of Proposition \ref{th:cuterr}, we then find
\begin{align} 
&
\limsup_{\Lambda\to\infty} |Q^{\rm err}_{\rm pti}[g,f](t)|^2 \le N_0^2
(\kappa')^2
C\norm{f}_2^2 t^2 \lambda^{-4}
\sup_{0\le s\le t\lambda^{-2}, N_0/2\le n \le N_0}
\nonumber \\ & \quad \times
\sum_{J\text{ interlaces }(n,n)} \
 \sum_{\ell,\ell' \in G_n}  \sum_{S\in \pi(I_{4 n+2})}
\limsup_{\Lambda\to\infty}
|\mathcal{G}_n^{\rm ampl}(S,J,\ell,\ell',s,\kappa)|\, .
\end{align}
For any graph which is not fully paired, we take absolute values up to the
time-integrations and apply Lemma \ref{th:basicFest}.
Using the notations of Lemma \ref{th:ivdegrees}, then
\begin{align} 
& |\mathcal{G}_n^{\rm ampl}(S,J,\ell,\ell',s,\kappa)|
\nonumber \\ & \quad
\le
\norm{\FT{g}}_\infty^2
 \rme^{t} \frac{t^{n_0} }{(n_0)!} \lambda^{r+(1-b) n_1} C^{2 n +1}
\sabs{\ln \lambda}^{2+4 n}
\prod_{A\in S} \sup_{\Lambda,k,\sigma} |C_{|A|}(\sigma,k;\lambda,\Lambda)|\, .
\end{align}
If the graph is higher order, $r\ge 1$, and if the graph is not fully paired,
$n_1\ge 1$.  Thus for both types of graphs, and trivially for irrelevant
graphs, we can use a bound
\begin{align} 
 \lambda^{1/4}  \norm{\FT{g}}_\infty^2
 \rme^{t} \sabs{c t}^{n}
\sabs{\ln \lambda}^{2+4 n}
\prod_{A\in S} \sup_{\Lambda,k,\sigma} |C_{|A|}(\sigma,k;\lambda,\Lambda)|\, .
\end{align}

Consider then a fully paired graph.  In this case, all clusters are pairings
with 
$C_2((\sigma',\sigma),(k',k))=\1(\sigma'+\sigma=0) W^\lambda_\Lambda(\sigma k)$.
By Lemma \ref{th:unifW2}, $W^\lambda_\Lambda(k)= W(k)+\Delta$,
where $\limsup_{\Lambda\to \infty}\sup_{k\in \T^d} |\Delta|\le  2 c_0^2
\lambda$. Since $W(-k)=W(k)$, we have for any finite index set $I$ 
\begin{align} 
\limsup_{\Lambda\to\infty}
\Bigl|\prod_{i\in I} W^\lambda_\Lambda(\pm k_i)-\prod_{i\in I} W(k_i)\Bigr|\le
|I| C^{|I|-1} 2 c_0^2 \lambda \, ,
\end{align}
where $C=2 c_0^2 \lambda'_0+\norm{W}_\infty<\infty$.
(The statement can be proven by induction in $|I|$.)

If $S$ is a pairing, we have $|S|=2 n +1$, and
applying Lemma \ref{th:basicFest}, we can thus exchange in the definition of
$\mathcal{G}_n^{\rm ampl}$ all $C_2$ terms by $W(K_i)\1(o_i=-o_j)$
with an error whose $\limsup$ is bounded by
\begin{align} 
C n \norm{\FT{g}}_\infty^2
 \rme^{t} \sabs{c t}^{n} \sabs{\ln \lambda}^{2+4 n} \lambda \, .
\end{align}
In the resulting formula, we first resolve all $\delta_\Lambda$-functions
as explained before.  The summations over the free momenta are then turned
into Lebesgue integrals as explained in Section \ref{sec:firstpf}, with an
integrand
which is uniformly bounded and has a pointwise limit when
$\Lambda\to \infty$. By dominated convergence we can thus take
the limit $\Lambda\to \infty$ inside the integrals which shows that the limit
is given by $\mathcal{G}_n^{\rm pairs}$ defined in (\ref{eq:defGnpairs}).

Now collecting all the above bounds together, estimating the number of terms
in the $J,\ell,\ell',S$ sums as before,
readjusting $c$ and $C$ whenever necessary,
and using $N_0^2 (\kappa')^2 \lambda^{-4}=N_0^{2+2 b_0}$, proves that
(\ref{eq:ptierr1}) holds.
\qed \end{proof}

\subsection{Main term}

We recall the graphical representation of the main term, and the related
notations, in particular, Proposition \ref{th:main1st} and the definition of
$\gamma(m)$
in (\ref{eq:defgammam}).
\begin{proposition}\label{th:main2nd}
Suppose $t>0$ and $0<\lambda<\lambda'_0$ are given, and define
$N_0$ and $\kappa$, as in Definition \ref{th:defkappaetc}.
There is a constant
$C>0$ depending only on $f$ and $g$, and
$c>0$ depending only on $\omega$ and $\lambda'_0$ such that
\begin{align}\label{eq:main2nd}
 &  \limsup_{\Lambda\to\infty}   |\Qmain [g,f](t)-
\Qpairs [g,f](t)|
\le C \rme^t \sabs{c t}^{N_0} N_0^{N_0+4} (2 N_0)!
\sabs{\ln \lambda}^{2 N_0+2} \lambda^{1/4} \, ,
\end{align}
where $\Qpairs$ is defined by
\begin{align} 
\Qpairs [g,f](t) =
 \sum_{n=0}^{N_0-1}
 \sum_{\ell \in G_n} \sum_{S\in \pi(I_{0,2 n+1})}
\mathcal{F}_n^{\rm pairs}(S,\ell,t/\vep,\kappa)
\end{align}
with $\mathcal{F}_n^{\rm pairs}(S,\ell,t/\vep,\kappa)=0$, if the
graph defined by $S,\ell$ is not fully paired, and otherwise
it is equal to
\begin{align}\label{eq:defFnpairs}
&\mathcal{F}_n^{\rm pairs}(S,\ell,t/\vep,\kappa) =
(-\ci \lambda)^n
\sum_{\sigma\in \set{\pm 1}^{\mathcal{I}''_{n}}}
\int_{(\T^d)^{\mathcal{I}''_{n}}} \!\rmd k\,
\Delta_{n,\ell}(k,\sigma)  \FT{g}(k_{n,1})^* \FT{f}(k_{n,1})
 \nonumber \\ &  \quad \times
\1(\sigma_{n,1}=1)\1(\sigma_{0,0}=-1)
\prod_{A=\set{i,j}\in S}\!\Bigl[
 \delta(k_{0,i}+k_{0,j}) \1(\sigma_{0,i}=-\sigma_{0,j}) W(k_{0,i}) \Bigr]
 \nonumber \\ & \quad \times
\prod_{i=1}^{n} \Bigl[ \sigma_{i,\ell_{i}} \PFone(k_{i-1;\ell_i}) \Bigr]
\int_{(\R_+)^{I_{0,n}}}\!\rmd r \,  \delta\Bigl(t\lambda^{-2}-\sum_{i=0}^{n}
r_i\Bigr)
\prod_{m=0}^{n} \rme^{-\ci r_m \gamma(m)}
 \, .
\end{align}
\end{proposition}
\begin{proof}
We can apply Lemma \ref{th:basicFest} with $n'=0$ to estimate
$|\mathcal{F}_n^{\rm ampl}(S,\ell,t/\vep,\kappa)|$.
The steps of the proof are otherwise
identical to those used in the proof of Proposition
\ref{th:ptierr1}.  To avoid repetition, we skip the rest of the details here.
\qed \end{proof}

\section{Fully paired graphs}
\label{sec:fullypaired}

By the results proven in the previous section, only fully paired graphs
remain to be estimated, with the corresponding amplitudes given by
$\mathcal{G}_n^{\rm pairs}$ and
$\mathcal{F}_n^{\rm pairs}$.
In these terms, all sums over $\Lambda^*$ have already been
replaced by integrals over $\T^d$, and we have changed the covariance
function to its $\lambda\to 0$ limiting value. The related momentum graphs
differ by the number of
interactions in the minus tree: for $\mathcal{F}_n^{\rm pairs}$,
we have $n'=0$, and for $\mathcal{G}_n^{\rm pairs}$, $n'=n$.
For such relevant graphs, we have $r=0$ and $n_1=0$, and thus by
Proposition \ref{th:ivordering} for any $1\le i\le n+n'$,
\begin{align} 
n_2(i)\le n_0(i)\qand n_0(i)\ge \frac{i}{2}\, .
\end{align}
In addition, by Lemma \ref{th:ivdegrees} also $n_2=n_0=\frac{n+n'}{2}$.
Since $n+n'$ must then be even this implies that any
$\mathcal{F}_n^{\rm pairs}$ with odd $n$ is zero.
Also necessarily $\deg v_1=0$ and $\deg v_{n+n'}=2$.
Therefore, we can conclude that
either the degrees of the interaction
vertices form an alternating sequence
$(0,2,0,2,\ldots,0,2)$, or this alternating behavior ends in
two or more consecutive zeroes somewhere in the middle of the whole sequence.

Moreover, the first phase is always zero, since by Lemma
\ref{th:omOmconv} for any relevant graph
\begin{align} 
\re \gamma(0;J) =
\sum_{i=1}^{2 n+1}\!\! \sigma_{0,i} \omega(k_{0,i})
+\sum_{i=1}^{2 n'+1}\!\! \sigma'_{0,i} \omega(k'_{0,i})
\end{align}
and the pairing of momenta and parities on the initial time slice implies that
the terms cancel each other pairwise.
For simplicity, let us now drop the dependence on $J$ from the notation, \itie ,
we denote $\gamma(m)$ instead of $\gamma(m;J)$ also for $\mathcal{G}_n^{\rm
pairs}$.
We recall that the time slice $m<N$ is called long, if $\deg v_{m+1}=0$.
We now say that a long time slice $m$ is \defem{trivial}, if additionally $\re
\gamma(m)=0$.
Thus, for instance, the time slice $m=0$ is long and trivial in every relevant
pairing graph.
We also need to consider the
\defem{index of the last trivial long time slice}, by which we mean the last
trivial long time slice
in the initial sequence of trivial long slices:
the index is defined as
\begin{align}
m'_0=\max \defset{0\le m\le N}{\re \gamma(j)=0\text{ if }0\le j\le m
\text{ and the slice }j\text{ is long}} ,
\end{align}
where we have again set $N=n+n'$.  For relevant fully paired graphs, $0\le
m'_0\le N$,
and we will soon show that $m'_0=N$ if and only if the graph is leading.

Every degree two interaction vertex $v$
is at the top of the domain of influence
of its two free momenta.
We denote these by $f_1,f_2$, with
$f_1<f_2$, and let $e_3$ denote the third (integrated)
edge in $\edges_-(v)$, and $e_0$ the
unique edge in $\edges_+(v)$.
We typically also denote $k_1=k_{f_1}$, $k_2=k_{f_2}$, $k_3=k_{e_3}$,
and $k_0=k_{e_0}$. In particular, then
$k_{3}=k_{0}-k_1-k_2$ uniformly in the free momenta.

Consider the two loops
associated with $f_1$ and $f_2$ which would be created in the
spanning tree by the addition of the edge.  We travel the loops starting from
the edge $f_i$ and finishing with the edge $e_3$.
There is a {\em unique\/} vertex $v'$ where the two loops meet (this is the
first
vertex in common between the two paths). The vertex $v'$ is called the
\defem{X-vertex}, or
\defem{crossing-vertex}, associated with the \defem{double-loop of $v$.}
Clearly, $v'$ must belong to at
least three distinct {\em integrated\/} edges,
and thus it has to be a degree zero interaction vertex.
The remaining vertices (if there are any)
along the two loops are called \defem{$T$-vertices}, or
\defem{through-vertices}, of the double-loop of $v$.  These are
classified according to which of the three
possible combinations of the free momenta appear in that vertex:
if the vertex belongs to
the path from $v\to v'$ containing $f_1$,
it is called a \defem{$T_1$-vertex};
if to the path from $v\to v'$ containing $f_2$,
\defem{$T_2$-vertex};
if to the path from $v'\to v$,
\defem{$T_3$-vertex}.  The names are explained by the following
observation: {\em if $w$ is a $T_n$-vertex, $n=1,2,3$,
then exactly two of the momenta $k_e$, $e\in\edges(w)$, depend on
$k_{f_1}$ or $k_{f_2}$, and the dependence occurs via $\pm k_n$.\/}
The graph in Fig.~\ref{fig:crossing} illustrates these definitions.

\begin{figure}
\centering
  \myfigure{width=0.65\textwidth}{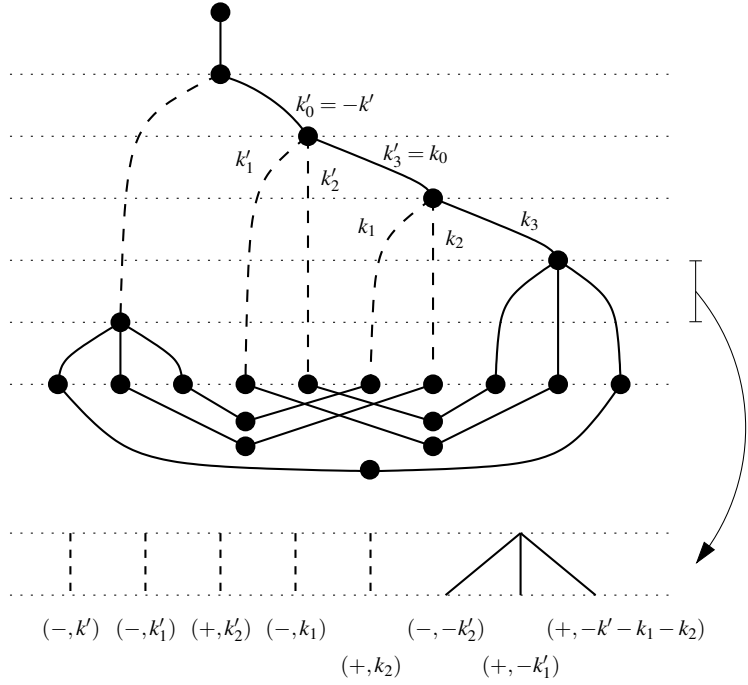}
  \medskip
\caption{Example of a crossing graph.  (Although this graph does not directly 
appear in our expansions, it could be completed to such a graph by adding a
loss motive to the top of the graph.)  The graph has two double-loops:
the double-loop of $v_3$, for which $v_1$ is an $X$-vertex and $v_2$ is a
$T_3$-vertex, and the double-loop of $v_4$, for which $v_2$ is an $X$-vertex
and $v_3$ is a $T_3$-vertex.  By inserting the appropriate parities for each
edge, one can check the pairwise cancellation of phase factors on the time
slice $0$, \itie , that $\re \gamma(0)=0$.   
Under the graph, we have denoted explicitly the parities and momenta of each
of the edges intersecting the time slice $1$. This shows that 
$\re\gamma(1)
=-\omega(k')-\omega(k_1)+\omega(k_2)+\omega(-k'-k_1-k_2)\ne \Omega_3
= -\omega(k_0)-\omega(k_1)+\omega(k_2)+\omega(k_0-k_1-k_2)$, with
$k_0=-k'-k_1'-k_2'$. 
Thus the time slice $1$ is long and independent of the double-loop of $v_4$
but it propagates a crossing with the double-loop of
$v_3$.\label{fig:crossing}} 
\end{figure}

Consider the interaction phase at an interaction vertex $v_i$,
$1\le i\le N$,
\begin{align} 
\Omega_{i} := \Omega(k_{\edges_-(v_i)},\sigma_{\edges_+(v_i)}).
\end{align}
Since now $\re \gamma(0)=0$, for every time slice $m$ of a fully paired graph,
$0\le m\le N$, we have
\begin{align} 
\re \gamma(m) = \sum_{i=m+1}^N \Omega_i = -\sum_{i=1}^m \Omega_i \, .
\end{align}
Consider an arbitrary long time slice $m$,
$0\le m\le N-1$, which thus ends in
a degree zero interaction vertex $v_{m+1}$, and
an arbitrary degree two vertex $v_{i_2}$, with the corresponding free
momenta $k_1,k_2$.
If $\re \gamma(m)$ does not depend on $k_1$ or $k_2$, the time slice
$m$ is said
to be \defem{independent of the double-loop of $v_{i_2}$.}
If $\re \gamma(m)$ depends on $k_1$ or $k_2$, but $\re \gamma(m)-\Omega_{i_2}$
does not, the time slice is said to be
\defem{nested inside the double-loop of $v_{i_2}$.}
If both $\re \gamma(m)$ and  $\re \gamma(m)-\Omega_{i_2}$
depend on $k_1$ or $k_2$, the time slice $m$ is said to
\defem{propagate a crossing with the double-loop of $v_{i_2}$.}

We are now ready to give the precise definition of how the subleading
relevant fully
paired graphs are divided into nested and crossing graphs.
We iterate trough degree two vertices, starting from the bottom of the graph,
and consider the double-loop associated with the vertex. If every long
time slice
is independent of the double-loop, we move on to the next vertex in the list
(we will soon prove that
the double-loop is then formed by iteration of leading motives).
Otherwise, there is a long time slice which depends on the double-loop.
If all such time slices are nested inside the double-loop, the double-loop is
called a \defem{nest} and the graph is called {\em nested}.  We have given an
example of a nested graph in Fig.\ \ref{fig:nested}.
Otherwise, there is a long time slice which depends on the double-loop but is
not nested inside it.  We call the topmost of these time slices
the \defem{last propagated crossing slice}, and the graph itself is then
called a {\em crossing graph}. 
An example is given in Fig.\ \ref{fig:crossing}.

\begin{figure}
 \centering
  \myfigure{width=0.6\textwidth}{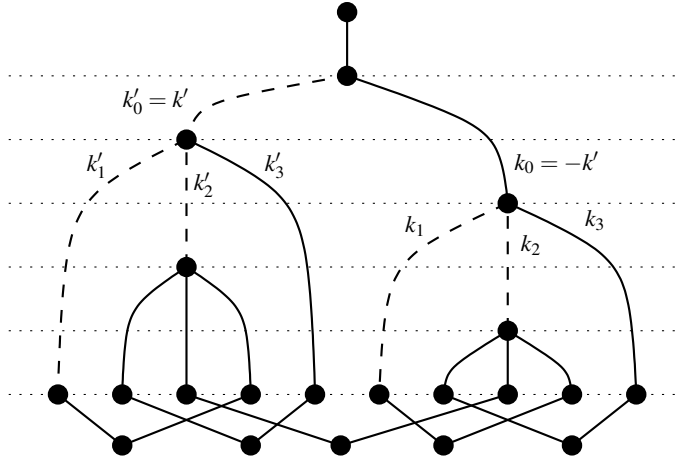}
  \medskip
\caption{Example of a nested graph.  The graph has two double-loops:
the double-loop of $v_3$, for which $v_1$ is an $X$-vertex, and the
double-loop of $v_4$, for which $v_2$ is an $X$-vertex.
Reading the appropriate parities and momenta from the graph shows that
$\re \gamma(1)=-\omega(k')-\omega(k_1)+\omega(k_2)+\omega(k_3)=\Omega_3$.
Thus the time slice $1$ is long and independent of the double-loop of $v_4$
but it is nested inside the double-loop of $v_3$.\label{fig:nested}}
\end{figure}

The following Proposition shows that these definitions yield a complete
classification of relevant fully paired graphs.
\begin{proposition}\label{th:leadingremains}
Suppose a fully paired graph has a non-zero associated amplitude,
and is not nested nor crossing.  Then  every long time slice of the graph is
trivial, and the graph is leading.  In addition, any relevant fully paired
graph,
for which all long time slices are trivial, is a leading graph.
\end{proposition}
The Proposition also provides a connection between the present definition of a
leading graph via iteration of leading motives, and the alternative earlier
definitions (\itcf \ ``Kinetic Conjecture'' on page
1078 in \cite{spohn05}, and ``Definition 5.6'' in \cite{ls09}).  In
particular, it implies that, if the edges ``cancel pairwise''
on every long time slice of a relevant graph, then the graph is obtained by
iteration of leading motives.
We also remark here that, due to the $\PFone$-factors at interaction vertices,
any graph, in which two of the three interacting momenta sum identically to
zero, is irrelevant.
Before giving a proof of the Proposition, we need two lemmas which show that
our dispersion relations are sufficiently non-degenerate.
\begin{lemma}\label{th:Oi0notconstant}
Consider an interaction vertex $v_{i_0}$ in a relevant graph.
If $f\in \fedges_e$ for some $e\in \edges(v_{i_0})$, then
$\nabla_{k_f}\Omega_{i_0}\ne 0$.
In addition, $\Omega_{i_0}$ cannot be independent of all free momenta.
\end{lemma}
\begin{proof}
Let us assume the converse, \itie , that $\Omega_{i_0}$ is a constant
in $k_f$.
Let us denote the momenta associated with the edges in $\edges(v_{i_0})$
by $k_i$, $i=0,1,2,3$, as explained above, and let analogously 
$\fedges_i=\fedges_{e_i}$. Since the graph is relevant, by Corollaries
\ref{th:zerok} and \ref{th:zeroisirr} 
each $k_i$ depends on some
free momenta, and $\fedges_i\ne \emptyset$.
By uniqueness of the loops used in the definition of
free edges, any free
momenta can appear in zero or exactly two of the four sets $\fedges_i$.
In particular, there are unique $i,j$ such that
$f\in \fedges_i\cap\fedges_j$, and let $i',j'$ denote the remaining two
indices.
Then $k_i=\pm (k_f + u)$,
$k_j=\pm (k_f + u')$, where $u$ and $u'$ are some linear
combinations of the remaining free momenta, possibly even zero.

Differentiating $\Omega_{i_0}$  with respect to $k_f$, we find that
$\nabla \omega (k_i)\pm \nabla \omega(k_i+u'-u)=0$, for all
$k_i\in \T^d$, and any $u'-u$, which is some
linear combination of the free momenta excluding $k_f$.
If $u'-u$ is not zero,
we can further differentiate with respect to a free momentum
appearing in the sum, which implies that the Hessian of $\omega$
is uniformly zero.
Since then $\omega$ is a linear map which is periodic, it
is a constant.  However, a constant dispersion relation
obviously cannot satisfy Assumption (DR\ref{it:DRdisp}).

Thus we can assume that $u'=u$, which implies $k_i=\pm k_j$.
We recall that $k_0=k_1+k_2+k_3$.  Suppose first that $k_i=k_j$.
If either $i$ or $j$ is zero, this implies $k_{i'}+k_{j'}=0$
uniformly, where $i',j'\in\set{1,2,3}$.
However, in this case by Proposition \ref{th:PFcorr}
the factor $\PFone(\pm k)$ appearing at the interaction vertex
is identically zero and the amplitude of
the graph is also identically zero.
If both $i,j$ are different from zero,
we have $k_0=k_{i'}+ 2 k_i$ which is not compatible with
the uniqueness of the representation in (\ref{eq:kesol3}), unless $k_i=0$.
(Any free edge in $\fedges_i$ appears already also in $\fedges_j$ and thus
cannot appear in $\fedges_{i'}$.)
This however is possible only if the graph is irrelevant.
If $k_i=-k_j$, and one of $i,j$ is zero, say $j=0$,
then $2 k_i+k_{i'}+k_{j'}=0$
which is not possible unless the graph is irrelevant.
Otherwise, $k_i+k_j=0$ implies $\PFone(\pm k)=0$ and the graph must again be
irrelevant.

We can thus conclude that $\nabla_{k_f}\Omega_{i_0}\ne 0$.
Since for instance $\fedges_1\ne \emptyset$, the second statement is an
obvious corollary of this.
\qed \end{proof}

\begin{lemma}\label{th:immrecstart}
Suppose that all long time slices of a relevant fully paired graph
are independent of the double-loops of
the first $M$ degree two vertices.  Then all of the corresponding
double-loops are immediate recollisions: there is a graph from which
the full graph can be obtained by iteratively adding $M$ leading motives.
\end{lemma}
\begin{proof}
We do the proof by induction in $M$.  Case $M=0$ is vacuously true.
We make the induction assumption that the statement holds for
up to $M-1\ge 0$ degree two vertices,
independently of the other properties of the graph.
We then suppose that also the double-loop of the
$M$:th degree two vertex has no long time slices dependent on it.
The proof heavily uses the iterative cluster scheme, and to facilitate
following it, we have given two applications of the scheme in
Fig.~\ref{fig:IC2}.  The first of these examples is not leading, but it is
related to the discussion in the next paragraph.  The second is a leading
graph, and provides a convenient example for the rest of the proof. 

\begin{figure}
  \centering
  \medskip
  \myfigure{width=0.95\textwidth}{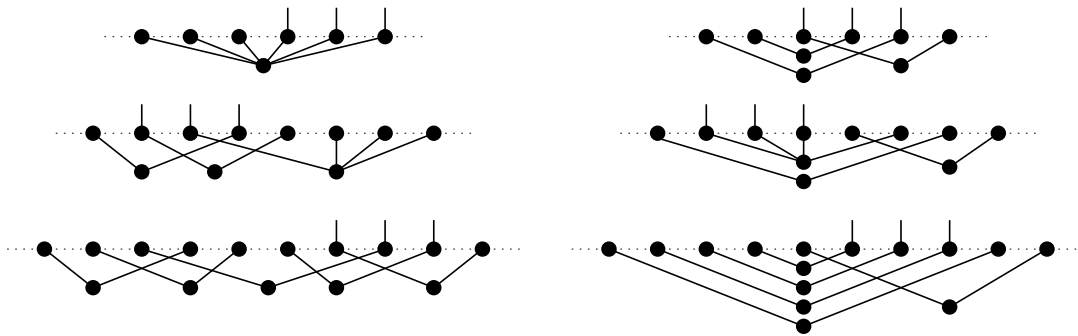}
  \medskip
  \caption{Two applications of the iterative cluster scheme, using the same
    notations as in  
  Fig.~\ref{fig:iterclusters}.  The left tower corresponds to the nested graph
  in Fig.~\ref{fig:nested}, and  
  the right one to the right leading graph in Fig.~\ref{fig:leaditer}.  We
  have left out the last two iteration steps  
  which are very similar for any fully paired graph.\label{fig:IC2}}
\end{figure}

Consider the iterative cluster scheme just before addition of the
first degree two vertex $v=v_{i_2}$, in which case 
$\deg v_i=0$ for all $i<i_2$. 
As before, let $f_1,f_2$ denote the free edges in 
$\edges_-(v)$, and $k_i$, $i=0,1,2,3$, 
the momenta associated with $v$.
The double-loop of $v$ contains also some
other interaction vertices
(any double-loop contains a crossing-vertex, and it is straightforward to
check that this must be an interaction vertex in a pairing graph). 
By construction, all of these vertices have degree zero.
Let $i$ denote the minimum of
indices such that $v_i\ne v$ belongs to the double-loop of $v$.
Suppose first that $i<i_2-1$.  Since $\deg v_{i}=0=\deg v_{i+1}$, the
time slices $i-1$ and $i$ are then long and, by assumption, independent of
$k_1,k_2$. 
This implies that $\re \gamma(i-1)-\re \gamma(i) = \Omega_i$ is also
independent of $k_1,k_2$. Since at least one of the free
edges $f_1,f_2$  appears in $\fedges_e$ for some $e\in \edges(v_i)$, by Lemma
\ref{th:Oi0notconstant}
we can conclude that $\Omega_i$ can be independent of the corresponding
momentum only in an irrelevant graph.

Thus we can assume that $i=i_2-1$, which implies that the
double-loop of $v$ contains only one interaction vertex, $v_i$, which then is
a crossing-vertex.  Since $\deg v_i =0$,
in the iterative cluster scheme
the addition of $v_i$ combines three
distinct clusters $A_1,A_2,A_3$ into one new cluster $A'$.
We denote $A'_j=A'\cap A_j$, $j=1,2,3$, and define $A'_0=\edges_+(v_i)$.
Clearly, $\set{A'_j}$ then forms a partition of the set $A'$.
Since $v$ is a degree two vertex, all edges in $\edges_-(v)$ belong to the
same iterative cluster,
which must be $A'$ since $v_i$ is along the double-loop of $v$.  Thus for each
$e\in \edges_-(v)$ there is
$j_e\in I_{0,3}$ such that $e\in A'_{j_e}$.  In addition, all indices must be
different, since otherwise $v_i$ cannot be the crossing vertex. Finally,
suppose that $|A'_j|>1$ for some $j$.  Then $j\ne 0$, $|A_j|>2$, and, since
all clusters are pairs, any path from one edge of $A_j$ to another must then
go via an interaction vertex $v'\ne v,v_i$.  However, then there is $i'<i$
such that the double-loop goes through $v_{i'}$ and this is not allowed, since
$i$ was assumed to be the smallest of such indices. 
Thus we can conclude that $A'_{j_e}=\set{e}$ for all $e\in \edges_-(v)$.

If $j_e\ne 0$ for all $e$, it follows that all of $A_j$ are pairings, and
that the addition
of $v_i$ and $v$ is equivalent to splitting of a pairing using a gain motive,
as in the second example in Fig.~\ref{fig:IC2}. 
Else $j_e=0$ for some $e$.  Then the remaining two edges
connect via a pairing to $v_i$, while the size of the third
iterated cluster remains unaffected.  This is equivalent to an addition of a
loss motive to one of the edges in the third cluster. In both cases, the
result is an immediate recollision, which thus leaves the $\gamma$-factors and 
$k$-dependence
invariant.  This implies that we can apply the induction assumption to the
graph which is obtained by cutting out the leading motive of $v$
from the original graph, \itie , by removing the time slices $i-1$ and $i$,
all edges and vertices associated with the pairings used in the leading
motive, and then repairing the graph
by either adding the pairing previously formed by a gain term or by
reconnecting the two ends previously joined by a loss motive.
We can then apply the induction assumption and conclude that the 
statement holds for arbitrary $M\ge 0$. 
\qed \end{proof}

\begin{proofof}{Proposition \ref{th:leadingremains}}
By assumption the
graph is relevant, but none of its long time slices depends on
any of the double-loops.  Then by Lemma \ref{th:immrecstart}
all double-loops correspond to immediate recollisions, and if all recollisions
are removed,
a simple loop corresponding to $n=n'=0$ is left over.
Thus the graph is leading, and as immediate recollisions preserve the phase
factor, which is initially zero, all long time slices are trivial.
Conversely, if all long slices are trivial,
they are zero independently of all free momenta.  Therefore, also then the
graph is leading.\qed
\end{proofof}

\subsection{Crossing graphs}

\begin{proposition}\label{th:crossingbnd}
There is a constant $c_0$, which
depends only on $\omega$, and a constant $C$, which depends only on
$\omega,f,g$,
such that the amplitudes of all crossing graphs
satisfy the bounds
\begin{align} 
& |\mathcal{G}_n^{\rm pairs}(S,J,\ell,\ell',s,\kappa)|
\le C \lambda^{2 \gamma}
\rme^{s\lambda^2} \sabs{c_0 s\lambda^2}^{n-1}
\sabs{\ln \lambda}^{3+c_2+2 n}\, , \\
& |\mathcal{F}_n^{\rm pairs}(S,\ell,t \lambda^{-2},\kappa)|
\le C \lambda^{2 \gamma}
\rme^{t} \sabs{c_0 t}^{n/2-1}
\sabs{\ln \lambda}^{3+c_2+n} \, .
\end{align}
\end{proposition}
\begin{proof}
Both bounds can be derived simultaneously, if we consider
a general crossing graph.  Obviously, it also suffices to derive the bound
merely for relevant graphs.
By Lemma \ref{th:immrecstart},
then there is $i_2\in I_{2,N}$ such that
$\deg v_{i_2}=2$ and every $1\le i<i_2$ has either $\deg v_i=0$ or corresponds
to
an immediate recollision.  In addition, there is a long time slice
which propagates a crossing with the double-loop of $v_{i_2}$.
Let $i_0-1$ be the largest index of such time slices, \itie , of the last
propagated crossing slice.

We denote by $k_i$, $i=0,1,2,3$, the momenta of the edges in 
$\edges(v_{i_2})$, as before.  In particular, then $k_1$ and $k_2$ are free
momenta.  Now $i_0<i_2$, and
\begin{align} 
\re \gamma(i_0-1) = \Omega_{i_2} + \Omega_{i_0} +
a_1 + \alpha_2\, ,
\end{align}
where
\begin{align} 
a_1 =  \sum_{i=i_0+1}^{i_2-1} \Omega_i
\qand
\alpha_2 =  \sum_{i>i_2} \Omega_i \, .
\end{align}
By construction of the spanning tree, $\alpha_2$ cannot depend on $k_1,k_2$,
nor on any other double-loop
of $v_i$ with $i\le i_2$.
We prove next that there is $\alpha_1$, which is also independent of all such
double-loops,
and $p\in\set{0,1}$ such that $a_1=-(1-p) \Omega_{i_2}+\alpha_1$.
This implies
\begin{align} 
\re \gamma(i_0-1) = p\Omega_{i_2} + \Omega_{i_0} +
\alpha_1 + \alpha_2\, .
\end{align}
Then the vertex $v_{i_0}$ has to be
either an $X$- or $T$-vertex for the double-loop of $v_{i_2}$.
Otherwise, $\Omega_{i_0}$ does not depend on $k_1,k_2$, which would imply
that either $\re \gamma(i_0-1)$ or
$\re \gamma(i_0-1) - \Omega_{i_2}$ is independent of $k_1,k_2$ contradicting
the assumption that $i_0-1$ propagates a crossing.

Let us first consider $i$ such that $\deg v_i=2$ and $i_0+1\le i<i_2$.  By
construction, the corresponding double-loop is determined by a leading motive, 
and thus $\Omega_{i-1}=-\Omega_i$. If $i\ge i_0+2$, the corresponding terms
thus cancel each other in the sum defining $a_1$.  On the other hand,
$i=i_0+1$ implies $\Omega_{i_0}=-\Omega_{i_0+1}$,
and thus $\re \gamma(i_0-1)=\re \gamma(i_0+1)$.  Since $i_0-1$ is the last
propagated crossing slice, this
is allowed only if $\deg v_{i_0+2}=2$, but then necessarily $i_0+2=i_2$.
However, in this case  
$\re \gamma(i_0-1) = \Omega_{i_2} + \alpha_2$ and $i_0-1$
does not propagate a crossing.
Therefore, we can conclude that if we define $I'=\defset{i\in
I_{i_0+1,i_2-1}}{\deg v_{i}=0,\text{with }i=i_2-1\text{ or } \deg v_{i+1}=0}$,
then
$a_1 = \sum_{i\in I'} \Omega_i$, which implies that $a_1$ is independent of all
double-loops of $v_i$, $i<i_2$.
In particular, if $I'=\emptyset$, then $a_1=0$ and the claim holds with $p=1$.
Otherwise, $I'\ne\emptyset$, and there is $i'=\min I'$. Then $\deg v_{i'}=0$,
the time slice $i'-1>i_0-1$ is long, and
we have $\re \gamma(i'-1)=\Omega_{i_2} + a_1 + \alpha_2$.  Since this slice
cannot propagate a crossing,
we must have that either $a_1$ or $a_1+\Omega_{i_2}$ is independent of the
double-loop of $v_{i_2}$.
In the first case, we define $\alpha_1=a_1$ and $p=1$, and in the second we
let $\alpha_1=a_1+\Omega_{i_2}$ and $p=0$.
With these definitions, all of the previous claims hold.

Set $m'=i_0-1$.
We follow the iteration scheme used in the basic $\mathcal{F}$-estimate
(Lemma \ref{th:basicFest}) with the following exceptions:
we now have $A_1=\emptyset$, and we define $A=\set{m',2 n}\cup A_2$,
\itie , we move the index $m'$ from $A_0$ to $A$.
Since $|A_2|=|A_0|=N/2$,
the integrated phase factor satisfies
\begin{align} 
& \Bigl|\int_{(\R_+)^{I_{0,N}}}\!\rmd r \,
\delta\Bigl(s-\sum_{i=0}^{N} r_i\Bigr)
 \prod_{i=0}^{N} \rme^{-\ci r_i \gamma(i)}
\Bigr|
 \nonumber \\ & \quad
\le  \frac{s^{N/2-1}}{(N/2-1)!}
\oint_{\Gamma_N} \frac{|\rmd z|}{2\pi}
\frac{\rme^{s (\im z)_+}}{|z|}
\prod_{i\in \set{m'}\cup A_2} \frac{1}{|z-\gamma(i)|}
 \, .
\end{align}
Then we follow the iteration
procedure until index $m'$ is reached.  At that point we have to deal with the
dependence of the factor $1/|z-\gamma(m')|$ on the various free momenta.
If $z$ does not belong to the top of the integration path, we can
estimate trivially $1/|z-\gamma(m')|\le 1$, and then complete the iterative
estimate as in the proof of Lemma \ref{th:basicFest}.  This yields an
improvement of the upper bound by a full factor of $\lambda^2$.

For those $z$ belonging to the top of the integration path, we have
$z=\alpha+\ci \beta$ for some $|\alpha|\le 1+2 (N+1) \norm{\omega}_\infty$
and with $\beta=\lambda^2>0$.
We next remove any dependence on $\kappa'$:
since  $\im \gamma(m)\le 0$ for all $m$,
we can estimate all the remaining resolvent factors by
\begin{align} 
& \frac{1}{|\alpha-\gamma(m)+\ci \beta|} \le
 \frac{1}{|\alpha-\re \gamma(m)+\ci \beta|}
\, .
\end{align}
As we have shown above, $\re \gamma(m')$ is independent of any free momenta
appearing before $k_1,k_2$.  These can thus be estimated as before.
Finally, we arrive at the $k_1,k_2$-integral, which is equal to
\begin{align}\label{eq:crossTint}
&  \int_{(\T^d)^2}\!\! \rmd k_1\rmd k_2\,
\frac{1}{|\alpha-\alpha_2-\Omega_{i_2}+\ci\beta|
|\alpha-\alpha_2-\alpha_1-p \Omega_{i_2} -\Omega_{i_0} +\ci\beta|}
 \, .
\end{align}
We represent both factors in terms of the oscillating integrals, using
(\ref{eq:restophases}).  Since all $\alpha$-terms above
are independent of $k_1,k_2$,
Fubini's theorem shows that the integral is bounded by
\begin{align}\label{eq:crossTint2}
& \sabs{\ln \beta}^2
\int_{\R^2}\!\! \rmd r\rmd s\, F(r;\beta) F(s;\beta)
 \Bigl| \int_{(\T^d)^2}\!\! \rmd k_1\rmd k_2\,
 \rme^{-\ci (r+p s) \Omega_{i_2}-\ci s\Omega_{i_0}} \Bigr|
 \nonumber \\ & \quad
\le 4 \sabs{\ln \beta}^2 \biggl( 1 +
\int_{\R^2}\!\! \rmd r\rmd s\, \rme^{-\beta |s|}
 \Bigl| \int_{(\T^d)^2}\!\! \rmd k_1\rmd k_2\,
 \rme^{-\ci (r+p s) \Omega_{i_2}-\ci s\Omega_{i_0}} \Bigr|
\biggr) \, .
\end{align}

Suppose $v_{i_0}$ is an $T_j$-vertex.  Then
$\Omega_{i_0} = \pm \omega(k_j+u)\pm \omega(k_j+u') + \alpha'$
for some choice of the signs and for
$\alpha'$ and $u,u'$ which are independent of $k_1,k_2$.
The $j$-part of the double-loop goes through the vertex via two edges.
If both edges $e,e'$ are in $\edges_-(v_{i_0})$, we have
$k_e=\sigma(k_j+u)$ and $k_{e'}=-\sigma (k_j+u')$ for some 
$\sigma\in \set{\pm 1}$, which implies
$u'-u=-\sigma(k_e+k_{e'})$.  Otherwise, the loop uses 
$\tilde{e} \in \edges_-$
and $\tilde{e}'\in \edges_+$, and then
$k_{\tilde{e}}=\sigma(k_j+u)$ and $k_{\tilde{e}'}=\sigma (k_j+u')$, implying
$u'-u=\sigma(k_{\tilde{e}'}-k_{\tilde{e}})=\sigma(k_{e}+k_{e'})$ 
where $e,e'$ are the remaining two edges in $\edges_-$.
Thus by Lemma \ref{th:nokkdiff}, for any free momenta of a
degree two vertex, $u'-u$ is either independent of the momenta, or depends
on it by
``$\pm k_{j'}$'' for some $j'\in\set{1,2,3}$.  In addition, if $u'-u$ is
independent of all free momenta, then Corollary \ref{th:kkpl} implies
$k_e+k_{e'}=0$, and the corresponding graph is thus irrelevant.

We change variable $r$ to $t=r+p s$, and estimate
\begin{align} 
&  \Bigl| \int_{(\T^d)^2}\!\! \rmd k_1\rmd k_2\,
 \rme^{-\ci t \Omega_{i_2}-\ci s\Omega_{i_0}} \Bigr|
\nonumber \\ & \quad
= \Bigl| \int_{(\T^d)^2}\!\! \rmd k_1\rmd k_2\,
 \rme^{-\ci t (\pm \omega_1\pm \omega_2\pm \omega_3)
   -\ci s (\pm \omega(k_j+u)\pm \omega(k_j+u'))} \Bigr|
\nonumber \\ & \quad
\le \norm{p_{t}}_3^2
\norm{K(\pm t,\pm s,\pm s,u,u')}_3\, ,
\end{align}
where we have used a convolution estimate similar to (\ref{eq:convsplit}).
It is obvious from the definitions that not only
$\norm{p_{-t}}_3=\norm{p_{t}}_3$, but also the norm of $K$ remains invariant
under a swap of the signs of its time arguments.
Thus we can use this to change the first argument of $K$ to $t$.
The resulting integral over $t,r$
is of a form given in (\ref{eq:crossingest1c}).  Thus by
Assumption (DR\ref{it:DRcrossing})
we find that (\ref{eq:crossTint2}) is bounded by
\begin{align} 
 4 \sabs{\ln \beta}^2 \beta^{\gamma-1} ( 1 +  \Fbcr(u'-u;\beta) ) \, .
\end{align}
As mentioned above, for a relevant graph, $u'-u$ must depend on some free
momenta.
We iterate the basic estimates, until the first such momenta appear.
Since the dependence of $u'-u$ is of the form $\pm k_{j'}$, we can then apply
the
second part of Assumption (DR\ref{it:DRcrossing}).
If $u'-u$ depends only on free momenta of the top fusion vertex, we use the
first estimate,
otherwise we use the second estimate.  The remainder of the integrals can be
iterated as in the basic estimate.  Comparing the resulting bound to the
basic estimate shows that we have gained an improvement by a factor of
\begin{align}\label{eq:crossinggain}
C \frac{1}{\sabs{s\beta}}\sabs{\ln \beta}^{c_2+1} \beta^\gamma\, .
\end{align}
This yields the bounds given in the Proposition, since for a
$\mathcal{G}_n$-graph we have $N=2 n$ and for
$\mathcal{F}_n$-graph $N=n$.

We still need to consider the case where
$v_{i_0}$ is an $X$-vertex.  Then we have
$\Omega_{i_0} = \sum_{i=1}^3 (\pm \omega(k_i+u_i)) + \alpha'$
for some choice of the signs and for
$\alpha'$ and $u_i$, $i=1,2,3$, which are independent of $k_1,k_2$.
Suppose $u_i=0$ for all $i$.
Then also the fourth momentum is equal to
$\pm k_0$.  We can use the iterative cluster scheme as in the
proof of Lemma \ref{th:immrecstart} to show that then
$\re \gamma(i_0-1)$ would be independent of $k_1,k_2$, which is against the
construction.
For this, consider the addition of the degree zero vertex $v_{i_0}$ and let
$A_i$, $i=1,2,3$, $A'$, and its partition $A'_i$,  $i=0,1,2,3$, be defined as
in  Lemma \ref{th:immrecstart}.
Since $v_{i_0}$ is not part of an immediate recollision,
the next added vertex is either $v_{i_2}$ or has degree zero.  In the latter
case, there can be also more vertices before $i_2$ is added.  Of these, we can
ignore all immediate recollisions, since they leave the momenta and phase
factors invariant.  Thus we only need to consider the iterative cluster scheme
when a number of degree zero vertices are added
before $i_2$.  Any such addition either leaves $A'$ invariant, or increases
the number of edges in one of $A'_j$
by at least two.   However, the argument used in  Lemma \ref{th:immrecstart}
implies that at the moment when $v_{i_2}$
is added, all of the sets $A'_i$ which intersect $\edges_-(v_{i_2})$ must be
singlets, as else one of $u_i$ is not zero,
or $v_{i_0}$ is not an $X$-vertex.  Then none of $\Omega_{i}$, $i\in I'$, can
depend on $k_1,k_2$,
and thus we have $p=1$ above.  On the other hand,
$v_{i_0}$ is effectively a delayed recollision, and an explicit computation
shows that $\Omega_{i_0}=-\Omega_{i_2}$,
implying that $\re \gamma(i_0-1)$ is independent of $k_1,k_2$.

Therefore,
there is $j'=1,2,3$ such that $u_{j'}$ depends on some free momenta.
We change variable $r$ to $t=r+p s$, and estimate
\begin{align} 
& \Bigl| \int_{(\T^d)^2}\!\! \rmd k_1\rmd k_2\,
 \rme^{-\ci t \Omega_{i_2}-\ci s\Omega_{i_0}} \Bigr|
= \Bigl| \int_{(\T^d)^2}\!\! \rmd k_1\rmd k_2\,
\prod_{i=1}^3
 \rme^{-\ci (\pm t \omega(k_i)\pm s \omega(k_i+u_i))} \Bigr|
\nonumber \\ & \quad
\le \prod_{i=1}^3
\norm{K(t,\pm s,0,u_i,0)}_3\, ,
\end{align}
where we have again used the invariance of $\norm{K}_3$ under reversal
of its time arguments.
Employing the assumption (DR\ref{it:DRcrossing})
we thus find that (\ref{eq:crossTint2}) is then bounded by
\begin{align} 
 4 \sabs{\ln \beta}^2 \beta^{\gamma-1} ( 1 +  \Fbcr(u_{j'};\beta) ) \, .
\end{align}
By construction, $u_{j'}$ depends on some free momenta,
and by Lemma \ref{th:nokkdiff} the dependence is of the earlier encountered
form.
Thus we can then
conclude the rest of the estimate following the steps used for the
$T$-vertex.  This results in an improvement by a factor given in
(\ref{eq:crossinggain})
compared to the basic estimate, and concludes the proof the Proposition.
\qed \end{proof}

\subsection{Leading and nested graphs}

For the leading and nested graphs we cannot take the absolute value of too
many phase factors.  In addition, the contribution from the immediate
recollisions needs to be estimated more carefully.

Let us thus consider a {\em relevant\/} graph which is either nested or
leading.  The momentum cut-offs have now fulfilled their purpose,
and need to be removed.  We use an iterative scheme to expand
$\PFone = 1- \PFzero$ one by one, going through the interaction vertices
$i'$ from the bottom to the top.  At each step, we obtain two terms,
one of which corresponds to replacing $\PFone \to 1$ in the iterated
vertex.  This term will be continued for the next iteration step.
In the other term we take absolute values inside the $k$-integrals and
estimate the phase factor using the iteration scheme of the basic estimate.
We can then estimate the $\PFzero$-factor
using Proposition \ref{th:PFcorr}:
$\PFzero (\pm k)\le \sum_{e_i<e_j}
\1\!\bigl(d(k_{e_i}+k_{e_j},\Msing)<\lambda^b\bigr)$
where the sum runs over pairs of edges in $\edges_-(v_{i'})$.
Since the graph is relevant, whatever pair we choose, the sum depends
on some free momenta, none of which ends before $v_{i'}$.

In the resulting integral, we can use the iteration step of the basic
estimate, until we reach the first double-loop on which the extra
characteristic function depends.
If there is no such double-loop, the characteristic function depends only on
the free momenta at the top fusion vertex, and we get then an additional
factor $c\lambda^{b(d-1)}$ using Lemma \ref{th:Foneprop}.
Otherwise, at that iteration step
there is exactly one resolvent factor left over which depends on the
double-loop of the interaction vertex.
We estimate this factor trivially
and use Lemma \ref{th:Foneprop} to gain a factor $c\lambda^{b(d-1)}$ from the
remaining double-loop integral.
In this case, we thus find a bound
$c\lambda^{b(d-1)-2}\le c\lambda^{1/4}\le c \lambda^{\gamma'}$.
After these steps, the characteristic function has been removed, and we can
continue the iteration scheme
exactly as in the basic estimate.  In any case, we then find that
the term containing the additional characteristic function has an upper bound
\begin{align}\label{eq:somebound}
C \lambda^{\gamma'}
\rme^{s\lambda^2} \sabs{c s\lambda^2}^{N/2}
\sabs{\ln \lambda}^{1+N/2} \, .
\end{align}
Since after $N$ iteration steps we have exchanged all
$\PFone$ to $1$, the difference between this and the original integral
is bounded by $N$ times (\ref{eq:somebound}).

We then consider the term which does not contain any $\PFone$.
Before further estimates of the time-integrals, we
integrate out immediate recollisions at the bottom of the graph.
Since an ``internal'' momentum of one recollision can be the ``external''
momentum of earlier
recollisions, also this step needs to be performed iteratively from the bottom
to the
top.  To study the effect of one immediate recollision, let us consider
$G_{s,\tau}: L^2(\T^d)^4 \to L^2(\T^d)$ defined
for $s\in \R$, $\tau\in \set{\pm 1}^{I_{0,3}}$, by
\begin{align} 
& G_{s,\tau}[f_0,f_1,f_2,f_3](k_0)
\nonumber \\ & \quad
=
\int_{(\T^d)^3}\!\! \rmd k_1 \rmd k_2\rmd k_3\,
 \delta(k_0-k_1-k_2-k_3) \prod_{i=0}^3
\left( \rme^{-\ci s \tau_i \omega(k_i)} f_i(k_i) \right) ,\quad  k_0\in \T^d\, .
\end{align}
Obviously, then
$|G_{s,\tau}(k_0)| \le |f_0(k_0)| \prod_{i=1}^3 \norm{f_i}_2$,
and thus $G_{s,\tau}\in L^2(\T^d)$,
$\norm{G_{s,\tau}}_2 \le \prod_{i=0}^3\norm{f_i}_2$.
Let us also denote the
free evolution semigroup on $\ell_2$ by $U_t$, \itie , let
\begin{align} 
  (U_t g)(x) := \sum_{y\in \Z^d} p_t(x-y) g(y)\, , \quad g\in
  \ell_2(\Z^d)\, ,
\end{align}
where
$p_t(x) :=  \int_{\T^d}\!\rmd k \, \rme^{\ci 2 \pi x \cdot k -\ci t \omega(k)}$
is defined as before.
The convolution structure can again be employed to improve the above
bound, at the price of introducing stronger norms.
With an absolutely convergent sum and denoting the inverse Fourier-transform
of $f_i$ by $\IFT{f}_i$,
\begin{align}\label{eq:Gsest1}
 |G_{s,\tau}(k)|\le |f_0(k)|
\prod_{i=1}^3 \norm{U_{\tau_i s} \IFT{f}_i}_3 \, .
\end{align}

Let $v_{i_2}$ be the first degree two vertex corresponding to an immediate
recollision.
Then $\gamma(i_2-1)= \Omega_{i_2} + \zeta_{i_2}$
where $\zeta_{i_2}=\sum_{i>i_2} \Omega_i+\ci \im \gamma(i_2-1)$ is
independent of the free momenta of $v_{i_2}$.
In addition, $\zeta_{i_2}$ is also
independent of free momenta of any other immediate
recollisions.  Thus for the first step where immediate recollisions are
integrated out, we can take the corresponding exponential factors out of these
integrals.
If the recollision corresponds to a gain motive,
adding it to a pairing
corresponds to changing a factor $W(k_0)$ to
\begin{align} 
-G_{s,\tau(\sigma)}[1,W,W,W](k_0)  \rme^{-\ci s \zeta_{i_2}}\, ,
\end{align}
where $s=s_{i_2-1}$
is the time-variable of the ``recollision'' time slice
and we have defined $\tau(\sigma)=(-\sigma,-1,\sigma,1)$, $\sigma$
being the parity of the part
where the \defem{higher} of the two vertices is attached.
Similarly, adding a loss motive to a line with parity $\sigma$
and momentum $k_0$ changes a factor $W(k_0)$ to
\begin{align} 
\sigma \tau_{1,j} G_{s,\tau(\sigma)}[W,W_{j,1},W_{j,2},W_{j,3}](k_0)
 \rme^{-\ci s \zeta_{i_2}}\, ,
\end{align}
where $j\in \set{1,2,3}$ and
$W_{j,j}=1$, $W_{j,i}=W$ if $i\ne j$.  The addition of
the remaining immediate recollisions corresponds to
similar modification of the integrand.  However,
an input function can then also be one of the previously generated
$G_{s,\tau}$-factors in addition to the initial factors $W$.

We need to control the time-integrability of these iterated terms.
We use (\ref{eq:Gsest1}), which requires controlling the $\ell_3$-norm
of $G_{s,\tau}$.  For this, we note that for any
$h\in \ell_1(\Z^d)$ and $t\in \R$, $x\in \Z^d$,
\begin{align} 
\sum_{x\in \Z^d} |(U_t h)(x)|^3 \le
\sum_{x\in \Z^d} \sum_{y\in (\Z^d)^3} \prod_{i=1}^3 (|p_t(x-y_i)|\, |h(y_i)|)
\le \norm{p_t}_3^3 \norm{h}_1^3 \, ,
\end{align}
and thus $\norm{U_t h}_3\le \norm{p_t}_3 \norm{h}_1$.
Then Assumption (DR\ref{it:DRdisp}) provides decay in $t$.
However, the estimate is useful only if $\ell_1$-norm of $h$
remains bounded, and this requires carefully separating the free evolution
from the initial states; we note that even if
$\norm{h}_1<\infty$, typically
$\norm{U_t h}_1=\order{t^{p}}$ with $p \ge  d/2$.

Consider thus one of the factors obtained from the iteration of the leading
motive integrals and let $f\in \ell_2$ denote its inverse Fourier-transform.
Since $G_{s,\tau}$ is linear in all of its arguments, we
can neglect the sign- and $\zeta$-factors in the estimation of the
$\ell_3$-norm. However, we have to iteratively expand the first argument
until either $1$ (gain motive)
or $W$ (the initial pairing for a sequence of loss motives) is reached.
Let $M\ge 1$ denote the number of iterations needed for this.
Since both $1$ and $W$  have an inverse Fourier-transform in $\ell_1$,
we conclude that the factor is then of the form
\begin{align} 
 \hat{f}(k) := \hat{h}_0(k)
\prod_{m=1}^M \rme^{\ci \sigma_m s_m \omega(k_0)}
\prod_{m=1}^M  F_{m}(k) \, ,
\end{align}
where
$\sigma_m\in\set{\pm 1}$, $s_m\in \R$, are the appropriate parity 
and time variables, $\FT{h}_0\in\set{1,W}$, and
\begin{align} 
& F_{m}(k_0) =
\int_{(\T^d)^3}\!\! \rmd k_1 \rmd k_2\rmd k_3\,
 \delta(k_0-k_1-k_2-k_3) \prod_{i=1}^3
\left( \rme^{-\ci s_m \tau_{m,i} \omega(k_i)} \FT{f}_{m,i}(k_i) \right)
\, ,
\end{align}
where $\tau_{m,i}=\tau(\sigma_m)_i$ and
all of the functions
$\FT{f}_{m,i}$ are obtained from earlier iterations, and thus are
one of $1$, $W$, or $G_{s,\tau}$.
In any case, $h_0\in \ell_1(\Z^d)$.

Now for any $t\in \R$,
\begin{align} 
(U_t f)\hat{\;}(k) = (U_{t-\sum_{m=1}^M \sigma_m s_m}H)\hat{\;}(k),
\quad \text{where }
\FT{H}(k) = \hat{h}_0(k)
\prod_{m=1}^M  F_{m}(k) \, .
\end{align}
As $F_{m}(k_0)
= \sum_{x\in \Z^d} \rme^{-\ci 2 \pi x\cdot k_0}
\prod_{i=1}^3 (U_{\tau_{m,i} s_m } f_{m,i})(x)$,
we have
\begin{align} 
H(y) = \sum_{x\in (\Z^d)^M} h_0\Bigl(y-\sum_{m=1}^M x_m\Bigr)
\prod_{i=1}^3 \prod_{m=1}^M (U_{\tau_{m,i} s_m } f_{m,i})(x_m)\, .
\end{align}
Therefore,
$\norm{H}_1\le \norm{h_0}_1 \prod_{i=1}^3
\prod_{m=1}^M \norm{U_{\tau_{m,i} s_m } f_{m,i}}_3$.  We
conclude that
\begin{align}\label{eq:Giterest}
\norm{U_t f}_3
\le c_1  \norm{p_{t-\sum_{m=1}^M \sigma_m s_m}}_3
\prod_{m=1}^M \prod_{i=1}^3
\norm{U_{\tau_{m,i} s_m } f_{m,i}}_3
\, ,
\end{align}
where $c_1=\max(1,\norm{\IFT{W}}_1)<\infty$.

\begin{proposition}\label{th:leadingest}
There are constants $c,c_0>0$, which depend only on $\omega$,
and a constant $C$, which depends only on $\omega,f,g$,
such that for any leading graph
\begin{align}
& |\mathcal{G}_n^{\rm pairs}(S,J,\ell,\ell',s,\kappa)|
\le C \lambda^{\gamma'}
\rme^{s\lambda^2} \sabs{c s\lambda^2}^{n}
\sabs{\ln \lambda}^{2+n} +
\frac{C}{n!} (c_0 \lambda^2 s)^n \, ,
\label{eq:Gnleadb}\\
& \left|\left.\mathcal{F}_n^{\rm pairs}(S,\ell,t \lambda^{-2},\kappa)
  \right|_{\PFone \to 1} \right|
\le
\frac{C}{(n/2)!} (c_0 t)^{n/2} \, , \label{eq:Fnleadb1}\\
& \left|\mathcal{F}_n^{\rm pairs}(S,\ell,t \lambda^{-2},\kappa)
-\left.\mathcal{F}_n^{\rm pairs}(S,\ell,t \lambda^{-2},\kappa)
  \right|_{\PFone \to 1}\right|
\le C \lambda^{\gamma'}
\rme^{t} \sabs{c t}^{n/2}\sabs{\ln \lambda}^{2+n/2}
\label{eq:Fnleadb2} \, .
\end{align}
\end{proposition}
\begin{proof}
Consider a leading graph with $N$ interaction vertices.  Then
$N$ is even.  The first term in (\ref{eq:Gnleadb}) and the bound in
(\ref{eq:Fnleadb2})
arise from exchanging all
$\PFone$ factors to $1$ and both follow from applying (\ref{eq:somebound}).
In the remaining term, we leave the time-integrals unmodified,
and perform first all $k$-integrals apart from the top fusion integral
on which the original $\FT{f}$- and $\FT{g}$-factors depend.

A leading graph consists of a sequence of $N/2$ leading motives.
Since the leading motives preserve the phase, we thus
have $\re \gamma(i)=0$ for all even $i$.
In addition, for all odd $i$ we have $\Omega_{i} = -\Omega_{i+1}$.
Therefore, in this case the total phase is
\begin{align}\label{eq:leadingphases}
\sum_{i=0}^N r_i \re \gamma(i) =
\sum_{j=1}^N \Omega_j \sum_{i=0}^{j-1} r_i
= \sum_{m=1}^{N/2} \Omega_{2m} r_{2 m-1} \, .
\end{align}
As explained above, performing the immediate recollision
$k$-integrals results in an iterative application
of $G_{r_i,\tau}$.  We estimate the absolute value of the amplitude
by taking the absolute value inside the time-integrals.  For the outmost
(\itie , last) application of $G_{r_i,\tau}$ corresponding to $m=N/2$
we use (\ref{eq:Gsest1}) and in the resulting bound
we can iterate  estimates (\ref{eq:Gsest1}) and (\ref{eq:Giterest})
further until only $\ell_1$-norms of $\tilde{W}$ remain.
This shows that for each $m=1,2,\ldots,N/2$ there are three subsets
$B_{m,i}$, $i=1,2,3$,
of $I_{1,m-1}$ such that the $k$-integrated phase factor has a
bound
\begin{align}\label{eq:leadingphase}
c_1^{3 N/2+1} \prod_{m=1}^{N/2} \prod_{i=1}^3
\norm{p_{\pm r_{2 m -1}-\sum_{j\in B_{m,i}} (\pm r_{2 j-1})}}_3 \, ,
\end{align}
for some choice of signs.
By H\"{o}lder's inequality and assumption (DR\ref{it:DRdisp}) there is a
constant $c_0$ such that if this bound is integrated over
all of $r_j$, $j$ odd, the result is bounded by
$c_0^{N/2}$.  (The integration over $r_{N-1}$ is performed first, and
the rest are iterated until $r_{1}$ is reached.)
We can apply an estimate similar to that used in (\ref{eq:intrtrick})
to separate the odd and even integrations and obtain an additional factor
$s^{N/2}/(N/2)!$ from the even integrations.  Collecting all the estimates
together yields the bounds stated in the Proposition.
\qed \end{proof}

\begin{proposition}\label{th:nestedbnd}
There is a constant $c_0$, which
depends only on $\omega$, and a constant $C$, which depends only on
$\omega,f,g$,
such that the amplitudes of all nested graphs
satisfy the bounds
\begin{align} 
& |\mathcal{G}_n^{\rm pairs}(S,J,\ell,\ell',s,\kappa)|
\le C \lambda^{\gamma'}
\rme^{s\lambda^2} \sabs{c_0 s\lambda^2}^{n}
\sabs{\ln \lambda}^{2+n}\, , \\
& |\mathcal{F}_n^{\rm pairs}(S,\ell,t \lambda^{-2},\kappa)|
\le C \lambda^{\gamma'}
\rme^{t} \sabs{c_0 t}^{n/2}
\sabs{\ln \lambda}^{2+n/2} \, .
\end{align}
\end{proposition}
\begin{proof}
Consider a relevant nested graph.
Let $i_2$ denote the index of the
first degree two interaction vertex $v=v_{i_2}$
which is not an immediate recollision.
By assumption,
every long time slice which depends on the double-loop of $v$
is nested inside the double-loop, and there is at least one such
time slice.
Let $N_2$ collect the indices of these time slices.
In addition, applying Lemma \ref{th:immrecstart}, we can conclude
that every double-loop before $i_2$
corresponds to an immediate recollision.

As explained in the beginning of the
section, it suffices to consider the case where every 
$\PFone$ at an interaction vertex $v_i$ with $i\le i_2$
has been replaced by one.  A bound for such a term then needs to be
summed with (\ref{eq:somebound}) 
times $N$ to get a bound for the original amplitude.

Let $j_0=\min N_2<i_2-1$.
Since $\re \gamma(0)=0$, we have
$j_0>0$ and thus there is a time slice $j_0-1$.
If it is short, then $j_0-1>0$ and
it belongs to an immediate recollision. This however leads to contradiction,
since immediate
recollisions preserve the phase factor, and thus
$\re \gamma(j_0)=\re \gamma(j_0-2)$ implying that the slice $j_0-2<j_0$ is
also nested inside
the double-loop of $v$.
Thus the slice $j_0-1$ must be long, but not nested inside the double-loop.
This implies that $\re \gamma(j_0-1)$ is independent, and thus
$\re \gamma(j_0)-\re \gamma(j_0-1)= -\Omega_{j_0}$
depends on the double loop of $v$.  However, as $j_0$ is nested inside the
double-loop,
we then also have that
$\re \gamma(j_0)-\re \gamma(j_0-1)-\Omega_{i_2}$ is independent,
\itie , that $\Omega_{j_0}+\Omega_{i_2}$ is independent of the
double-loop.  The vertex $v_{j_0}$ belongs to the double-loop of $v$.
It cannot be a
$T_j$-vertex, as then
by differentiation of $\Omega_{j_0}+\Omega_{i_2}$ in a direction orthogonal to
$k_j$ we should have a constant dispersion relation $\omega$.  Therefore,
$v_{j_0}$ is an $X$-vertex.

Consider the contribution of the first $i_2$ time slices to the total phase,
that is,  $\rme^{-\ci \sum_{i=0}^{i_2-1} r_i \gamma(i)}$.
Expanding $\gamma(i)$ we have here
\begin{align} 
\sum_{i=0}^{i_2-1} r_i \gamma(i)
= \zeta(r) + \sum_{i=0}^{i_2-1} r_i \sum_{j=i+1}^{i_2} \Omega_{j} \, ,
\end{align}
where $\zeta(r)$ does not depend on any of the free momenta appearing before
$i_2$ and $\im \zeta(r)\le 0$.
Denote $B_j=\defset{1\le i < i_2}{\deg v_i=j}$, $j=0,2$.
Since every $v_i$ with $i\in B_2$ ends an immediate recollision,
then $i-1\in B_0$ and $\Omega_{i-1}=-\Omega_i$.
We denote the remaining indices by
$B'_0 = \defset{j\in B_0}{j+1 \not \in B_2}$.
Since the time slices $j_0-1$ and $j_0$ are long, we have
$j_0\in B'_0$.  Therefore,
\begin{align} 
& \sum_{i=0}^{i_2-1} r_i \sum_{j=i+1}^{i_2} \Omega_{j}
= \sum_{j\in B_2} \Omega_{j} r_{j-1}
  + \Omega_{i_2} \sum_{i=j_0}^{i_2-1} r_i
 + (\Omega_{j_0}+\Omega_{i_2}) \sum_{i=0}^{j_0-1} r_i
  + \sum_{j\in B'_0\setminus\set{j_0}} \Omega_{j} \sum_{i=0}^{j-1} r_i \, .
\end{align}

Since $v_{j_0}$ is an $X$-vertex, by an argument similar to what was used for
crossing graphs,
we see that there cannot be any index
$j\in B'_0\setminus\set{j_0}$ such that $\Omega_j$ depends on the double-loop
of $v$ or on any of the immediate recollision momenta.
Therefore, there is $\tilde{\zeta}(r)$, which does not depend on any
double-loop before $i_2+1$ and has $\im \tilde{\zeta}(r)\le 0$,
such that
\begin{align}\label{eq:nestedphasef}
\exp\Bigl( -\ci \sum_{i=0}^{i_2-1} r_i \gamma(i)\Bigr)
=\rme^{-\ci \tilde{\zeta}(r)}
\rme^{-\ci \Omega_{i_2} \sum_{i=j_0}^{i_2-1} r_i  }
\prod_{j\in B_2} \rme^{-\ci r_{j-1} \Omega_{j} } \, .
\end{align}

We then apply the basic iterative estimate with slight modifications:
we integrate all free momenta of vertices $i\le i_2$ before
taking the absolute value of the phase-integral.  We recall the definition of
the sets $A_0$ and $A_2$ whose sizes for a fully paired graph are
equal to $N/2$.  Here we do not include all elements of $A_2$ in $A$, but use
$A=\set{N}\cup \defset{i_2\le i<N}{\deg v_{i+1}=2}$
in Theorem \ref{th:resolvents}. Then
$A'=\set{*}\cup I_{0,i_2-1} \cup A_0$ and we find
\begin{align} 
& \int_{(\R_+)^{I_{0,N}}}\!\rmd r \,
\delta\Bigl(s-\sum_{i=0}^{N} r_i\Bigr)
 \prod_{i=0}^{N} \rme^{-\ci r_i \gamma({i})}
= -\oint_{\Gamma_N} \frac{\rmd z}{2\pi}
\frac{\ci}{z}
\prod_{i\in A_2; i\ge i_2} \frac{\ci}{z-\gamma(i)}
 \nonumber \\ & \qquad  \times
\int_{(\R_+)^{A'}}\!\rmd r \,  \delta\Bigl(s-\sum_{i\in A'} r_i\Bigr)
\rme^{-\ci r_* z}
 \prod_{i=i_2; i\in A_0}^{N-1}\rme^{-\ci r_i\gamma(i)}
\exp\Bigl( -\ci \sum_{i=0}^{i_2-1} r_i \gamma(i)\Bigr)
 \, .
\end{align}

We apply (\ref{eq:nestedphasef}) and integrate over all double-loop
momenta of the immediate recollisions $i\in B_2$
and of the nesting vertex $v_{i_2}$.  We note that also the last integral
yields a $G_{s,\tau}$-factor but with a sum of multiple time-variables in its
argument.  (In fact, the nested integral corresponds to a leading
motive in the iterative cluster scheme, however, with a {\em delayed}
recollision.)
Let $G_2(k,r)$ denote the resulting factor.
After this, we take absolute value of
the amplitude inside all of the remaining integrals, yielding a bound
\begin{align}\label{eq:pkbound}
& \oint_{\Gamma_N} \frac{|\rmd z|}{2\pi}
\frac{1}{|z|} \rme^{s \lambda^2}
\prod_{i\in A_2; i\ge i_2} \frac{1}{|z-\gamma(i)|}
\int_{(\R_+)^{A'}}\!\rmd r \,  \delta\Bigl(s-\sum_{i\in A'} r_i\Bigr)
|G_2(k,r)|
 \, .
\end{align}
Now if we neglect all extra decay arising from the possible
$\im \gamma(i)\le 0$, we have
$|G_2|\le |\tilde{G}_2|$ where the upper bound depends only on
the time-integrals corresponding to the index set
$B=I_{j_0,i_2-1}\cup \defset{i}{i+1\in B_2}$.  Therefore,
\begin{align} 
&\int_{(\R_+)^{A'}}\!\rmd r \,  \delta\Bigl(s-\sum_{i\in A'} r_i\Bigr)
|G_2(k,r)|
 \nonumber \\ & \quad
\le
\int_{(\R_+)^{B}}\!\rmd r \,
|\tilde{G}_2(k,r)| \1\Bigl(\sum_{i\in B} r_i\le s\Bigr)
\int_{(\R_+)^{A'\setminus B}}\!\rmd r \,
\delta\Bigl(s-\sum_{i\in A'} r_i\Bigr)
 \nonumber \\ & \quad
\le \frac{s^{\tilde{n}}}{\tilde{n}!}
C^{|B_2|+1}
\int_{(\R_+)^{B}}\!\rmd r \,
 \1\Bigl(\sum_{i\in B} r_i\le s\Bigr)
\prod_{i=1}^3
\norm{p_{\sum_{i=j_0}^{i_2-1} r_i+\sum_{j\in B; j<j_0} a_{j_0,i,j} r_{j}}}_3
 \nonumber \\ & \qquad \times
\prod_{m\in B; m<j_0} \Bigl(
 \1\Bigl(\sum_{i\in B; i\le m} r_i\le s\Bigr)
\prod_{i=1}^3
\norm{p_{r_{m}+\sum_{j\in B; j<m} a_{m,i,j} r_{j}}}_3 \Bigr)\, ,
\end{align}
where $\tilde{n}=|A'\setminus B|-1$, $a_{m,i,j}\in \set{-1,0,1}$,
and we have used
(\ref{eq:leadingphase}) to estimate $|\tilde{G}_2|$.

We estimate the time-integrals iteratively,
starting from the last index.  At each step we first use H\"{o}lder's
inequality to simplify the argument of the integral into a single
third power and then use Assumption (DR\ref{it:DRdisp}).
Then in the first iteration step we need to estimate
an $M$-dimensional integral,
$M=i_2-j_0\ge 2$,  of the type
\begin{align} 
& \int_{(\R_+)^M} \! \rmd t\, \1\Bigl(\sum_{i=1}^{M} t_i\le s\Bigr)
\Bigl\langle\sum_{i=1}^M t_i +
\alpha\Bigr\rangle^{-1-\delta}
\le \int_0^s\!\rmd T \, T^{M-1} \sabs{T+\alpha}^{-1-\delta}
\nonumber \\ & \quad
\le s^{M-2}
\int_0^s\!\rmd T \, (T+\alpha-\alpha) \sabs{T+\alpha}^{-1-\delta}
\nonumber \\ & \quad
\le s^{M-2} \Bigl(
\int_0^s\!\rmd T \, \sabs{T+\alpha}^{-\delta}
+|\alpha|
\int_0^s\!\rmd T \, \sabs{T+\alpha}^{-1-\delta}
\Bigr) \, ,
\end{align}
where $|\alpha|\le \sum_{j\in B; j<j_0} r_j\le s$.
Since $2\delta\ge \gamma'$ and $\gamma'<1$, we have
\begin{align} 
\int_0^s\!\rmd T \, \sabs{T+\alpha}^{-\delta}
\le s
\int_{\alpha/s}^{1+\alpha/s}\!\rmd x \, \sabs{s x}^{-\gamma'/2}
\le s^{1-\gamma'/2} \int_{-1}^{2}\!\rmd x |x|^{-\gamma'/2}\, ,
\end{align}
where the last integral is convergent.  Since also
\begin{align} 
\int_0^s\!\rmd T \, \sabs{T+\alpha}^{-1-\delta}\le
\int_{-\infty}^\infty\!\rmd y \, \sabs{y}^{-1-\delta}<\infty \, ,
\end{align}
we can
conclude that there is a constant $C$, depending only on $\delta$, such that
\begin{align} 
\int_0^s\!\rmd T \, \sabs{T+\alpha}^{-\delta}
+|\alpha|
\int_0^s\!\rmd T \, \sabs{T+\alpha}^{-1-\delta}
\le C \Bigl( s^{1-\gamma'/2} + \sum_{j\in B; j<j_0} r_j \Bigr)\, .
\end{align}
This estimate can be iterated for the remaining
$r$-integrations, which proves that there is a constant $C$ such that
\begin{align} 
&\int_{(\R_+)^{A'}}\!\rmd r \,  \delta\Bigl(s-\sum_{i\in A'} r_i\Bigr)
|G_2(k,r)|
\le s^{\tilde{n}+i_2-j_0-1-\gamma'/2}
C^{|B_2|+1} \, .
\end{align}
Here $\tilde{n}=
1+m'+|B_2|+1-(i_2-j_0)-|B_2|-1$,  and thus
$\tilde{n}+i_2-j_0-1=N/2$.
Therefore, (\ref{eq:pkbound}) is bounded by
\begin{align} 
C^{|B_2|+1} \lambda^{\gamma'-N}
\sabs{\lambda^2 s}^{N/2} \rme^{s \lambda^2}
\oint_{\Gamma_N} \frac{|\rmd z|}{2\pi}
\frac{1}{|z|}
\prod_{i\in A_2; i\ge i_2} \frac{1}{|z-\gamma(i)|} \, .
\end{align}
This bound can then be integrated over the remaining $k$-variables
using the iteration scheme of the
basic estimate.  This results in the following bound for this contribution
to the amplitude
\begin{align} 
C \lambda^{\gamma'}
\rme^{s\lambda^2} \sabs{c s\lambda^2}^{N/2}
\sabs{\ln \lambda}^{2+N/2} \, .
\end{align}
Comparing this with (\ref{eq:somebound}) proves the
bounds
stated in the Proposition.
\qed \end{proof}

\section{Completion of the proof of the main theorem}
\label{sec:completion}

Since for a graph with $N$ interaction vertices there are $2 N+2$ fields at
time $0$, the total number of pairing graphs is bounded by $2^{N+1} (N+1)!$.
However, there are much fewer leading graphs, at most
$2^{3 N} (N/2)!$, as the following Lemma shows.
\begin{lemma}
Consider momentum graphs with $n'$ interaction vertices in the minus tree and
$n$ in the plus tree.  If $n+n'$ is odd, none of these graphs is
leading.  If $n+n'$ is even, there are at most
\begin{align}\label{eq:nofleading}
 4^{n+n'} (n+n'-1)!! \le 8^{n+n'} \Bigl(\frac{n+n'}{2}\Bigr)!
\end{align}
leading graphs.  (We have defined $(-1)!!=1$.)
\end{lemma}
\begin{proof}
Since all leading motives contain $2$ interaction vertices,
every leading graph has even number of them.  Thus we can assume
that there is an integer $m$ such that
$2 m=n+n'$, and we need to prove that the number of leading diagrams is
then bounded by $16^{m} (2 m-1)!!$.  We make the proof by induction in $m$.
If $m=0$, then $n'=0=n$ and
the result is obviously true.  (There is only one diagram, which is leading.)
We make the induction assumption that the above is true for any graph with
$2m$ interaction vertices, $m\ge 0$.

Consider adding a leading motive to the bottom of such a graph.
There are $4 m +2$ edges at the bottom, and thus $2 m+1$ pairs.
The loss motive can be connected to left or right edge of a pair, in six
different ways.  (This is independently of the parity of the edge.)
The gain motive splits the cluster vertex of the pair, in four different ways
which respect the parities,
see Fig.~\ref{fig:leadmot}.
Thus there are altogether $2*6+4=16$ different ways of connecting the
new motive into an existing pair.
Using the induction assumption,
we find that there are at most
$16 (2 m+1)* 16^{m} (2 m-1)!!=16^{m+1} (2 m+1)!!$ ways to
make a leading diagram with $2(m+1)$ interaction vertices.
This completes the induction step.  We remark that the ``at most''-part is
necessary
since not all of these graphs have $n'$ interactions on the left.

The inequality in (\ref{eq:nofleading}) follows then from
$(2 m-1)!!\le 2^m m!$.
\qed \end{proof}

Let $c_0>0$ denote a constant for which Proposition \ref{th:leadingest}
holds.
We choose $t_0 = (2^6 c_0)^{-1}>0$, when for all $0<t<t_0$ and
$0\le s\le t\lambda^{-2}$ we have
$2^6 c_0 \lambda^2 s \le t/t_0<1$,
and
$\sum_{m=0}^\infty (2^6 c_0 t)^m$ is always summable.
\begin{corollary}
Let $t_0$ be the constant defined above, and assume $0<t<t_0$.
Then
\begin{align}\label{eq:QpairsisOK}
&\lim_{\lambda\to 0}\limsup_{\Lambda\to \infty}\left|
\Qlfin[g,f](t)-\Qpairs [g,f](t)\right| = 0\, , \\
&\lim_{\lambda\to 0}\left|
\Qpairs [g,f](t)-  \sum_{n=0; n \text{ even }}^{N_0-1}
\sum_{\text{leading graphs}}
\left.\mathcal{F}_n^{\rm pairs}(S,\ell,t\lambda^{-2},\kappa)
\right|_{\PFone \to 1}
\right| = 0\, . \label{eq:FnpairisOK}
\end{align}
\end{corollary}
\begin{proof}
To prove the first limit, we apply Proposition \ref{th:Qmainerr} and triangle
inequality, which shows that  
\begin{align}
\left|\Qlfin-\Qpairs \right|
\le \left|\Qmain-\Qpairs \right|
+ | Q^{\rm err}_{\rm pti} | +
|Q^{\rm err}_{\rm cut}| + |Q^{\rm err}_{\rm amp} |\, .
\end{align}
By Proposition \ref{th:main2nd},
the first term on the right hand side vanishes in the limit.  The third and
fourth terms also vanish by Propositions  
\ref{th:cuterr} and \ref{th:amperr}, respectively.

To study the remaining second term, we first apply Proposition
\ref{th:ptierr1}.  To derive an upper bound for the second term on the right
hand side of (\ref{eq:ptierr1}), consider arbitrary $s$ and $n$ such that  
$0\le s\le t\lambda^{-2}$ and $N_0/2\le n<N_0$.
Then $\mathcal{G}_n^{\rm pairs}(S,J,\ell,\ell',s,\kappa)$ is non-zero only if
the corresponding  
graph is fully paired.  If the graph is either crossing or nested, we can
apply Proposition \ref{th:crossingbnd} or \ref{th:nestedbnd} to bound
$|\mathcal{G}_n^{\rm pairs}|$.  There are at most $2^{2 N_0 +1} (2 N_0+1)!$
such graphs, 
and it can then be checked that the sum over these graphs decays faster than
the factor $N_0^{2+2 b_0}$ in 
(\ref{eq:ptierr1}).   All other graphs are leading, with the total number
bounded by $2^{6 n} n!$, as shown by the above Lemma.  
According to the estimate (\ref{eq:Gnleadb}) in Proposition
\ref{th:leadingest}, for any such graph 
$|\mathcal{G}_n^{\rm pairs}|$ is bounded by a sum of two terms,
a term containing a power $\lambda^{\gamma'}$ and 
$\frac{C}{n!} (c_0 \lambda^2 s)^n$.  
The sum over the first terms leads to a vanishing contribution, similarly to
what happened for crossing and nested graphs.  Since here $n\ge N_0/2$,
the sum over the second terms is bounded by $C (t/t_0)^{N_0/2}$.  
Now $(t/t_0)^{N_0/2} N_0^{2+2 b_0} \to 0$ when $\lambda\to 0$, and
we can thus conclude 
that also $\limsup_{\Lambda\to\infty} | Q^{\rm err}_{\rm pti}|\to 0$.  This
concludes the proof of   (\ref{eq:QpairsisOK}).

Propositions \ref{th:main2nd},
\ref{th:crossingbnd}, \ref{th:leadingest}, and \ref{th:nestedbnd},
combined with the above estimates for the number of fully paired and
leading graphs, can be used to prove similarly that also (\ref{eq:FnpairisOK})
holds. 
\qed \end{proof}

To complete the proof of the main theorem, the sum over the leading graphs
needs to be computed.
\begin{lemma}\label{th:sumleadingm}
For any $n=2 m$, $m\ge 0$, and with $\Gamma(k)$ defined by
(\ref{eq:defGamma}),  we have
  \begin{align}\label{eq:sumleading}
&\lim_{  \lambda \to 0 }
 \sum_{\text{leading graphs}}
  \left.\mathcal{F}_n^{\rm pairs}(S,\ell,t\lambda^{-2},\kappa)
  \right|_{\PFone \to 1}
 = \int_{\T^d}\rmd k\, \FT{g}(k)^* \FT{f}(k) W(k) \Gamma(k)^m
  \frac{(-t)^m}{m!}\, .
  \end{align}
\end{lemma}
Proposition \ref{th:leadingest} shows that each even
term in the sum over $n$ can be dominated by $c (t/t_0)^{n/2}$ which is
summable.  Thus we can move the $\lambda\to 0$ limit inside the sum over $n$.
Combined with the above results this proves that Theorem \ref{th:main}
holds.

\begin{proofof}{Lemma \ref{th:sumleadingm}}
The convergence of the leading diagrams has been discussed in detail in
\cite{ls09}, we will mainly sketch the argument here under the present
notations and assumptions.  We have already proven the result for $m=0$,
and let us thus assume $m\ge 1$.

By  (\ref{eq:defFnpairs}) and (\ref{eq:leadingphases}) here
\begin{align}\label{eq:Fnpairs2}
&\left. \mathcal{F}_{2 m}^{\rm pairs}(S,\ell,t \lambda^{-2},\kappa)
 \right|_{\PFone \to 1}
 = (-1)^m \lambda^{2 m}
\int_{\T^d} \!\rmd k_{e_1}\, \FT{g}(-k_{e_1})^* \FT{f}(-k_{e_1})
 \nonumber \\ &  \quad \times
\int_{(\T^d)^{2 m}} \!\prod_{j=1}^m (\rmd k_{m,1} \rmd k_{m,2})\,
\prod_{j=1}^{m} \Bigl[ \sigma_{2 j,\ell_{2 j}} \sigma_{2 j-1,\ell_{2 j-1}}
\Bigr]
\prod_{A=\set{i,j}\in S} W(k_{0,i})
 \nonumber \\ & \quad \times
\int_{(\R_+)^{I_{0,2m}}}\!\rmd r \,
\delta\Bigl(t\lambda^{-2}-\sum_{i=0}^{2 m}
r_i\Bigr)
\prod_{j=0}^{2 m} \rme^{-r_j \kappa_{2 m-j}}
\prod_{j=1}^{m} \rme^{-\ci r_{2 j-1} \Omega_{2 j}}
 \, .
\end{align}
We next assume that $\lambda$ is so small that $N_0(\lambda)> 4m $.
Then $\kappa_{2 m-j}=0$ for all $j$, and only the $\delta$-function depends
on $r_i$, for $i$ even.  We change integration variables from $r$
to $(t,s)$ with
$t_i=\lambda^2 r_{2 i}$, for $i=0,1,\ldots,m$, and $s_i=r_{2 i-1}$ for
$i=1,\ldots,m$.  Then the last line of (\ref{eq:Fnpairs2}) is equal to
\begin{align} 
& \lambda^{-2 m} \int_{(\R_+)^{I_{m}}}\!\rmd s
\prod_{j=1}^{m} \rme^{-\ci s_{j} \Omega_{2 j}}
 \int_{(\R_+)^{I_{0,m}}}\!\rmd t \,
\delta\Bigl(t-\sum_{i=0}^{m} t_i - \lambda^{2}\sum_{i=1}^{m} s_i\Bigr)
 \nonumber \\ & \quad
= \lambda^{-2 m} \int_{(\R_+)^{I_{m}}}\!\rmd s
\prod_{j=1}^{m} \rme^{-\ci s_{j} \Omega_{2 j}}
\1\Bigl(\sum_{i=1}^{m} s_i \le t \lambda^{-2}\Bigr)
\frac{1}{m!} \Bigl(t-\lambda^{2}\sum_{i=1}^{m} s_i\Bigr)^m\, .
\end{align}

Next we integrate over the double-loop momenta.  This leads to iterated
applications of $G_{s_i,\tau_i}$ which yields a
function $\tilde{G}(-k_{e_1},s;S,\ell)$.
This function also has
an $s$-integrable upper bound which is independent of $k$ and
$\lambda$, and thus we can apply dominated convergence to
take the $\lambda\to 0$ limit inside the $s$-integration.
This proves that for a leading graph
\begin{align}\label{eq:sumleading2}
&\lim_{  \lambda \to 0 }
  \left.\mathcal{F}_n^{\rm pairs}(S,\ell,t\lambda^{-2},\kappa)
  \right|_{\PFone \to 1}
 = (-1)^m \frac{t^m}{m!}
\int_{\T^d}\rmd k\, \FT{g}(k)^* \FT{f}(k)
\int_{(\R_+)^{I_{m}}}\!\rmd s \, \tilde{G}(k,s;S,\ell) \, .
\end{align}

We now sum over the
leading graphs with $m$ motives, which is a finite sum and thus can be
taken directly of $\tilde{G}(k,s;S,\ell)$.
Every leading diagram in the sum in (\ref{eq:sumleading})
is obtained by iteratively adding $m$ leading motives
to the graph formed by using a single pairing cluster.  It is then easy to see
that there is a one-to-one correspondence
between leading graphs with $m$
motives and no interaction vertices in the minus tree,
and graphs which are obtained by adding
a leading motive to a pairing cluster of such a graph with $m-1$ motives
so that the result does not contain any interaction vertices in
the minus tree.

Consider thus a graph with $m-1$ motives which could give rise to a
leading diagram in the sum in (\ref{eq:sumleading}).
This graph has $2 m -1$ pairing clusters,
exactly one of which connects to the minus tree.
If we add the leading motive to this special pairing,
then the resulting graph should not contain any interaction vertices in
the minus tree.  This rules out all gain motives and six of the loss motives,
those added to the left leg of the pairing cluster.  Thus we get only six
new terms from such an addition, and these correspond to addition of a loss
motive to the right leg of the pairing cluster.

There are no such restrictions to  motives added to any of the remaining
$2m-2$ pairings since all the new interaction vertices then belong to the
plus tree.  Fix the graph and such a pairing, and let
$\sigma$ denote the parity on the left leg of the pairing.
We sum over all graphs obtained by adding a leading motive to this pairing.
Then the iteration steps to obtain the corresponding $\tilde{G}(k,s;S,\ell)$
are equal apart from the first step which is of the form
$\pm G_{s_1,\tau}(k_0)$.  Computing the sum over all possibilities in the first
iteration yields
a contribution
\begin{align}\label{eq:alllead}
& \int_{(\T^d)^2}\!\! \rmd k_1 \rmd k_2\, \Bigl[
\rme^{-\ci s_1 \Omega(k,\sigma)}
\left( -2 W_1 W_2 W_3+ 2 \sigma W_0 W_1 W_2
+ 2 W_0 W_1 W_3- 2 \sigma W_0 W_2 W_3\right)
\nonumber \\ & \quad
+\rme^{-\ci s_1 \Omega(k,-\sigma)}
\left( -2 W_1 W_2 W_3- 2 \sigma W_0 W_1 W_2
+ 2 W_0 W_1 W_3+ 2 \sigma W_0 W_2 W_3\right)
\Bigr]\, ,
\end{align}
where $W_i=W(k_i)$ and $k_3=k_0-k_1-k_2$.  Since always we have
$\Omega((k_1,k_2,k_3),-\sigma)=-\Omega((k_3,k_2,k_1),\sigma)$,
we can make a change of integration
variables $k_1\to k_3$ in the second term, which shows that
(\ref{eq:alllead}) is equal to
\begin{align}
& 2 \int_{(\T^d)^2}\!\! \rmd k_1 \rmd k_2\, \Bigl[
\rme^{-\ci s_1 \Omega(k,\sigma)} +\rme^{\ci s_1 \Omega(k,\sigma)}
\Bigr] W_0 W_1 W_2 W_3
\nonumber \\ & \qquad \times
\left( -W_0^{-1}+ \sigma W_3^{-1} + W_2^{-1} - \sigma W_1^{-1}\right)
\, .
\end{align}
Since $W(k)^{-1} = \beta(\omega(k)-\mu)$, the final factor in parenthesis
is equal to $\beta \sigma \Omega(k,\sigma)$.  Thus
integrating 
(\ref{eq:alllead}) over $s_1\in [0,M]$ for some $M>0$ yields
\begin{align} 
& 2 \beta\sigma
  \int_{-M}^M\! \rmd s_1 \int_{(\T^d)^2}\!\! \rmd k_1 \rmd k_2\,
 \Omega(k,\sigma) \rme^{-\ci s_1 \Omega(k,\sigma)}
W_0 W_1 W_2 W_3
\nonumber \\ & \quad
= 2 \beta\sigma
 \int_{(\T^d)^2}\!\! \rmd k_1 \rmd k_2\,
 \ci \Bigl[
\rme^{-\ci M \Omega(k,\sigma)} +\rme^{\ci M \Omega(k,\sigma)}
\Bigr]
W_0 W_1 W_2 W_3
\, ,
\end{align}
which vanishes as $M\to \infty$.  Therefore, the sum over such graphs in
(\ref{eq:sumleading2}) is
exactly zero, even though the individual terms can be non-vanishing.
(The vanishing is not accidental, as a comparison with the discussion in
Section \ref{sec:link} reveals:  
$W$ is a stationary solution of
the corresponding nonlinear Boltzmann equation, (\ref{eq:4}), and the above
sum corresponds to the action of its collision operator $\mathcal{C}$ to $W$,
which thus should be equal to zero.)

Therefore, in the sum over the relevant leading diagrams, only those terms
can be non-zero which come from the application of a loss term to the right
leg of the unique pairing which connects to the minus tree.  Since this
only changes the multiplicative factor associated with the pairing,
we can iterate the above argument and conclude that the sum must be equal
to $W(k) \Gamma(k)^m$, where $\Gamma(k)$ is the result from the sum over
the six relevant loss terms.  As above, this can be computed explicitly, and
it is seen to be equal to $\Gamma(k)$ defined in (\ref{eq:defGamma})
after a change of variables and using the evenness of the functions
$\omega$ and $W$.
Collecting all the results together yields
(\ref{eq:sumleading}).\qed
\end{proofof}

\appendix

\section{Nearest neighbor interactions}\label{sec:appNN}

Let us consider here the dispersion relation $\omega$ defined by
\begin{align} 
  \omega(k):= c - \sum_{\nu=1}^d \cos p^\nu, \quad \text{with }p=2\pi k,
\end{align}
where $c\in \R$ is arbitrary.
This clearly satisfies (DR1).
We consider the function $K$ defined in (\ref{eq:defp2tx}).
Then
\begin{align} 
 &   K(x;t_0,t_1,t_2,\tfrac{1}{2\pi} q_1, \tfrac{1}{2\pi} q_2)
=   \rme^{-\ci c (t_0+t_1+t_2)}
    \prod_{\nu=1}^d
    \int_{0}^{2\pi} \!\frac{\rmd p}{2\pi}\, \rme^{\ci p x^\nu}
    \rme^{\ci (t_0 \cos p +t_1 \cos (p +q_1^\nu) + t_2 \cos (p +q_2^\nu))} \, .
\end{align}
Since
\begin{align} 
& t_0 \cos p +t_1 \cos (p +q_1^\nu) + t_2 \cos (p +q_2^\nu)
= \re\!\left[ \rme^{\ci p}
  ( t_0 + t_1 \rme^{\ci q_1^\nu} + t_2 \rme^{\ci q_2^\nu} ) \right]\, ,
\end{align}
there is $\varphi^\nu$, which does not depend on $p$,
such that this is equal to
\begin{align} 
  R^\nu \cos(p+\varphi^\nu),\quad \text{with}\quad
 R^\nu = | t_0 + t_1 \rme^{\ci q_1^\nu} + t_2 \rme^{\ci q_2^\nu} |\, .
\end{align}
This proves that
\begin{align} 
| K(x;t_0,t_1,t_2,\tfrac{1}{2\pi} q_1, \tfrac{1}{2\pi} q_2)|
= \prod_{\nu=1}^d \left|
    \int_{0}^{2\pi} \!\frac{\rmd p}{2\pi}\,
    \rme^{\ci p x^\nu+\ci R^\nu \cos p}\right|  \, ,
\end{align}
and, therefore,
\begin{align} 
 \norm{K(t_0,t_1,t_2,\tfrac{1}{2\pi} q_1, \tfrac{1}{2\pi} q_2)}_3
\le \prod_{\nu=1}^d \norm{p^{(d=1)}_{R^\nu}}_3
\le C \prod_{\nu=1}^d \frac{1}{\sabs{R^\nu}^{\frac{1}{7}}},
\end{align}
where we have applied a known bound for the $\ell_3$-norm of 
the one-dimensional propagator, following \cite{HoLa97,LaLu01}.

We note that $p_t(x)=K(x;t,0,0,0,0)$ and thus the above bound shows that
\begin{align}\label{eq:pt33bound}
\norm{p_t}_3^3 \le C \sabs{t}^{-\frac{3 d}{7}}
 \le C \sabs{t}^{-1-\frac{2}{7}}
\end{align}
for all $d\ge 3$.  Thus (DR2) is satisfied then.

On the other hand,
\begin{align} 
  \Bigl| \int_{\T^d}\!\rmd k\, \rme^{-\ci t (\omega(k)+\sigma
    \omega(k-k_0))}\Bigr|
 = |K(0;t,\sigma t,0,-k_0,0)|
 \le C \prod_{\nu=1}^d \frac{1}{\sabs{R^\nu}^{\frac{1}{2}}}, \,
\end{align}
since $\cos p$ is a Morse function.
Here $R^\nu = |t|\, | 1 + \sigma \rme^{-\ci 2 \pi k_0^\nu}|
\ge |t| \,|\sin (2 \pi k_0^\nu )|$.
Thus, given $k_0$, let $\nu_0$ denote the index corresponding to
the {\em second\/} largest of the numbers $|\sin (2 \pi k_0^\nu )|$.
(This might not be unique, but this is irrelevant for the following
estimates.)
Then we have
\begin{align} 
  \Bigl| \int_{\T^d}\!\rmd k\, \rme^{-\ci t (\omega(k)+\sigma
    \omega(k-k_0))}\Bigr|
 \le C \sabs{t}^{-1} \frac{1}{|\sin (2\pi k_0^{\nu_0})|}
 \le C \sabs{t}^{-1} \frac{1}{d(k_0^{\nu_0},\set{0,\frac{1}{2}})} \, .
\end{align}
Thus (DR3) holds if $d\ge 3$ and we choose
$\Msing$ to consist of those $k$
for which all but one component belong to the set $\set{0,\frac{1}{2}}$.
(This set is clearly a union of lines.)

Therefore, we only need to check (DR4).  Let us first consider
(\ref{eq:crossingest1b}), where
$t_0=t$, $t_1=\pm s$, $t_3 = 0$, $q_2=0$, and we have some fixed $n\in
\set{1,2,3}$.
Then
$R^\nu = | t\pm s \rme^{\ci q_1^\nu}|$
which, by the triangle inequality, has a lower bound
$| |t| -|s| |$.  On the other hand, by inspecting only the imaginary
part, we find that it also has a lower bound
$|s| |\sin q_1^\nu|$.
We use the second bound in the $n$:th factor in
(\ref{eq:crossingest1b}) and the first bound in the remaining two factors.
This shows that
the left hand side of (\ref{eq:crossingest1b})
can be bounded by
\begin{align}
& C  \int_{\R^2}\!\! \rmd t\rmd s\, \rme^{-\beta |s|}
\sabs{ |t|-|s| }^{-\frac{8}{7}} |s|^{-\frac{4}{7}}
 \prod_{\nu=1}^4 \frac{1}{|\sin (2 \pi u_n^\nu)|^{\frac{1}{7}}}
\le C \beta^{\frac{4}{7}-1}
\prod_{\nu=1}^d \frac{1}{|\sin (2 \pi u_n^\nu)|^{\frac{1}{7}}}
 \, ,
\end{align}
where we have assumed $d\ge 4$.

In the other inequality (\ref{eq:crossingest1c}), we need to consider
$t_0=t$, $t_1=\pm s$, $t_3 = \pm s$.
We apply the previous estimates which shows that
the left hand side of (\ref{eq:crossingest1c}) is bounded by
\begin{align} 
& C  \int_{\R}\! \rmd t\, \sabs{t}^{-\frac{8}{7}}
\int_{\R}\!\rmd s\, \rme^{-\beta |s|}
 \prod_{\nu=1}^4 \frac{1}{\sabs{R^\nu}^{\frac{1}{7}}}
\end{align}
where
\begin{align} 
 R^\nu = | t \pm  s (\rme^{\ci q_1^\nu} \pm \rme^{\ci q_2^\nu})|
\ge \left|\, |s|\, |1\pm \rme^{\ci (q_2^\nu-q_1^\nu)}|-|t|\,\right| \,.
\end{align}
Since here $|1\pm \rme^{\ci (q_2^\nu-q_1^\nu)}|\ge |\sin (q_2^\nu-q_1^\nu)|$
and $q_2-q_1=2\pi(u_2-u_1)$, this shows that
(\ref{eq:crossingest1c}) can be bounded by
$C \beta^{\frac{4}{7}-1}
\prod_{\nu=1}^d |\sin (2 \pi (u_2-u_1)^\nu)|^{-\frac{1}{7}}$.

Thus we now only need to check that the second item in (DR4) holds for
$\Fbcr$ defined by (\ref{eq:defnnFcr}).  Since $\Fbcr$ is independent of
$\beta$ and obviously belongs to $L^1(\T^d)$,
(\ref{eq:crossingest2a}) holds with $c_2=0$.
For the second integral we need to estimate
\begin{align}\label{eq:DR4nn2est}
\left| \int_{(\T^d)^2} \rmd k_1\rmd k_2\, \Fbcr(k_1+u;\beta)
\rme^{-\ci s (\sigma_1 \omega(k_1)+ \sigma_2 \omega(k_2)+ \sigma_3
  \omega(k_1+k_2-k_0))}\right|
\end{align}
for any choice of the signs $\sigma\in \set{\pm 1}^3$. Since
\begin{align} 
  \cos(p_2) \pm \cos(p_1+p_2-p_0) =
 \cos(p_2+\varphi) | 1\pm \rme^{\ci(p_1^\nu-p_0^\nu)}|
\end{align}
for some $\varphi$ independent of $p_2$, we can bound (\ref{eq:DR4nn2est}) by
\begin{align} 
 C \sabs{s}^{-\frac{d}{2}} \prod_{\nu=1}^d \left[ \int_0^{2\pi} \rmd p_1^\nu\,
 |\sin(p_1^\nu +2\pi u^\nu)|^{-\frac{1}{7}}
 |\sin(p_1^\nu - 2\pi k_0^\nu)|^{-\frac{1}{2}}
 \right] \, .
\end{align}
Then an application of H\"{o}lder's inequality with the conjugate pair
$(3,\frac{3}{2})$, reveals that the remaining integral is bounded uniformly in
$u$ and $k_0$.
Thus (\ref{eq:DR4nn2est}) is uniformly bounded by $C\sabs{s}^{-2}$ which is
integrable in $s$.  Then an application of (\ref{eq:restophases})
shows that (\ref{eq:crossingest2}) holds for $n=1$
with $c_2=0$.  This implies that the same bound is valid also for $n=2$ and
$n=3$ by a simple change of integration variables.
This completes the proof that the nearest neighbor dispersion relation
satisfies Assumption \ref{th:disprelass} for any $d\ge 4$.

\end{document}